\documentclass[article,reprint,amsfonts, amssymb, amsmath,  showkeys,aps,pra, superscriptaddress, twocolumn,longbibliography]{revtex4-2}
\usepackage{xcolor}

\usepackage{wrapfig}

\newcommand{\e}{\ensuremath{\mathrm{e}}}

\providecommand{\myvec}[1]{\ensuremath{\boldsymbol{#1}}}

\newcommand{\br}{\boldsymbol{r}}
\newcommand{\bs}{\boldsymbol{s}}

\newcommand{\bd}{\boldsymbol{d}}

\usepackage{minitoc}
\usepackage{tocloft}
\renewcommand \partname{}
\usepackage{mathdots}
\renewcommand\dots{\makebox[1em][c]{.\hfil.\hfil.}\thinspace}

\usepackage{algpseudocode,algorithm}
\usepackage{algorithm}
\usepackage{algpseudocode}

\providecommand{\ee}{\ensuremath{\myvec{e}}}

\providecommand{\ll}{\ensuremath{\myvec{l}}}

\providecommand{\calA}{\ensuremath{\mathcal{A}}}
\providecommand{\calB}{\ensuremath{\mathcal{B}}}
\providecommand{\calC}{\ensuremath{\mathcal{C}}}

\providecommand{\calG}{\ensuremath{\mathcal{G}}}

\providecommand{\calL}{\ensuremath{\mathcal{L}}}

\providecommand{\calN}{\ensuremath{\mathcal{N}}}
\providecommand{\calO}{\ensuremath{\mathcal{O}}}

\providecommand{\calS}{\ensuremath{\mathcal{S}}}

\providecommand{\calU}{\ensuremath{\mathcal{U}}}

\providecommand{\calX}{\ensuremath{\mathcal{X}}}

\providecommand{\bbE}{\ensuremath{\mathbb{E}}}

\providecommand{\ee}{\ensuremath{\myvec{e}}}

\providecommand{\ll}{\ensuremath{\myvec{l}}}

\providecommand{\calA}{\ensuremath{\mathcal{A}}}
\providecommand{\calB}{\ensuremath{\mathcal{B}}}
\providecommand{\calC}{\ensuremath{\mathcal{C}}}

\providecommand{\calG}{\ensuremath{\mathcal{G}}}

\providecommand{\calL}{\ensuremath{\mathcal{L}}}

\providecommand{\calN}{\ensuremath{\mathcal{N}}}
\providecommand{\N}{\ensuremath{\mathbb{N}}}
\providecommand{\calO}{\ensuremath{\mathcal{O}}}

\providecommand{\calS}{\ensuremath{\mathcal{S}}}

\providecommand{\calU}{\ensuremath{\mathcal{U}}}

\providecommand{\calX}{\ensuremath{\mathcal{X}}}

\providecommand{\bbE}{\ensuremath{\mathbb{E}}}

\def\01{\{0,1\}}

\makeatletter
\let\oldabs\abs
\def\abs{\@ifstar{\oldabs}{\oldabs*}}
\let\oldnorm\norm
\def\norm{\@ifstar{\oldnorm}{\oldnorm*}}
\makeatother

\usepackage{float}
\makeatletter
\let\newfloat\newfloat@ltx
\makeatother

\makeatletter
\newsavebox{\@brx}
\newcommand{\llangle}[1][]{\savebox{\@brx}{\(\m@th{#1\langle}\)}%
  \mathopen{\copy\@brx\kern-0.5\wd\@brx\usebox{\@brx}}}
\newcommand{\rrangle}[1][]{\savebox{\@brx}{\(\m@th{#1\rangle}\)}%
  \mathclose{\copy\@brx\kern-0.5\wd\@brx\usebox{\@brx}}}
\makeatother
\usepackage[english]{babel}
\usepackage[utf8]{inputenc}
\usepackage{graphics}
\usepackage{selinput}
\usepackage[normalem]{ulem}
\usepackage[shortlabels]{enumitem}

\usepackage{braket}
\usepackage{amsthm}
\usepackage{mathtools}
\usepackage{physics}
\usepackage{graphicx}
\usepackage[left=16mm,right=16mm,top=35mm,columnsep=15pt]{geometry} 
\usepackage{adjustbox}
\usepackage{placeins}
\usepackage[T1]{fontenc}
\usepackage{lipsum}
\usepackage{csquotes}
\usepackage{bm}

\newcommand{\fsnull}[1]{}
\newcommand{\old}[1]{}

\usepackage[makeroom]{cancel}
\usepackage[toc,page]{appendix}
\usepackage[colorlinks=true,citecolor=blue,linkcolor=magenta]{hyperref}
\usepackage[capitalize]{cleveref}

\usepackage{tikz}
\tikzset{every picture/.style=remember picture}

\usepackage[utf8]{inputenc}
\usepackage{graphicx}
\usepackage{xcolor}
\usepackage{amsmath}
\usepackage{amsthm}
\usepackage{bm}
\usepackage{bbm}
\usepackage{comment}
\usepackage{appendix}
\usepackage{mathdots}
\usepackage{lipsum}
\usepackage{verbatim}
\usepackage{natbib}
\usepackage{nccmath}
\usepackage{amsfonts} 

\providecommand{\calU}{\ensuremath{\mathcal{U}}}

\providecommand{\calS}{\ensuremath{\mathcal{S}}}
\newcommand{\fnormone}[2][]{ {#1| \! #1|} #2  {#1| \! #1|_{\mathbb{F},1}}}
\renewcommand{\norm}[2][]{ {#1| \! #1|} #2  {#1| \! #1|}}

\newcommand{\poly}{\operatorname{poly}}

\DeclareMathOperator{\I}{\mathbb{I}}
\newcommand{\C}{\mathbb{C}}
\DeclareMathOperator{\R}{\mathbb{R}}

\renewcommand{\geq}{\geqslant}
\renewcommand{\leq}{\leqslant}

\renewcommand{\Re}{\text{Re}}
\renewcommand{\Im}{\text{Im}}

\def\calA{\mathcal{A}}
\def\calB{\mathcal{B}}
\def\calC{\mathcal{C}}

\def\calG{\mathcal{G}}

\def\calL{\mathcal{L}}

\def\calN{\mathcal{N}}
\def\calO{\mathcal{O}}

\def\C{\mathbb{C}}

\def\be{\begin{equation}}
\def\ee{\end{equation}}
\def\bs{\begin{split}}
\def\e{\end{split}}
\def\ba{\begin{eqnarray}}
\def\bea{\begin{eqnarray}}

\def\tea{\end{eqnarray}}
\def\ea{\end{eqnarray}}
\def\eea{\end{eqnarray}}

\def\R{\mathds{R}}

\def\bbE{\mathbb{E}}

\newtheorem{maintheorem}{Theorem}

\newtheorem{mainproposition}{Proposition}

\newtheorem{maindefinition}{Definition}

\newtheorem{theorem}{Theorem}[section]
\newtheorem{lemma}{Lemma}[section]

\newtheorem{definition}{Definition}[section]

\usepackage{amssymb}
\usepackage{dsfont}

\usepackage{yhmath}

\def\be{\begin{equation}}
\def\te{\end{equation}}
\def\ee{\end{equation}}
\def\ba{\begin{eqnarray}}
\def\bea{\begin{eqnarray}}

\def\tea{\end{eqnarray}}
\def\ea{\end{eqnarray}}
\def\eea{\end{eqnarray}}

\usepackage{minitoc}
\usepackage{tocloft}

\renewcommand \partname{}
\begin{document}

\makeatletter
\makeatother
 
\doparttoc 
\faketableofcontents 


\title{When quantum resources backfire: \\ Non-gaussianity and symplectic coherence in noisy bosonic circuits}


\author{Varun Upreti}
\email{varun.upreti@inria.fr}
\affiliation{DIENS, \'Ecole Normale Sup\'erieure, PSL University, CNRS, INRIA, 45 rue d’Ulm, Paris, 75005, France}

\author{Ulysse Chabaud}
\affiliation{DIENS, \'Ecole Normale Sup\'erieure, PSL University, CNRS, INRIA, 45 rue d’Ulm, Paris, 75005, France}

\author{Zo\"{e} Holmes}
\affiliation{Institute of Physics, \'Ecole Polytechnique F\'{e}d\'{e}rale de Lausanne (EPFL),   Lausanne, Switzerland}
\affiliation{Centre for Quantum Science and Engineering, \'Ecole Polytechnique F\'{e}d\'{e}rale de Lausanne (EPFL),   Lausanne, Switzerland}

\author{Armando Angrisani}
\email{armando.angrisani@epfl.ch}
\affiliation{Institute of Physics, \'Ecole Polytechnique F\'{e}d\'{e}rale de Lausanne (EPFL),   Lausanne, Switzerland}
\affiliation{Centre for Quantum Science and Engineering, \'Ecole Polytechnique F\'{e}d\'{e}rale de Lausanne (EPFL),   Lausanne, Switzerland}

\begin{abstract}

Analyzing the impact of noise is of fundamental importance to understand the advantages provided by quantum systems. While the classical simulability of noisy discrete-variable systems is increasingly well understood, noisy bosonic circuits are more challenging to simulate and analyze. 
Here, we address this gap by introducing the \textit{displacement propagation} algorithm, a continuous-variable analogue of Pauli propagation for simulating noisy bosonic circuits. 
By exploring the interplay of noise and quantum resources, we identify several computational phase transitions, revealing regimes where even modest noise levels render bosonic circuits efficiently classically simulable.
In particular, our analysis reveals a surprising phenomenon: computational resources usually associated with bosonic quantum advantage, namely non-Gaussianity and symplectic coherence, can make the system easier to classically simulate in presence of noise.

\end{abstract}

\maketitle

\section{Introduction}

In recent years, a large class of quantum systems known as bosonic systems have emerged as promising candidates for building quantum computers \cite{Gottesman2001,Knill2001,menicucci2006universal,madsen2022quantum,liu2024hybrid}, driven by experimental breakthroughs such as the generation of complex quantum states \cite{yokoyama2013ultra,larsen2025integrated} and remarkable error-correction capabilities \cite{michael2016new, terhal2020towards, cai2021bosonic, sivak2023real, putterman2024hardware, albert2025bosonic}. These ubiquitous systems include photonics, superconducting resonators, or motional modes of trapped ions and neutral atoms, and are often referred to as continuous-variable (CV) systems, governed by harmonic oscillator-like, infinite-dimensional degrees of freedom.

In practice, however, noise and imperfections inevitably affect the capability of quantum computers to both perform non-trivial computations and to outperform their classical counterparts.
In the presence of such sources of noise, quantum computations that are classically intractable in an ideal, noiseless setting may become trivial (collapsing toward an input-independent fixed point) or otherwise 
easy to simulate on a classical computer. The standard sources of noise for bosonic systems are (photon) loss and thermal noise. These differ in mathematical structure from standard qubit noise sources, such as depolarizing noise and dephasing noise, in discrete-variable (DV) systems. 

The classical simulability of noisy DV systems is increasingly well understood with a series of results establishing efficient classical algorithms for simulating systems in the average case for arbitrary noise levels and models~\cite{aharonov2022polynomial,schuster2024polynomial,fontana2023,angrisani2024classically} and in the worst case for sufficiently high levels of depolarizing noise~\cite{rall2019simulation, gonzalez2024pauli}. These results rely on the \textit{Pauli propagation} framework~\cite{beguvsic2023fast,rudolph2023classical,shao2023simulating,angrisani2024classically, beguvsic2024real,lerch2024efficient, bermejo2024improving, rudolph2025pauli, gonzalez2024pauli, schuster2024polynomial, aharonov2022polynomial, fontana2023} whereby a quantum observable is propagated through a sequence of circuit layers. At each step, the observable is updated according to the action of the corresponding gate or noise channel, with tailored approximation schemes ensuring the process remains tractable on a classical computer. 

The simulability of bosonic circuits, both with and without noise, is more challenging than their DV counterparts. For example, the naive approach of propagating observables as polynomials of quadratures scales \emph{doubly exponentially} with the number of non-Gaussian gates\ \cite{upreti2025interplay}, whereas analogous brute-force methods for qubit circuits scale only exponentially with the number of non-Clifford gates. As a result, even relatively shallow bosonic circuits can quickly become intractable for brute-force simulation. More advanced classical simulation algorithms for noisy bosonic computations have instead mostly relied on phase-space simulation methods \cite{veitch2012negative,Mari2012,pashayan2015estimating,keshari2016,Chabaud2024phasespace,frigerio2024}. These methods reduce quantum dynamics to stochastic classical-like trajectories in phase space.
Other classical simulation algorithms have focused specifically on noisy linear-optical circuits \cite{lund2017exact,renema2018,GarciaPatron2019,renema2020,qi2020,oh2021classical,liu2023simulating,oh2024classical,oh2025classical}. 
However, existing classical simulation algorithms do not identify the noise regimes in which universal bosonic computations are classically simulable for typical noise models.

In this paper, we introduce the \textit{displacement propagation} algorithm for simulating noisy bosonic circuits. This approach uses classical samples from characteristic functions to estimate observable expectation values for noisy bosonic computations. In particular, 
the displacement propagation formalism can be thought of as (i) a CV analogue of the Pauli propagation formalism \cite{aharonov2022polynomial,schuster2024polynomial,fontana2023,angrisani2024classically} for DV systems and (ii) a Fourier space version of phase-space simulation algorithms \cite{veitch2012negative,Mari2012,pashayan2015estimating}. We show that for the task of estimating local observables in bosonic quantum circuits based on Gaussian and noisy cubic-phase gates, our algorithm scales only exponentially with the number of noisy cubic-phase gates,  significantly improving upon the brute-force doubly exponential scaling.

Moreover, we use displacement propagation to explore how the computational cost of simulating noisy bosonic circuits depends on the noise rate and on the quantum resources supplied by the unitary gates, revealing several regimes where classical simulation is efficient.  

In {noiseless} bosonic circuits, large amounts of quantum resources such as non-Gaussianity~\cite{walschaers2021non} and symplectic coherence~\cite{upreti2025interplay} increase the runtime of classical simulation algorithms. 
In the {noisy} case, we identify regimes where sufficiently small non-Gaussianity leads to efficient classical simulation, in line with previous results \cite{pashayan2015estimating,chabaud2020classical,chabaud2023resources}.
However, we also find that large values of non-Gaussianity and symplectic coherence can render noisy simulations classically easier. In fact, for sufficiently high values of these resources, even a small amount of noise drives bosonic computations into a regime where the expectation values of a large class of both local and global observables are efficiently simulable. 
Our work thus provides a limit on the allowed noise levels, for different circuit classes and noise models, beyond which bosonic computations become trivial.

The rest of the paper is structured as follows: In section \ref{sec:prelims}, we detail the preliminary material and notations. Section \ref{sec:uniform} shows how a uniform layer of noise renders universal bosonic computations trivial. In section \ref{sec:gate_based}, we consider a weaker noise model with noisy single-mode gates and identify phases of the bosonic system where the output expectation values can either be guessed or efficiently classically simulated. We identify how the interplay of noise and the multiple computational resources characterizing a bosonic system leads to classicality in the bosonic system. Section \ref{sec:displacement_prop} explains the intuition behind the new simulation algorithm introduced in this work, which we term ``displacement propagation''. Finally, section \ref{sec:concl} summarizes the work and gives an outlook for future directions.


\section{Preliminaries}
\label{sec:prelims}
This section summarizes the notation and preliminary concepts used throughout the main text. Detailed mathematical background relevant to the simulation algorithm is provided in Sections \ref{appsec:math_prelims} and \ref{appsec:bosonic_prelims} of the Supplementary Material.

In CV quantum information, a mode refers to a degree of freedom associated with a specific quantum field of a CV quantum system, such as a single spatial or frequency mode of light, and is the equivalent of a qubit in the CV regime. In this paper, $m$ denotes the number of modes in the system. The $m$-mode vacuum state (the state with zero bosons) is denoted by $\ket 0^{\otimes m}$. The position and momentum quadratures of mode $i$ are denoted by $\hat q_i$ and $\hat p_i$, $\forall i \in \{1,\dots,m\}$, with the position and momentum quadratures of the same mode satisfying the commutation rule $[\hat q_i, \hat p_i] = 2i \mathbb{I}$ with the convention $\hbar = 2$.  The energy operator is denoted by
\begin{equation}
    \hat N = \sum_{j=1}^m \hat N_j = \sum_{j=1}^m \frac{\hat q_j^2 + \hat p_j^2}{2}.
\end{equation}

We provide detailed preliminaries on Gaussian unitaries in Section \ref{appsec:bosonic_prelims} of the Supplementary Material. Here, we note that the action of an $m$-mode Gaussian unitary operation $\hat{G}$ on the vector of quadratures $\boldsymbol{\Gamma}= [\hat{q}_1,\dots,\hat{q}_m,\hat{p}_1,\dots,\hat{p}_m]$ is given by
\begin{equation}
    \hat{G}^\dagger \boldsymbol{\Gamma} \hat{G} = S \boldsymbol{\Gamma} + \boldsymbol{d},
\end{equation}
where $S$ is a $2m \times 2m$ symplectic matrix and $\boldsymbol{d} \in \mathbb{R}^{2m}$ is a displacement vector associated to the Gaussian unitary gate $\hat G$. Single-mode displacement operators are defined as $\hat{D}(\alpha) = e^{\alpha \hat{a}^\dagger - \alpha^* \hat{a}}$, where $\alpha \in \C$. Coherent states $\ket \alpha$ are generated by the action of displacement operator on the vacuum state $\ket \alpha = \hat D(\alpha) \ket 0$. 

Non-Gaussian operations are necessary for enabling quantum advantage since Gaussian gates acting on Gaussian states and followed by Gaussian measurements can be classically simulated efficiently \cite{Bartlett2002}. One prominent example of a non-Gaussian gate is the cubic phase gate $e^{i\gamma\hat{q}^3}$ \cite{Gottesman2001}, and a standard model of universal CV quantum computation (CVQC) is defined by a vacuum state input into a circuit with Gaussian unitaries and cubic phase gates \cite{lloyd1999quantum,Sefi2011,arzani2025}. While describing the circuit model in this work, we represent the Gaussian and cubic phase gates by the quantum channels $\cal G$ and $\ \calC$, such that 
\begin{eqnarray}
    \calG(\cdot) &=&  G (\cdot) G^\dagger, \nonumber \\
    \calC(\cdot) &=& e^{i\gamma \hat q^3} (\cdot) e^{-i\gamma\hat q^3}.
\end{eqnarray}

Recent work showed that the ability of Gaussian gates to mix position and momentum quadratures, termed \textit{symplectic coherence}, is necessary for quantum advantage, alongside non-Gaussianity and entanglement \cite{upreti2025interplay}. Informally, symplectic coherence in a Gaussian gate amounts to the presence of a non-zero off-block diagonal matrix in the symplectic matrix $S$ characterizing the Gaussian gate, when $S$ is written in the basis $[\hat q_1,\dots,\hat q_m,\hat p_1,\dots,\hat p_m]$. In this work, we also mention ``symplectic coherence with respect to the first mode'' to refer to the ability of a Gaussian unitary gate to mix the position and momentum quadrature of the first mode.

Thermal loss arises from the interaction of a system with the thermal state of a finite-temperature environment. We model its effect using a quantum channel $\Lambda_{\bar n,\eta}(\cdot)$. The exact action of this channel on a quantum state $\rho$ is detailed in Section \ref{appsec:bosonic_prelims} of the Supplementary Material. 

Finally, displacement operators form an operator basis, i.e.\ for any operator $\hat{O}$ living in the $m$-mode Hilbert space $\mathcal{H}^{\otimes m}$, we can write
\begin{equation}
    \hat{O} = \frac{1}{\pi^m} \int_{\bm\alpha \in \C^m}\mathrm d^{2m}\bm\alpha \chi_{\hat O} (\bm\alpha) \hat{D}^\dagger(\bm\alpha),
\end{equation}
where $\hat{D}(\bm\alpha) = \hat{D}(\bm\alpha_1)\otimes\dots\otimes\hat D(\bm\alpha_m)$ is the multi-mode displacement operator, and $\chi_{\hat O} (\bm\alpha)$ is the characteristic function of the operator $\hat O$, and is given by
\begin{equation}
    \chi_{\hat O} (\bm\alpha) = \Tr[\rho \hat{D} (\bm\alpha)].
\end{equation}
Further, we can write $ \Tr[\rho\, \hat{O}]$ in terms of the characteristic function as
\begin{equation}
    \Tr[\rho\, \hat{O}] = \frac{1}{\pi^m}\int_{\bm\alpha \in \C^m} \mathrm d^{2m}\bm\alpha \chi_\rho^*(\bm\alpha)\chi_{\hat O}(\bm\alpha).
\end{equation}
This relation is an optical equivalence theorem for characteristic functions and forms the basis for our simulation algorithms, as explained in section \ref{sec:displacement_prop}. We also define the Fourier 1-norm of an operator as
\begin{equation}
    \fnormone{\hat O} \coloneqq \frac{1}{\pi^m} \int_{\mathbb{R}^{2m}}\mathrm d^{2m}\br |\chi_{\hat O}(\br)|.
\end{equation}
Having detailed the necessary prerequisites, we now move to a warm-up example on the effect of noise in universal bosonic computations.
\section{Warm up: overlap estimation under uniform thermal loss}
\label{sec:uniform}

In this section, we start with a warm-up example showing that thermal loss exponentially suppresses the expectation values of projective measurements of bosonic systems. This can be viewed as an analogous result to the well-known exponential suppression induced by local depolarizing noise in DV systems~\cite{aharonov1996limitations, muller2016relative, hirche2020contraction, wang2020noise, quek2022exponentially}. 
However, DV analyses are based on entropy-accumulation arguments and exploit the fact that depolarizing noise drives any quantum state towards the infinite-temperature thermal state, commonly referred as the maximally mixed state. As a result, these analyses are not directly applicable to the continuous-variable setting, where dissipation is modeled by finite-temperature noise channels. Nevertheless, we show an analogous phenomenon for a specific class of observables, namely projective measurements. 

To illustrate the effect of thermal loss, we consider a uniform noise model where the output state of a circuit is acted upon by a layer of thermal loss on all of its $m$ modes, after which we compute its overlap with another state (Figure \ref{fig:uniform_layer}). A prominent example of this overlap estimation problem is heterodyne detection, where the problem of estimating a value of the probability density function maps to the problem of estimating the overlap of the output state with a coherent state. 
\begin{figure}
    \centering
    \includegraphics[width=0.7\linewidth]{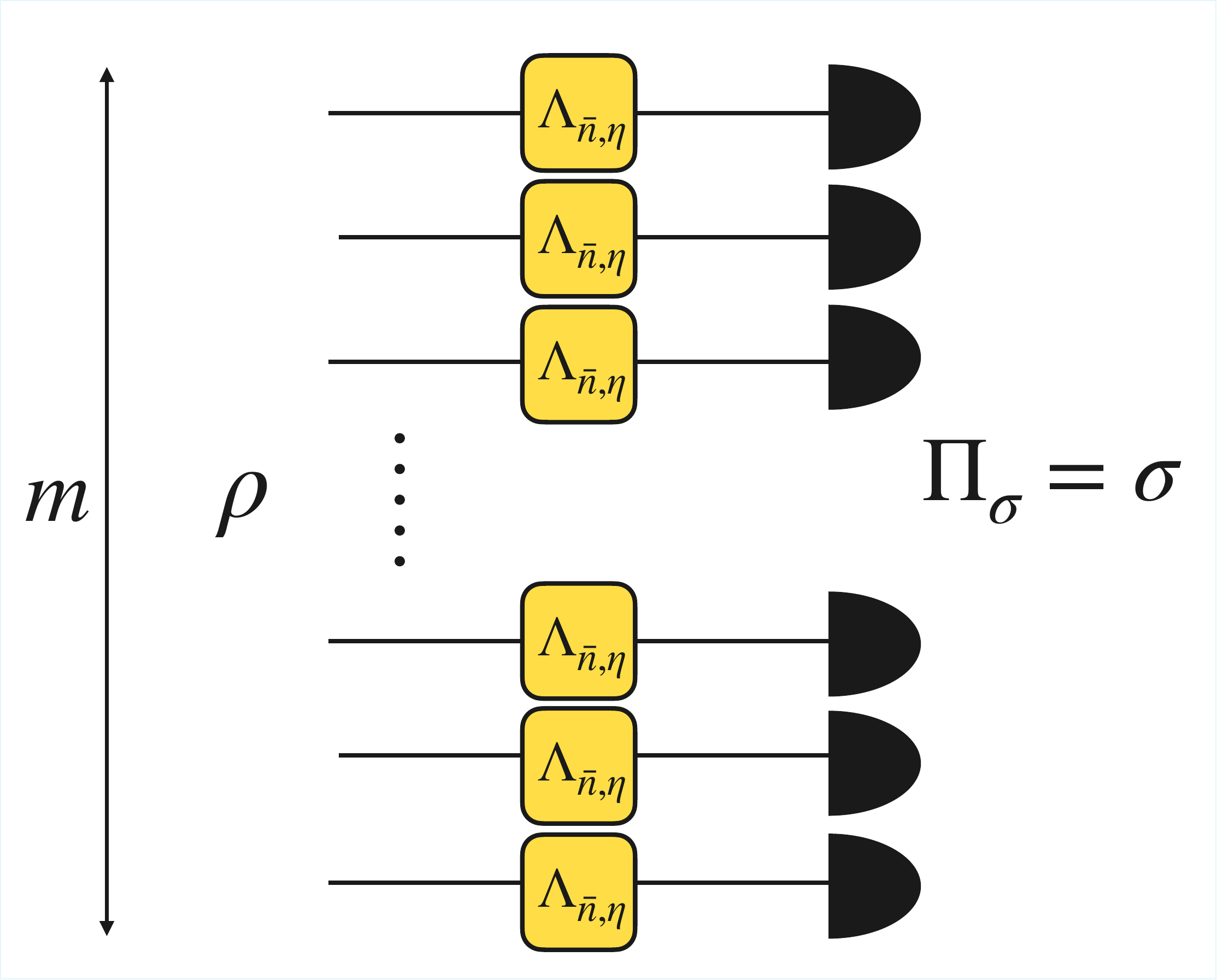}
    \caption{The setting for \cref{prop:purity_uniform}. Computing the overlap between an arbitrary bosonic quantum state $\rho$ going through a uniform layer of noise $\Lambda_{\bar n,\eta}^{\otimes m}$ and any other arbitrary quantum state $\sigma$ becomes trivial as the overlap rapidly decays to zero.} 
    \label{fig:uniform_layer}
\end{figure}

Under this noise model, we find that overlaps are exponentially  suppressed:
\begin{mainproposition}[Overlap decay under thermal loss]\label{prop:purity_uniform} Given two quantum states $\rho$ and $\sigma$ and the $m$-mode thermal loss channel $\Lambda_{\bar n,\eta}^{\otimes m}$ satisfying $\bar n(1-\eta) \in \Omega(1)$,
\begin{align}
   \left|\Tr[\Lambda_{\bar n,\eta}^{\otimes m}(\rho)\sigma] \right|
   \in \exp(-\Omega(\bar n(1-\eta)m)).
\end{align}   
\end{mainproposition}
\noindent The proof of Proposition \ref{prop:purity_uniform} is given in section \ref{appsec:purity_uniform}  of the Supplementary Material and follows from combining H\"older's inequality with a technical result on the decay of the purity of a quantum state under a thermal loss channel. 

Proposition \ref{prop:purity_uniform} can be viewed as delineating the boundary for meaningful quantum computations in the presence of thermal loss. Namely, when thermal loss affects all modes of a bosonic system, the expectation value of projective measurements become trivial in the sense of vanishing exponentially (in the number of modes $m$) towards zero. Or, turning it around, a necessary condition for non-trivial quantum computations in this model, is that the noise rate $\bar n(1 - \eta)$ shrinks with the number of modes $m$ at least as $\sim\frac1m$.

Note that, while insightful, the noise model behind Proposition \ref{prop:purity_uniform} suffers from two limitations. Firstly, the uniform noise model is rather strong and it would be interesting to see if a weaker noise model, where only certain gates are noisy also allows for classical simulation. Secondly, the result only holds nontrivially for $\bar n \neq 0$, and as such, does not apply to pure loss channels. 

Hereafter, we utilize new simulation techniques (detailed in Section \ref{sec:displacement_prop}) for weaker gate-based noise models. These simulation techniques can be seen as a combination of phase-space techniques commonly used in CV quantum information and Pauli propagation techniques used in DV quantum information.


\section{Computational regimes under gate-based noise}
\label{sec:gate_based}

In this section, we consider a quantum state $\rho_0$ evolving under the quantum circuit described by the channel:
\begin{figure}
    \centering
\includegraphics[width=\linewidth]{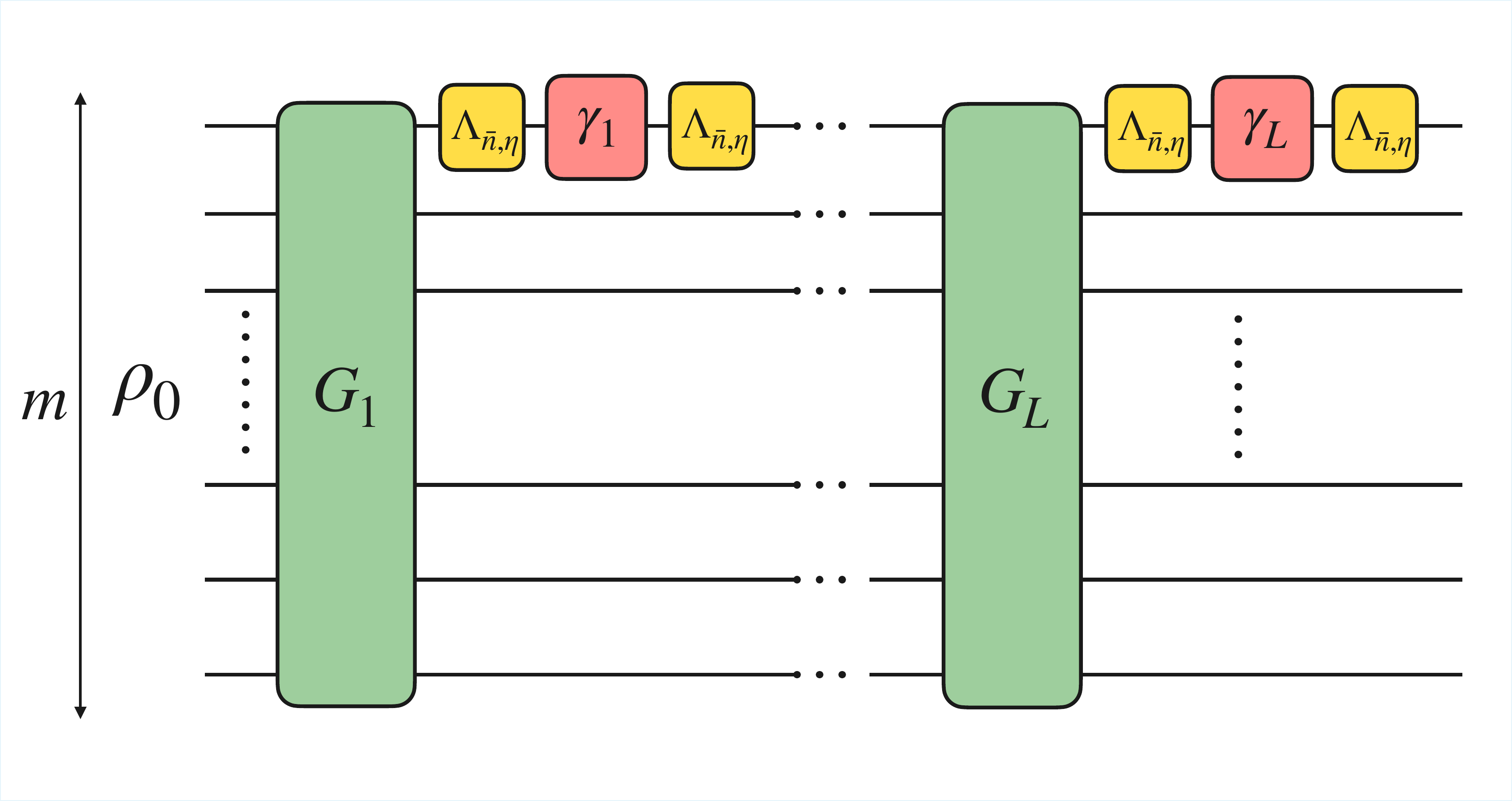}
    \caption{A circuit for bosonic computations with noisy cubic phase gate. $\forall i \in \{1,\dots,L\}$, $G_i$ are Gaussian unitary gates, $\gamma_i$ represents the cubic phase gate with cubicity $\gamma_i$ (acting on the first mode), and $\Lambda_{\bar n, \eta}$ is the thermal loss channel.}
    \label{fig:noisy_bosonic_circit}
\end{figure}
\begin{equation}\label{eq:noisy_cubic_bosonic}
    \calU \coloneqq \calU_L \circ \calU_{L-1}\circ \dots \circ \calU_1,
\end{equation}
where
\begin{equation}
    \calU_j \coloneqq \Lambda_{\bar n, \eta}\circ  \calC_j\circ \Lambda_{\bar n, \eta}\circ \calG_j,
\end{equation}
$\forall j \in \{1,\dots,m\}$, where $\calC_j$ is a single-mode cubic phase gate with cubicity $\gamma_j$, $\Lambda_{\bar n, \eta}$ is a single-mode thermal loss channel characterized by $\bar n$ and $\eta$, and $\calG_j$ is an $m$-mode Gaussian unitary gate characterized by a $2m \times 2m$ symplectic matrix $S_j$ and a $2m$-dimensional displacement vector $\bm{d}_j$. We assume that both $\mathcal{C}_j$ and $\Lambda_{\bar n, \eta}$ act on the same mode, which without loss of generality\ \footnote{If a cubic phase gate acted on a different mode, it could equivalently be represented as a cubic gate acting on the first mode, preceded and followed by swap gates. Since swaps are Gaussian operations, they can be absorbed into the surrounding Gaussian unitaries $\mathcal{G}_j$.} we take to be the first (see Figure~\ref{fig:noisy_bosonic_circit}). Furthermore, we assume that the characteristic function of the input state $\rho_0$ is efficiently computable---this is the case for instance for tensor product states with efficiently computable single-mode characteristic functions.

Gaussian circuits can be simulated efficiently with a classical computer \cite{Bartlett2002}. For near-Gaussian circuits, in which the cubic phase gates are close to the identity (i.e., the cubicity of all cubic phase gates goes to zero), one would expect that efficient simulation of the circuit is possible. For completeness, we prove this in section \ref{appsec:near-Gaussian_approx} of the Supplementary Material, where we also identify the threshold cubicity below which overlap estimation and quadrature moment computations can be computed efficiently to arbitrary precision.

In contrast, here we initiate a systematic investigation of the more subtle interplay between noise and fundamental quantum resources in noisy bosonic circuits. Specifically, the following sections show how noise, non-Gaussianity, symplectic coherence, and energy collectively determine the boundaries of quantum computational regimes that are trivial or admit efficient classical simulation.

\subsection{Theoretical guarantees}\label{sec:theoreticalguarantees}

Here, we state a series of guarantees that establish the regimes in which the outputs of bosonic circuits become trivial in the sense of concentrating to zero (Theorem~\ref{theo:noise_induced_conc}) or become classically simulable (Theorems~\ref{theo:overlap_gate_noise} and\ \ref{theo:quad_gate_noise}). Our results are all stated in terms of \textit{contraction coefficients} that depend on the noise parameters and the quantum resources generated by the circuit (Definition~\ref{defi:contraction_coeff}). 
Here we focus on formally stating our theorems and postpone discussing them in Section~\ref{sec:interplay}, where we formally define and analyze these contraction coefficients. \\

\paragraph*{Exponential concentration.} We start by deriving an analogous concentration result to Proposition~\ref{prop:purity_uniform}, but this time for the gate-based noise defined in Eq.~(\ref{eq:noisy_cubic_bosonic}) and a vacuum state $\ket 0^{\otimes m}$ input. 

 \begin{maintheorem}[Noise-induced concentration in bosonic circuits]\label{theo:noise_induced_conc} Consider an input $m$-mode vacuum state $\rho_0 = \ket 0\!\bra 0 ^{\otimes m}$, a bosonic circuit described by Eq.~(\ref{eq:noisy_cubic_bosonic}) and an observable $\hat O$ such that the 
 trace norm of $\hat O$ is at most exponential in $m$. Then for a contraction coefficient (defined below in Definition \ref{defi:contraction_coeff}) $\mathfrak c < 1$ and sufficiently large $L =  \Omega(m)$, the associated expectation value satisfies
 \begin{equation}
     \left|\Tr[\mathcal U(\rho_0) \hat O]\right| \in \exp\left(-\Omega(m)\right). 
 \end{equation}
 \end{maintheorem}
 \noindent The proof of Theorem \ref{theo:noise_induced_conc} is given in \cref{apptheo:global_noise_induced conc} in the Supplementary Material and follows by expressing $\Tr[\mathcal U(\rho)  \hat O]$ using the optical equivalence theorem and proving that the magnitude of the expectation value goes to zero for $\mathfrak c < 1$ and sufficiently large $L = \Omega(m)$.  In particular, as quantum states have unit trace norm, Theorem~\ref{theo:noise_induced_conc} encompasses (global) projective measurements.

Theorem \ref{theo:noise_induced_conc} identifies another regime where noisy universal bosonic computations become trivial. This regime is determined by the contraction coefficient $\mathfrak c$, with $\mathfrak c = 1$ marking the computational phase transition between trivial and non-trivial behavior. The value of $\mathfrak c$ is set by the interplay of noise and the resources characterizing the universal bosonic circuit, as discussed in the Section~\ref{sec:interplay}. \\

\paragraph*{Classical simulability.}
Next, we turn our attention to the classical simulability of bosonic circuits and consider the problem of estimating the overlap with coherent state projectors $\left(\bigotimes_{i=1}^k \ket{\alpha_i}\bra{\alpha_i}\right)\otimes \mathbb{I}^{\otimes m-k}$:

\begin{maintheorem} [Classically estimating overlaps under gate-based noise]\label{theo:overlap_gate_noise}
    Given an initial $m$-mode quantum state $\rho_0$ evolving under the bosonic circuit described by Eq.~(\ref{eq:noisy_cubic_bosonic}), 
    the overlap of the output state with the coherent state projectors $\left(\bigotimes_{i=1}^k \ket{\alpha_i}\!\bra{\alpha_i}\right)\otimes  \mathbb{I}^{\otimes m-k} $ given by
    \begin{equation}
        \Tr\left[\mathcal U(\rho_0)\left(\bigotimes_{i=1}^k \ket{\alpha_i}\!\bra{\alpha_i}\right)\otimes  \mathbb{I}^{\otimes(m-k)} \right],
    \end{equation}
    for $k = \mathcal O(\log m)$ can be efficiently estimated classically with precision $\epsilon$ and success probability $1-\delta$ in time $\calO(\mathrm{poly}(m,\epsilon)\log(1/\delta))$ in the following three regimes:
    \begin{enumerate}
        \item[\emph{(}i\emph{)}] $L  \in \mathcal O(\log(m))$ for arbitrary $\mathfrak c \in \calO(1)$.
        \item[\emph{(}ii\emph{)}]  $\mathfrak c < 1$ and $L \in \calO(\mathrm{poly}(m))$
        \item[\emph{(}iii\emph{)}] $\mathfrak{d}_{\epsilon/2^k} < 1$ and $L \in \calO(\mathrm{poly}(m))$.
    \end{enumerate}
    where $\mathfrak c$ and $\mathfrak{d}_\epsilon$ are \textit{contraction coefficients} defined below in Definition~\ref{defi:contraction_coeff}.
\end{maintheorem}
\noindent Theorem \ref{theo:overlap_gate_noise} is based on the displacement propagation algorithm for estimating characteristic functions, detailed in Section \ref{sec:displacement_prop}, and combines Lemmas \ref{applem:adaptive_proj_estimation} and\ \ref{applem:high-Gaussian_proj_estimation} from the Supplementary Material.

Finally, we consider the problem of estimating quadratures expectation values. Namely, measuring operators of the form
\begin{align}
    \Tr[\hat{r}_j^k \calU(\rho_0)],
\end{align}
for $\hat r = \hat q, \hat p$ and $k=1,2$.
The runtime in this case includes an additional dependency on the moments of the local energy operators $\hat{N}_j \coloneqq \frac{1}{2}\left(\hat p_j^2 + \hat q_j^2 \right)$. For convenience, we introduce the following parameter:
\begin{align}
     E \coloneqq \max_{j\in[m]}\left(\max\left( \Tr[\hat N_j^3 \calU(\rho_0)] , \Tr[\hat N_j^2 \calU(\rho_0)] \right)\right)\, .
\end{align}

\begin{maintheorem} [Classically estimating quadratures under gate-based noise]\label{theo:quad_gate_noise}
    Given an initial $m$-mode quantum state $\rho_0$ evolving under the bosonic circuit described by Eq.~(\ref{eq:noisy_cubic_bosonic}),  
    and given $\hat O = \hat p_j, \hat q_j, \hat p_j^2,$ or $\hat q_j^2$ for some $j\in[m]$, the expectation value
    \begin{align}
        &\Tr[\hat O  \calU(\rho_0)],
    \end{align}
    can be estimated classically with precision $\epsilon$ and success probability $1-\delta$ in time $\calO(\mathrm{poly}(m,\epsilon, E)\log(1/\delta))$ in the following three regimes:
    \begin{itemize}
        \item[\emph{(}i\emph{)}] $L  \in \mathcal O(\log(m))$ for arbitrary $\mathfrak c \in \calO(1)$.
        \item[\emph{(}ii\emph{)}] $\mathfrak c < 1$ and $L \in \calO(\mathrm{poly}(m))$.
        \item[\emph{(}iii\emph{)}] $\mathfrak{d}_{\epsilon^4/(8E+12)^4} < 1$ and $L \in \calO(\mathrm{poly}(m))$.
    \end{itemize}
    where $\mathfrak c $ and $\mathfrak{d}_\epsilon$ are \textit{contraction coefficients} defined below in Definition~\ref{defi:contraction_coeff}.
\end{maintheorem}

\noindent Theorem \ref{theo:quad_gate_noise} combines Lemmas \ref{applem:adaptive_quad_estimation} and \ref{applem:high-cubicity_quad_estimation} from the Supplementary Material and is also based on our displacement propagation algorithm and a new technical result connecting quadrature moments to finite differences of characteristic functions (Lemma \ref{applem:quad_finite_diff}).

Note that the evaluation of $\hat q_j^2, \hat p_j^2,\forall j\in \{1,\dots,m\}$ also allows for calculating the expectation value of the number operator $\hat N := \sum_{j=1}^m \frac{\hat q_j^2 + \hat p_j^2}{2}$ and thus estimating the energy.

\subsection{The interplay of quantum resources and noise}\label{sec:interplay}

\paragraph*{Contraction coefficients.} The contraction coefficients $\mathfrak c$ and $\mathfrak{d}_\epsilon$ play a pivotal role in Theorems ~\ref{theo:noise_induced_conc}, ~\ref{theo:overlap_gate_noise} and~\ref{theo:quad_gate_noise}, determining the regimes of trivial quantum computations and efficient classical simulability. In this section, we provide the formal definition of these coefficients and use them to explore the interplay between noise and quantum computational resources. 

\begin{maindefinition}[Contraction coefficients]\label{defi:contraction_coeff}
    Given a bosonic circuit defined by Eq.~(\ref{eq:noisy_cubic_bosonic}), we define the associated contraction coefficients given by
    \begin{align} \label{eq:contraction_coeff}
        & \mathfrak c = \frac{\Gamma(1/4)}{\sqrt{24\pi}} \times \underbrace{\frac{1}{\sigma_{\min} \gamma_{\min}}}_{\text{resources}} \times \underbrace{\frac{1}{\eta (1/2 + \bar n)^{1/4}(1-\eta)^{1/4}}}_{\text{noise}},
        \\&\mathfrak{d}_\epsilon =  \frac{1}{\sqrt{12\epsilon}} \times \underbrace{\sqrt{\frac{ML}{\gamma_{\min}}}}_{\text{resources}} \\
        &\hspace{10mm}\times\underbrace{\frac{\exp(-(1/2+\bar n)(1-\eta)\epsilon^2/(4M^2L^2))}{\eta^{1/4}(1/2+\bar n)^{1/2}(1-\eta)^{1/2}}}_{\text{noise (with energy)}},
    \end{align}
    where
    \begin{align}
        &\gamma_{\min} = \min_j |\gamma_j| ,
        , \;\;
        \sigma_{\min} = \min_j |(S_j)_{q_1,p_1}|, \nonumber
        \\ & \label{eq:sigma}
        \\&M = \max_j\max_{\substack{ \br' \in \mathbb{R}^{2m} }} \bigg |\left\{\frac{\partial}{\partial q_1}\chi_{j}(\br)\right\}_{\br =\br'} \bigg |,
     \end{align}
    where $\chi_{j}$ is the characteristic function of the quantum state evolved under $j$ noisy circuit layers, i.e. $\rho_j \coloneqq  \calC_j \circ \Lambda_{\bar n, \eta}\circ \calU_{j-1} \circ\dots \circ \calU_1(\rho_0)$. 
\end{maindefinition}

The $\sigma$ component can be intuitively understood as the minimum amount by which the Gaussian gates are mixing the position quadrature with the momentum quadrature on the modes on which the cubic phase gates act, i.e.\ the minimum symplectic coherence with respect to those modes of the Gaussian gates of the given circuit \cite{upreti2025interplay}.  By our convention, all cubic phase gates act on the first mode hence the dependence on $q_1$ and $p_1$ in Eq.~\eqref{eq:sigma}. Here, we note that in an experimental implementation of a bosonic computation, $\sigma_{\min}$ can be computed by learning the symplectic matrices of the Gaussian unitary gates implemented in the computation, which can be done efficiently using methods provided in \cite{fanizza2025learning}.

On the other hand, $\gamma_{\min}$ can be intuitively seen as the minimum strength of the non-Gaussian gates present in the quantum system, while non-Gaussianity itself is captured by the number $L$ of such gates in the circuit.
Finally, as we show in the Supplementary Material (see Lemma\ \ref{lem:curvature_bound}), $M$ can be bounded in terms of the first and second moments of local quadrature operators evaluated on partially evolved states, and can thus be interpreted as a global upper bound on the energy throughout the computation.

The key role of these contraction coefficients in Theorems~\ref{theo:noise_induced_conc},\ \ref{theo:overlap_gate_noise} and\ \ref{theo:quad_gate_noise} mirrors that central role played by contraction coefficients in DV quantum information, where they provide a rigorous framework for characterizing the action of noise and, in particular, its interplay with unitary dynamics. They have been employed in several recent results on the limitations of noisy quantum devices\ \cite{hirche2020contraction, franca2020limitations, depalma2022limitations, mele2024noise, ibarrondo2025average}, including bounds on error-mitigation protocols\ \cite{quek2022exponentially} and regimes of classical simulability\ \cite{rall2019simulation, schuster2024polynomial, gonzalez2024pauli, angrisani2025simulating}.

In the DV case, analyses of the role of quantum resources and noise in delineating the boundary between non-trivial quantum computation and efficient classical simulation often treat these resources largely independently. This typically leads to relatively simple functional relations, where the ratio between noise strength and a single resource suffices to determine classical simulability.
For example, Ref.\ \cite{rall2019simulation} identifies regimes of classical simulability by examining the interplay between the noise rate and the rate at which magic is injected by non-Clifford gates. Even more strikingly, Ref.\ \cite{lerch2024efficient} shows that in noiseless circuits the injection rate of magic alone already suffices to determine regimes of classical simulability.
A further example is provided by the analyses of Refs.~\cite{de2021quantum, hirche2023quantum, mele2024noise}, which rely on contraction of the quantum Wasserstein distance. In these latter works, noise-induced concentration is characterized by the rate at which noise suppresses the non-local interactions generated by the unitary gates. \\

\paragraph*{Computational phase transitions in bosonic systems.}
In noisy bosonic circuits, in contrast to the DV case, the effects of quantum resources are intrinsically intertwined, leading to more intricate phase diagrams (see Figure~\ref{fig:d_phase_diag}). While the effect of the average number thermal photons is straightforward (increasing temperature makes the bosonic system ``more classical''), the interplay between the other resources is more subtle, with the cubicity parameters $\gamma_j$ never appearing in isolation in the computational complexity of our algorithms, but always in combination with either the rate at which the Gaussian gates inject symplectic coherence (captured by $\sigma$ in $\mathfrak{c}$) or with the energy of partially evolved states (captured by $M$ in $\mathfrak{d}_\epsilon$). 

By analyzing $\mathfrak c$ and $\mathfrak d_\epsilon$, we find that cubicity and symplectic coherence play vastly different roles, which we explain below:

(i) Combined with high symplectic coherence, large cubicity makes noisy bosonic computations easier to simulate classically. This is shown in Figure \ref{fig:d_phase_diag}, where for all values of noise parameters $\eta$ and $\bar n$, increasing the value of the quantum resources $\sqrt{\gamma_{\min}}\sigma$ decreases $\mathfrak c$ to a value less than one and brings the bosonic system to a regime of efficient classical simulation. Higher values of symplectic coherence are associated with increasing complexity of simulation in the noiseless case \cite{upreti2025interplay}. Here we observe that these nonclassical features are fragile: for noisy computations, a higher value of a quantum computational resource can facilitate the classical simulation of the system.

(ii) We also see that when the cubicity is sufficiently high throughout the circuit, namely
\begin{eqnarray}
   \gamma_{\min} &>& \frac{ML}{12\epsilon\eta^{1/2}(1/2 + \bar n)(1-\eta)} \nonumber \\ &\times& \exp\left(-(1/2+ \bar n)(1-\eta)\frac{\epsilon^2}{2M^2L^2}\right)  ,
\end{eqnarray}
we have $\mathfrak d_\epsilon < 1$ and fall into an efficient classical simulation regime. The cubicity threshold above which the circuit becomes classically simulable efficiently depends on the average energy, thermal photons and noise rate. Here again, we witness a scenario in which a higher cubicity in a noisy bosonic circuit enables its efficient classical simulation.

\begin{figure}
    \centering
    \includegraphics[width=\linewidth]{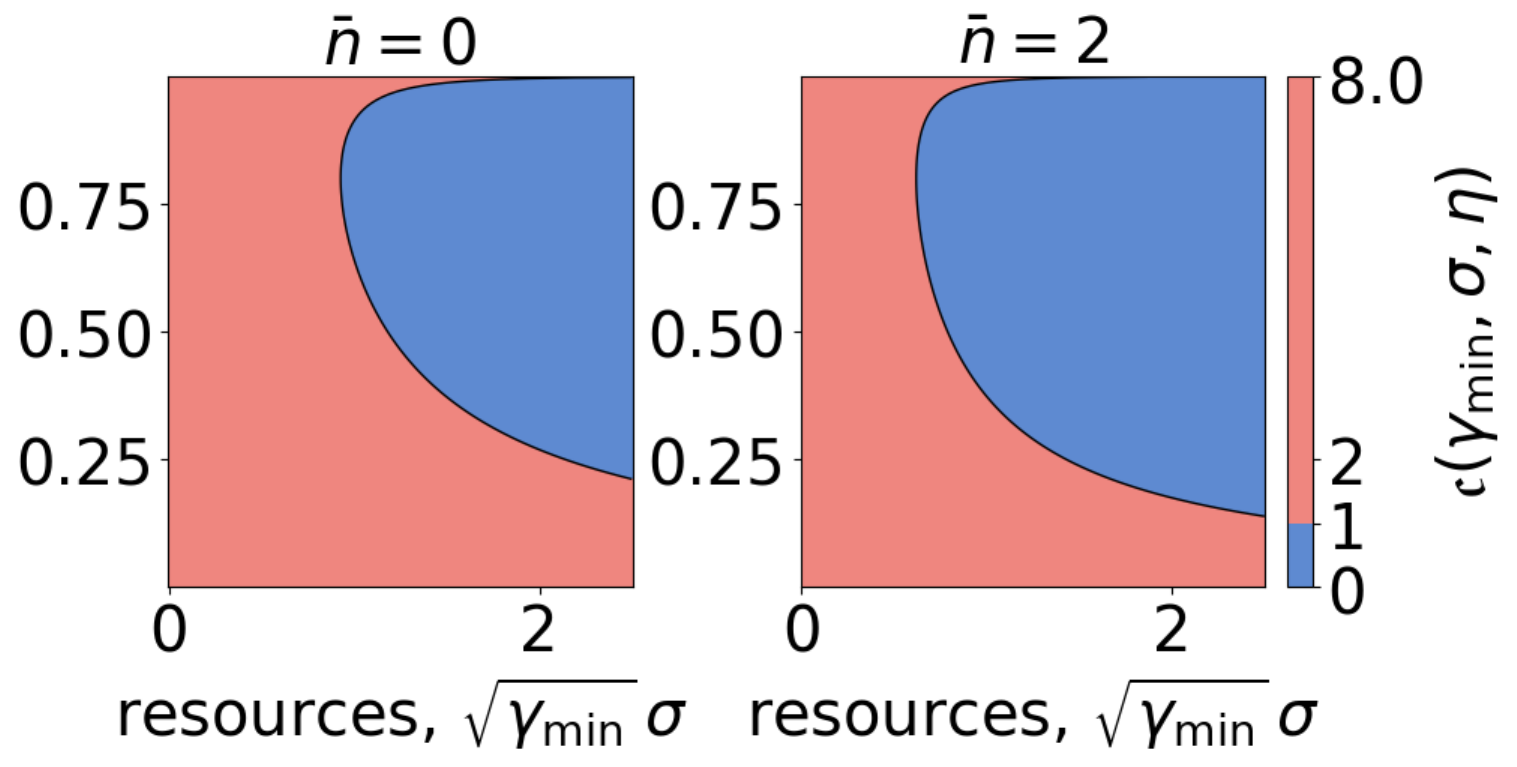}
    \caption{ Contraction coefficients $\mathfrak c$ as a function of photon loss $\eta$ and quantum resources $\sqrt{\gamma_{\min}} \sigma$ for different values of average thermal photons $\bar n$. $\forall \eta,\bar n$, increasing the quantum resources leads us to the efficiently classically simulable phase (marked by the blue region).}
    \label{fig:d_phase_diag}
\end{figure}

In both cases, the runtime of our classical simulation algorithms grows (at least) exponentially with circuit depth whenever $\mathfrak{c}, \mathfrak{d}_\epsilon > 1$, thereby leaving room for substantial quantum speed-ups. In contrast, when $\mathfrak{c}, \mathfrak{d}_\epsilon \leq 1$, the simulation can be carried out in polynomial time, and in some regimes the dynamics even become computationally trivial: expectation values are exponentially suppressed, severely restricting the potential for quantum advantage in estimating observables.

\section{Methods: The displacement propagation algorithm}\label{sec:displacement_prop}

\begin{figure*}
    \centering\includegraphics[width = \linewidth]{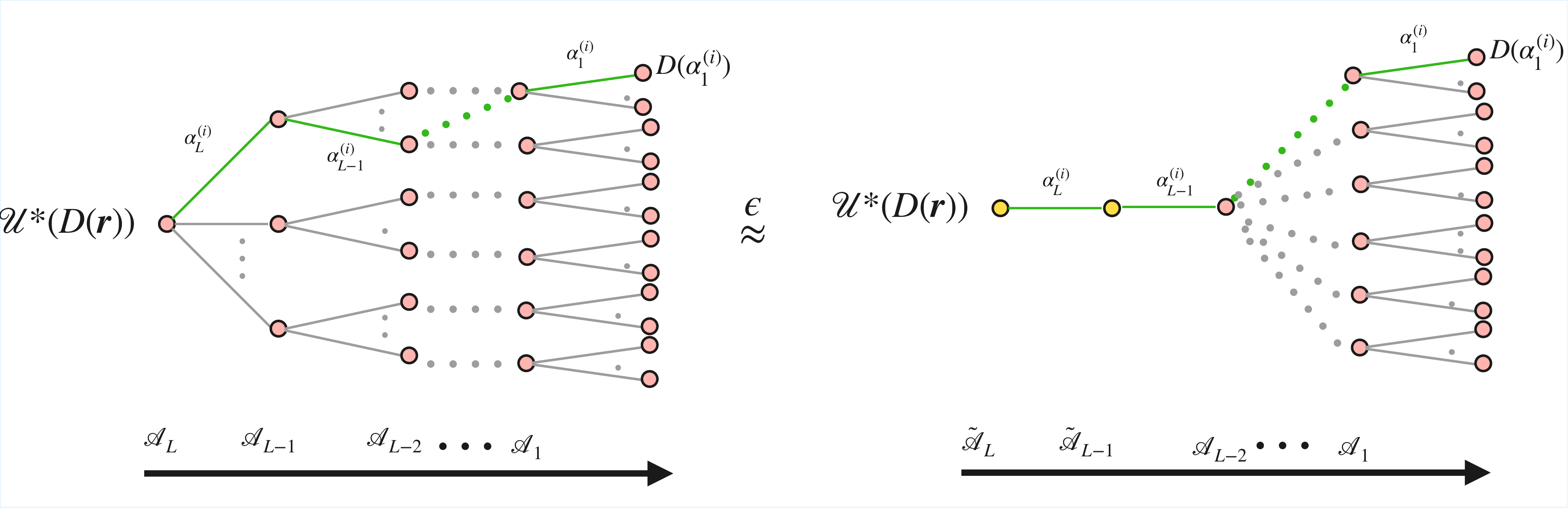}
    \caption{\emph{Intuition behind the displacement propagation algorithm.} To evaluate the characteristic function $\mathrm{Tr}[\mathcal U^* (\hat D(\br))\rho_0]$, we build a Markov chain by decomposing $\mathcal U^*$ into $L$ possibly non-physical quantum maps $\mathcal A_1 \circ \cdots \circ \mathcal A_L$ and for each of the $L$ layers, we sample a phase-space point $\alpha_{i}$ according to a probability distribution that depends on both the properties of the current layer and the previously sampled phase-space point. At the $L$th layer, we obtain a phase-space point $\alpha_1 \in \mathbb{C}^m$ at which we estimate the characteristic function $\mathrm{Tr}[\rho_0\hat D(\alpha_1)]$. By carefully choosing the quantum channels, we ensure that the statistical average of $N$ such characteristic functions approximates $\mathrm{Tr}[\mathcal{U}^*(\hat D(\br))\rho_0]$ to arbitrary precision using a polynomial number of samples. This process is illustrated by the tree on the left and formalized in Algorithm~\ref{alg:sim-unbiased} in the Supplementary Material. However, for certain regions of phase space, efficient sampling at every layer may not be feasible, particularly due to the non-linearity introduced by non-Gaussian cubic phase gates. In these cases, we employ adaptive simulation algorithms, in which the layer $\mathcal A_j$ is replaced by an approximate map $\tilde{\mathcal A}_j$ (the modification is designed so that the corresponding sampled phase-space point becomes deterministic; see Section~\ref{section:biased-oracle} of the Supplementary Material). This procedure ensures that the resulting Markov chain approximates the target characteristic function with only a negligible bias. This adaptive approach is depicted by the tree on the right and formalized in Algorithm~\ref{alg:sim-adaptive-general} in the Supplementary Material, where we also provide a refined procedure for high cubicity (Algorithm~\ref{alg:sim-adaptive-high-gaussian}).
    }
    \label{fig:intuition_displacement_propagation}
\end{figure*}

In this section we briefly present the \textit{displacement propagation} simulation framework that we developed to obtain our theoretical results above. This framework should be understood both as a theoretical tool for identifying potential regimes for quantum advantage as well as a new practical tool for simulating bosonic systems.

The displacement propagation algorithm follows from the optical equivalence theorem for the CV characteristic function: given an operator $\hat O$ and an output state $\mathcal U(\rho_0)$, the expectation value of $\hat O$ over $\mathcal U(\rho_0)$ is given by
 \begin{eqnarray}\label{eq:optical_equivalence_1}
     \Tr[\mathcal U(\rho_0) \hat O] &=& \frac{1}{\pi^m} \int_{\mathbb{R}^{2m}} \mathrm d^{2m}\br \Tr[\hat O \hat D(\br)]\times \Tr[\mathcal U(\rho_0) \hat D^\dagger(\br)], \nonumber \\
     &=& \frac{1}{\pi^m} \int_{\mathbb{R}^{2m}} \mathrm d^{2m}\br \chi_{\hat O} (\br) \chi^{*}_{\mathcal U(\rho_0)}(\br),
 \end{eqnarray}
 where $\br= \{q_1,p_1,\dots,q_m,p_m\}$. Recalling the Fourier 1-norm of $\hat O$ as $\norm{\hat{O}}_{\mathbb{F},1} \coloneqq \int_{\mathbb{R}^{2m}} \mathrm d\br \left\lvert \frac{\Tr[\hat O\hat D(\br)]}{\pi^m} \right\rvert$ and defining the probability density function
 \begin{equation}
   p_{\hat o}(\br) \coloneqq \frac{|\chi_{\hat O}(\br)|}{\pi^m \fnormone{\hat O}}  ,
 \end{equation}
 this can be written as
 \begin{equation}\label{eq:optical_equivalence_2}
     \Tr[\mathcal U(\rho_0) \hat O] = \fnormone{\hat O} \int_{\mathbb{R}^{2m}}\mathrm d^{2m}\br p_{\hat o}(\br) \chi^{*}_{\mathcal U(\rho_0)}(\br) \mathrm{arg}(\chi_{\hat O}(\br)),
 \end{equation}
 where $\mathrm{arg}(z)$ denotes the phase of the complex number $z$.
 Therefore, $\Tr[\mathcal U(\rho_0) \hat O]$ can be estimated by sampling a phase-space point $\br$ from the probability distribution $p_{\hat o}(\br)$ and estimating $\chi^{*}_{\mathcal U(\rho)}(\br) \mathrm{arg}(\chi_{\hat O}(\br))$ at that point. A similar technique can be applied repeatedly via a Markov Chain sampling argument.
 
 This technique is closely related to randomized versions of the Pauli propagation algorithm\ \cite{rall2019simulation, martinez2025efficient} and it  is similar to the sampling techniques with quasi-probability distribution commonly used in CV quantum information \cite{veitch2012negative,Mari2012,pashayan2015estimating,keshari2016,Chabaud2024phasespace,frigerio2024}, with the key difference being that using characteristic function for the estimation is less common, with only a few works for noiseless computations \cite{lim2025efficient,oh2025adaptive}. The main reason is that characteristic functions are complex-valued and not always normalizable to use as a probability distribution. However for 
 \begin{equation}
     \hat O = \left(\bigotimes_{i=1}^k \ket{\alpha_i}\!\bra{\alpha_i}\right) \otimes \mathbb{I}^{\otimes m-k},
 \end{equation}
 the modulus of characteristic function is given as
 \begin{equation}
    |\chi_{\hat O}(\br)| = \left(\Pi_{i=1}^k \exp(-(q_i^2 + p_i^2)/2)\right) \left(\Pi_{j=k+1}^m \delta(q_j) \delta(p_j)\right),
 \end{equation}
 which is normalizable and hence can be used as a probability distribution. Another problem is that the calculation of the characteristic function of the evolved state might not be straightforward.
 
 The main technical contribution of this work is to provide techniques to efficiently estimate this characteristic function for universal bosonic computations with noisy cubic phase gates in the specified regimes, by applying Monte Carlo Markov Chain sampling algorithms over carefully constructed Markov chains, obtained through suitable channel decompositions. We call this sampling algorithm the ``displacement propagation algorithm''. The key idea is to express noisy circuit layers as sequences of (possibly unphysical) linear maps and to model their action on displacement operators by chaining together estimators, chosen adaptively throughout the procedure. This is detailed in Appendix~\ref{app:c-sim} and a visual intuition for the displacement propagation algorithm is given in Figure \ref{fig:intuition_displacement_propagation}.

The methods for efficiently computing the characteristic function can also be extended to estimating quadrature operators. 
In particular, as we show in Lemma \ref{applem:quad_finite_diff}, expectation values of single-mode quadratures can be approximated as follows using a first-order Taylor expansion
 \begin{align}\label{eq:quadexpand}
     &\Tr[\rho \hat q_j]  = -\frac{i}{\delta}\left({\Tr[\rho \hat D_j(0,\delta)] - 1}\right) + \calO\left( \delta \Tr[\rho \hat N_j^2]  \right),
     \\&\Tr[\rho \hat p_j]  = -\frac{i}{\delta}\left(1-{\Tr[\rho \hat D_j(\delta,0)] }\right) + \calO\left( \delta \Tr[\rho \hat N_j^2]  \right).
 \end{align}
Importantly, the value of $\delta$ has to be chosen such that $\delta \ll \Tr[\rho \hat N_j^2] $ and in practice the state $\rho$ in Eq.~\eqref{eq:quadexpand} will be the state  evolved under the noisy circuit, i.e. $\rho\coloneqq \mathcal{U}(\rho_0)$. 
While this approximation enables the estimation of quadrature expectation values from characteristic functions, it comes at the cost of increasing the variance of the classical estimator by a factor of order $\delta^{-2}$. Since this factor must be larger than $\Tr[\rho \hat N_j^2]$, the classical simulation overhead grows with the energy of the system.
Analogous expression can be derived also for higher order polynomials (we give the expression for the second moment of quadratures in terms of characteristic functions in Lemma \ref{applem:quad_finite_diff}). These observations underlie Theorem~\ref{theo:quad_gate_noise}. 

\section{Conclusion}
\label{sec:concl}

In this work, we have developed classical simulation algorithms to estimate expectation values of observables at the output of noisy bosonic computations, for both uniform and gate-dependent noise models. We have uncovered a rich interplay between noise and quantum computational resources, identifying regimes where circuits of various depths can be simulated efficiently by classical means.

Our results help pinpoint the boundary for bosonic quantum computational advantage under noise, where quantum error correction becomes mandatory to maintain an edge over classical computers, and opens several research directions. Whether the same interplay of resources also determines efficient classical simulation beyond expectation values (e.g., sampling) represents an interesting research direction. Another promising direction is to investigate the classical simulability of bosonic circuits with different non-Gaussian gates, such as the Kerr gate, or with different noise models such as phase noise or partial distinguishability, and understand the interplay of resources leading to computational phase transitions.

Our main tool is the \textit{displacement propagation} formalism, a new, powerful framework which combines classical Monte Carlo Markov Chain sampling algorithms with the optical equivalence theorem for CV characteristic functions, generalizing to the bosonic setting the Pauli propagation formalism \cite{aharonov2022polynomial,schuster2024polynomial,fontana2023,angrisani2024classically}. While in this manuscript we have focused on use of displacement propagation as a tool to study the simulability of noisy bosonic circuits we foresee it also finding use as a practical simulation algorithm in regimes beyond that for which we currently have guarantees (as has been seen in the case of Pauli propagation~\cite{beguvsic2023fast,rudolph2023classical,beguvsic2024real}). A particularly interesting test case would be the simulation of bosonic error correction codes~\cite{albert2025bosonic}.

Another interesting direction for future research is to see whether the displacement propagation framework can be hybridized with Pauli propagation, or other classical simulation methods for simulating spin systems, to simulate hybrid spin-boson systems. Such hybrid algorithms could be applied to simulate open quantum systems. In particular, they might open up new methods for exploring the strong coupling limit and non-Markovian environments. We leave these research directions for future work.

\section*{Acknowledgements}
UC acknowledges inspiring discussions with G.\ Ferrini, A.\ Leverrier, A.\ Debray, M.\ Frigerio and M.\ Walschaers. 
AA acknowledges inspiring discussions with F.\ A.\ Mele and M.\ S.\ Rudolph.
VU and UC acknowledge funding from the European Union’s Horizon Europe Framework Programme (EIC Pathfinder Challenge project Veriqub) under Grant Agreement No.~101114899. ZH and AA acknowledge support from the Sandoz Family Foundation-Monique de Meuron program for Academic Promotion.

\bibliography{quantum, biblio}

 \clearpage
 \newpage
 \onecolumngrid
 
\renewcommand\partname{} 
\appendix
\begin{center}
 {\Large \textbf{Supplementary Material} }   
\end{center}

\part{}
\parttoc 


\newpage

\section{Notation}
    \noindent\begin{tabular}{ |p{4cm} | p{13.5cm}|  }
        \hline
        \textbf{Symbol} & \textbf{Definition}  \\
        \hline
        $\mathbb{R}$ & Set of real numbers \\
        $\mathbb{R}^*$ & Set of real numbers excluding zero \\
        $\mathbb{C}$ & Set of complex numbers \\
        $\mathbb{C}^*$ & Set of complex numbers excluding zero \\
        $\mathbb{N}$ & Set of natural numbers \\
        $\mathbb{N}^*$ & Set of natural numbers excluding zero\\
        $A\succeq B$ & The operator $A-B$ is positive semidefinite\\
        $\Gamma(\cdot)$ & Gamma function, $\Gamma(n+ \frac{1}{2}) = \binom{n-\frac{1}{2}}{n}n! \sqrt{\pi}$ for $n \in \mathbb{N}$ \\
        $\delta(\cdot)$ & Dirac delta function\\
        $\calN(\mu,\sigma^2)$ & Normal distribution with mean $\mu$ and variance $\sigma^2$\\
        $\mathrm{Gamma}(\alpha, \lambda)$ & Gamma distribution with shape $\alpha$ and rate $\lambda$\\
        $m$ & Number of modes of a bosonic system \\
        $\mathbb{I}$ & Infinite-dimensional identity operator\\
        $\Omega$ &  Standard symplectic form, $\Omega \coloneqq \bigoplus_{j=1}^m\begin{bmatrix}
            0 & 1 \\ -1 &0
        \end{bmatrix}$ \\
        $\br$ & $2m$-dimensional quadrature vector, i.e. $\br = (q_1,p_1,q_2,p_2,\dots,q_m,p_m) \in \mathbb{R}^{2m}$ \\
        $\hat D(\br)$ & Fourier--Weyl operator associated to the displacement $\br$\\
        $\ket{\br}$ & Coherent state associated to the displacement amplitude vector $\br$, $\ket{\br}\coloneqq\hat D(\br)\ket{0}$\\
        $\hbar$ & Reduced Planck constant, we adopt the convention $\hbar = 2$ throughout this work\\
        $\hat{a}, \hat{a}^\dag $ & Annihilation $(\hat{a})$ and creation $(\hat{a}^\dag)$ operators, satisfying $[\hat{a}, \hat{a}^\dag ] = \mathbb{I}$ \\
        $\hat{p}, \hat{q}$ & Position $(\hat{q}\coloneqq\hat{a} + \hat{a}^\dagger)$ and momentum $\hat{p}\coloneqq-i(\hat{a} - \hat{a}^\dagger)$ operators, satisfying $[\hat{q},\hat{p}] = 2i\I$\\ 
        $\hat n$ & Single-mode particle number operator, $\hat{n} \coloneqq\hat{a}^\dagger \hat{a} = \frac{1}{4}\left(\hat{p}^2 + \hat{q}^2 -2 \mathbb{I}\right)$\\
        $\hat N$ &  The energy operator: $\hat N = \sum_{j=1}^m \hat N_j = \sum_{j=1}^m \frac{\hat q_j^2 + \hat p_j^2}{2}$\\
        $\norm{f}_{L_p}$ & $p$-norm of an $L^p$-integrable function $f:\calX \rightarrow \mathbb{C}$, $\norm{f}_p\coloneqq \left(\int_{x\in\calX} \mathrm dx \abs{f(x)}^p\right)^{\frac{1}{p}}$\\
        $\chi_{\hat O}$ & characteristic function of the operator $\hat O$, $\chi_{\hat O}(\br) \coloneqq \Tr[OD(\br)]$ \\
        $\norm{\hat O}_p$ & Schatten $p$-norm, i.e. $\norm{\hat O}_p\coloneqq\Tr[\abs{\hat O}^p]^{\frac{1}{p}}$, where $|\hat O| = \sqrt{\hat O^\dagger \hat O}$\\
        $\norm{\hat O}_{\mathbb{F}, p}$ & Fourier $p$-norm (Def.\ \ref{def:fourier-p}). If $\chi_{\hat O}$ is $L^p$-integrable, then $\norm{O}_{\mathbb{F}, p} \coloneqq \left(\int_{\mathbb{R}^{2m}} \mathrm d\br \left\lvert \frac{\chi_{\hat O}(\br)}{\pi^m} \right\rvert^{p} \right)^{\frac{1}{p}}$\\
        \hline
    \end{tabular}
\vspace{0.25 cm}

\section{Mathematical preliminaries}\label{appsec:math_prelims}
We start by introducing mathematical preliminaries which will be repeatedly used throughout the paper. We begin with a few technical lemmas concerning measure theory, norm inequalities, and approximation bounds, and then recall some standard results from probability theory that underlie our Monte Carlo methods.
\subsection{Useful technical results}
\noindent The following result will be useful for deriving simple upper bounds on high-dimensional integrals.
\begin{lemma}\label{fact:ball}
The volume of the $2m$-dimensional ball of radius $\tau$ equals $\frac{\pi^m}{m!}\tau^{2m}$.   
\end{lemma}
\noindent The following Lemma provides a standard result on the manipulation of multi-dimensional integrals.
\begin{lemma}[Fubini-Tonelli's theorem]
\label{lem:fubini}
Let $f : \mathbb{R}^2 \rightarrow [0,\infty)$  be an integrable function. Then the following identities hold
\begin{align}
    \int_{\mathbb{R}^2}f(x,y)\mathrm dx\mathrm dy = \int_{\mathbb{R}}\left( \int_{\mathbb{R}} f(x,y)\mathrm dy\right)\mathrm dx = \int_{\mathbb{R}}\left( \int_{\mathbb{R}} f(x,y)\mathrm dx\right)\mathrm dy.
\end{align}
\end{lemma}
\noindent We then have the following result for a special type of integral that will be important hereafter:
\begin{lemma}\label{applem:gamma_integral}
    Given $a,b \in \mathbb{R}, A \in \mathbb{R}^+$ and assuming $a\neq0$, then
    \begin{equation}
        I = \int_{\mathbb{R}} \mathrm dx \frac{e^{-A(ax+b)^2}}{\sqrt{|ax+b|}} = \frac{\Gamma(1/4)}{A a^{1/4}}.
    \end{equation}
\end{lemma}
\begin{proof}
    Given the integral
    \begin{equation}
        \int_{\mathbb{R}} \mathrm dx \frac{e^{-A(ax+b)^2}}{\sqrt{|ax+b|}} = \int_{\mathbb{R}} \mathrm dx \frac{e^{-A(|a|x+\mathrm{sgn(a)}b)^2}}{\sqrt{||a|x+\mathrm{sgn(a)}b|}} ,
    \end{equation}
    we first make the substitution
    \begin{eqnarray}
    u &=& |a|x+\mathrm{sgn(a)}b, \nonumber\\
    \mathrm du &=& |a|\mathrm dx,
    \end{eqnarray}
    to get
    \begin{equation}
        I = \frac{1}{|a|} \int_{\mathbb{R}} \mathrm du \frac{e^{-Au^2}}{\sqrt{|u|}}.
    \end{equation}
    Again, making the substitution 
    \begin{eqnarray}
        v &=& \sqrt{A} u, \nonumber \\
        \mathrm dv &=& \sqrt{A} \mathrm du.
    \end{eqnarray}
    We get
    \begin{equation}
        I = \frac{1}{|a| A^{1/4}} \int_{\mathbb{R}} \mathrm dv \frac{e^{-v^2}}{\sqrt{|v|}} = \frac{2}{|a| A^{1/4}}\int_{\mathbb{R^+}} \mathrm dv \frac{e^{-v^2}}{\sqrt{|v|}} = \frac{4\Gamma(5/4)}{|a|A^{1/4}} = \frac{\Gamma(1/4)}{|a|A^{1/4}}.
    \end{equation}
\end{proof}
\noindent We next recall a classical expansion formula that will be useful when approximating integrals and operator-valued functions.

\begin{theorem}[Taylor's remainder theorem]\label{apptheo:taylor_rem}
Let $f : \mathbb{R} \rightarrow \mathbb{R}$ be $k+1$ times differentiable on the open interval between $a$ and $x$
with $f^{(k)}$ continuous on the closed interval between $a$ and $x$. Then
\begin{align}
    f(x) = \sum_{j=0}^k \frac{f^{(j)}(a)}{{j!}}(x-a)^j \; + \;\frac{f^{k+1}(y)}{(k+1)!}(x-a)^{k+1},
\end{align}
for some $y$ between $a$ and $x$, with $k \geq 1$.
\end{theorem}

\noindent We now recall some fundamental definitions and inequalities that will be central for handling operator norms.
\begin{definition}[Schatten $p$-norms]
Let $\mathcal{H}$ be a (possibly infinite-dimensional) Hilbert space and let $O$ be a compact operator on $\mathcal{H}$.  
For $p \in [1,\infty)$, the \emph{Schatten $p$-norm} of $O$ is defined as
\begin{align}
    \|O\|_{p} \coloneqq \left( \Tr\left[ \abs{O}^{p} \right] \right)^{\tfrac{1}{p}},
\end{align}
whenever the trace is finite. 
For $p = \infty$, the Schatten $\infty$-norm coincides with the operator norm:
\begin{align}
    \|O\|_{\infty} \coloneqq \sup_{\rho: \norm{\rho}_1=1} \Tr[O\rho],
\end{align}    
where the supremum is taken over all trace-class operators $\rho$ with unit trace norm.
\end{definition}
\noindent One of the key tools associated with Schatten norms is Hölder’s inequality, which we state next in the operator setting.
\begin{lemma}[Hölder's inequality\ \cite{baumgartner2011inequality}]\label{lem:holder}
Let $p,q\in [1,\infty]$ such that $\frac{1}{p} + \frac{1}{q} = 1$, and let $A, B$ be two operators such that $\norm{A}_p$ and $\norm{B}_q$ are finite. Then,
\begin{align}
    \abs{\Tr[AB]} \leq \norm{A}_p \norm{B}_q.
\end{align}
\end{lemma}
\noindent Finally, we provide a simple but useful inequality relating the expectation value of an operator to its symmetrized second moment.
\begin{lemma}[Tracial second moment]\label{lem:trace_second_moment}
 Let $A$ be an operator and $\rho$ be a quantum state.    
 \begin{align}
     \abs{\Tr[A\rho]}^2 \leq \frac{1}{2}\Tr[\{A, A^\dag \}\rho].
 \end{align}
\end{lemma}
\begin{proof}
    We decompose $A$ in its Hermitian and anti-Hermitian components:
    \begin{align}
        &A = A_1 + i A_2,
        \\ & A_1 = \frac{A + A^\dag}{2},
        \\ & A_2 = \frac{A - A^\dag}{2i},
    \end{align}
    where $A_1$ and $A_2$ are both Hermitian, and therefore $iA_2$ is anti-Hermitian. We have that
    \begin{align}
        \abs{\Tr[A\rho]}^2 = \abs{\Tr[A_1\rho] + i \Tr[A_2\rho]}^2  = \Tr[A_1 \rho]^2 + \Tr[A_2 \rho]^2,
    \end{align}
    where we used the fact that $\abs{x+iy}^2 = x^2 + y^2$ for $x,y \in \mathbb{R}$.
    Moreover, for any hermitian operator $H$, the variance $\Tr[H^2 \rho] - \Tr[H\rho]^2$ is always non-negative, hence
    \begin{align}
         &\Tr[A_1 \rho]^2  \leq \Tr[A_1^2 \rho] = \frac{1}{4} \Tr[(A^2 + (A^\dag)^2 + \{A,A^\dag\} )\rho], 
         \\& \Tr[A_2 \rho]^2  \leq \Tr[A_2^2 \rho] = \frac{1}{4} \Tr[(-A^2 - (A^\dag)^2 + \{A,A^\dag\} )\rho]. 
    \end{align}
    Putting all together, we find that
    \begin{align}
        \abs{\Tr[A\rho]}^2 \leq \frac{1}{2}\Tr[\{A, A^\dag \}\rho].
    \end{align}
\end{proof}
\subsection{Monte Carlo methods}
\noindent The simulation algorithms presented in the paper build on Monte Carlo sampling techniques. Here, we present an overview of these techniques. In their essence, these sampling techniques consists in approximating the value of a given integral by sampling a random variable from an appropriately defined probability distribution and computing the value of an estimator function at that point, such that the expectation value of the estimator function over the probability distribution gives the value of this integral. Finally, we take the average of the $N$ such sampled values of the estimator function and the statistical sampling error depends on the range of this estimator function.

We say that a real-valued random variable $X$ is distributed with respect to a \emph{probability density function} (PDF) $f_X(x)$ if for all $a, b \in \mathbb{R}, a\leq b$, 
\begin{align}
    \Pr[a \leq X \leq b] = \int_a^b f_X(x) \mathrm dx.
\end{align}
Equivalently, the distribution of $X$ can be entirely described by the associated \emph{cumulative distribution function} (CDF) $F_X(x) = \int_{-\infty}^x f_X(t) \mathrm dt$:
\begin{align}
    \Pr[a \leq X \leq b] = F_X(b) - F_X(a).
\end{align}
Here, we will give some results that will be involved in establishing Monte Carlo methods for our expectation value estimation problem in the subsequent sections.
\begin{lemma}[Complex Chernoff--Hoeffding bound \cite{chabaud2020building}]
\label{applem:CH}
 Let $\lambda > 0$, let $n \geq 1$, let $z_1, \ldots, z_n$ be i.i.d. complex random variables from a probability density $\mathcal{D}$ over $\mathbb{R}$, and let $f : \mathbb{C} \mapsto \mathbb{R}$ such that $|f(z)| \leq M$, for $M > 0$ and all $z \in \mathbb{C}$. Then
\begin{align}
    \Pr\left[
\left| \frac{1}{n} \sum_{i=1}^{n} f(z_i) - \mathbb{E}_{z \sim D}[f(z)] \right|
\geq \lambda \right]
\leq 2 \exp\left[ - \frac{n \lambda^2}{2 M^2} \right].
\end{align}
\end{lemma}
\noindent The following two Lemmas construct the probability distributions and estimator functions for approximating values of two different types of integrals, which will enable the development of simulation algorithms in the subsequent sections.
\begin{lemma}[Integral approximation via shifted–Gaussian samples]
\label{lem:MC-Gaussian}
Let $a,b\in\R$, $A>0$, and $G:\R\to\C$ satisfy $\|G\|_{L_\infty}\le1$.  Define
\begin{align}
    I \;\coloneqq\;\int_{\R} e^{-A\,(a\,y + b)^2}\;G(y)\;\mathrm dy.
\end{align}
If we draw
\begin{align}
    Y \;\sim\;\mathcal N\!\left(\;\mu,\;\sigma^2\right),
\quad
\mu = -\frac b a,
\quad
\sigma^2 = \frac1{2\,A\,a^2},
\end{align}
then the importance‑sampling estimator
\begin{align}
    X
\;\coloneqq\;
\frac{e^{-A\,(a\,Y + b)^2}}{p(Y)}\,G(Y),
\end{align}
where
\begin{align}
    p(y)
\coloneqq\frac1{\sqrt{2\pi}\,\sigma}\,
\exp\!\Bigl(-\tfrac{(y-\mu)^2}{2\sigma^2}\Bigr)
\;=\;
\sqrt{\frac{A}{\pi}}\;a\;\exp\!\bigl[-A\,(a\,y + b)^2\bigr],
\end{align}
satisfies
\begin{align}
  \bbE[X] \;:= \bbE_{Y\sim p}[\,X(Y)\,] = \; I,
\quad
\bigl|X\bigr|
\;\le\;
\frac{\sqrt{\pi/A}}{\lvert a\rvert}\,.  
\end{align}
In particular, $X$ is unbiased and uniformly bounded in magnitude by $\sqrt{\pi/A}/|a|$.    
\end{lemma}

\begin{proof}
We choose $Y\sim\mathcal N(\mu,\sigma^2)$ with

\begin{align}
     \mu = -\frac b a,\qquad
     \sigma^2 = \frac1{2A\,a^2}.
\end{align}
Its density is
\begin{align}
     p(y)
     = \frac1{\sqrt{2\pi}\,\sigma}\,
       \exp\!\Bigl(-\tfrac{(y-\mu)^2}{2\sigma^2}\Bigr)
     = \sqrt{\frac{A}{\pi}}\;a\;\exp\!\bigl[-A\,(a\,y + b)^2\bigr].
\end{align}
Notice that $p(y)$ is  proportional to the Gaussian kernel in the integrand.
\begin{align}
     I
     = 
     \int_{\R} e^{-A\,(a\,y + b)^2}\,G(y)\,\mathrm dy
     = 
     \int_{\R} \underbrace{\frac{e^{-A\,(a\,y + b)^2}}{p(y)}}_{=\,X(y)}\,G(y)\;p(y)\,\mathrm dy
     = \bbE_{Y\sim p}[\,X(Y)\,],
\end{align}
one obtains the unbiased estimator
\begin{align}
     X \;=\;
     \frac{e^{-A\,(a\,Y + b)^2}}{p(Y)}\,G(Y).
\end{align}
Since by construction
   $\;p(Y)=\sqrt{\tfrac{A}{\pi}}\,a\;e^{-A\,(aY+b)^2},$
   the weight simplifies to

\begin{align}
     \frac{e^{-A\,(aY + b)^2}}{p(Y)}
     =
     \frac{1}{\sqrt{\tfrac{A}{\pi}}\;a}
     =
     \frac{\sqrt{\pi/A}}{a}.
\end{align}
Hence
\begin{align}
     X
     = 
     \frac{\sqrt{\pi/A}}{a}\;G(Y),
     \qquad
     \bigl|X\bigr|
     \le
     \frac{\sqrt{\pi/A}}{\lvert a\rvert}\,\|G\|_{L_\infty}
     \;\le\;
     \frac{\sqrt{\pi/A}}{\lvert a\rvert}.
\end{align}
Thus $X$ is an unbiased estimator of $I$, with the simple deterministic bound

\begin{align}
|X|\;\le\;\frac{\sqrt{\pi/A}}{|a|},
\end{align}
as claimed. 
\end{proof}

\begin{lemma}[Integral approximation via Gamma-distributed samples]
\label{lem:MC-gamma}
Let $a,b \in \mathbb{R}$ and $A,B \in \mathbb{R}^+$. Let $G: \mathbb{R} \rightarrow \mathbb{C}$  such that $\norm{G}_{L_\infty} \leq 1$.   
Consider the integral
\begin{align}
    I = \int_\mathbb{R}  \frac{e^{-A(ay +b)^2}}{B\sqrt{\abs{ay+b}}} G(y) \mathrm dy.
\end{align}
Then, using only draws from the Gamma distribution $\mathrm{Gamma}(\frac{1}{4},{A})$, one may generate a random variable $X$ such that
\begin{align}
    \bbE [X] = I \qquad \text{and} \qquad
    \abs{X} \leq \frac{\Gamma(\frac{1}{4})A^{-\frac{1}{4}}}{aB}.
\end{align}
\end{lemma}
\begin{proof}
We start by performing the change of variable defined by the linear map $y \mapsto t = ay + b$. Then $\mathrm dy = \mathrm dt/a$ and  
\begin{align}
    &I = \int_{\mathbb{R}} \frac{e^{-At^2}}{aB\sqrt{\abs{t}}} G\left(\frac{t-b}{a}\right) {\mathrm dt} = I_- + I_+,
    \\&I_- = \int_{-\infty}^0 \frac{e^{-At^2}}{aB\sqrt{\abs{t}}} G\left(\frac{t-b}{a}\right) {\mathrm dt} ,
    \\& I_+ = \int_0^{\infty} \frac{e^{-At^2}}{aB\sqrt{\abs{t}}} G\left(\frac{t-b}{a}\right) {\mathrm dt}
\end{align}
We further perform the changes of variable $t\mapsto -\sqrt{u}$ and $t \mapsto \sqrt{u}$ respectively in the integrals $I_- $ and $I_+$:
\begin{align}
  &I_- =  \int_0^\infty \frac{e^{-Au}}{2aBu^{\frac{3}{4}}} G\left(\frac{-\sqrt{u} - b}{a}\right) \mathrm du
  \\& I_+ = \int_0^\infty \frac{e^{-Au}}{2aBu^{\frac{3}{4}}} G\left(\frac{\sqrt{u} - b}{a}\right) \mathrm du
\end{align}
Recall the $\mathrm{Gamma}(\alpha,\beta)$ probability density function with shape $\alpha$ and rate $\beta$ is
\begin{align}
    f(u)=\frac{\beta^\alpha}{\Gamma(\alpha)}\,u^{\alpha-1}e^{-\beta u}.
\end{align}
     For $\alpha=\frac{1}{4},\;\beta=A$, this is
\begin{align}
       f(u)
   = \frac{A^{1/4}}{\Gamma\left(\frac{1}{4}\right)}\;u^{-3/4}\,e^{-A u}.
\end{align}
Hence the integral $I$ can be rewritten as
\begin{align}
    I = \frac{1}{2}\int_0^\infty f(u) H_-(u) \mathrm du + \frac{1}{2}\int_0^\infty f(u) H_+(u) \mathrm du,
\end{align}
where $H_{\pm}(u) = \frac{\Gamma\left(\frac{1}{4}\right)}{aA^{1/4}B} G\left(\frac{\pm\sqrt{u} - b}{a}\right) $. 
Thus, given $U\sim \mathrm{Gamma}(\frac{1}{4},A) $, we define the derived random variable ``$\mathrm{sgn}$'' which equals $-$ with probability $1/2$ and $+$ with probability $1/2$. Lets denote this 50-50 probability distribution for ``$\mathrm{sgn}$'' by $p_{\mathrm{flip}}$. Then
\begin{align}
    \mathbb{E} [X] = \bbE_{U\sim \mathrm{Gamma}(\frac{1}{4},A), \mathrm{sgn} \sim p_{\mathrm{flip}}} \left[H_{\mathrm{sgn}}(U)\right]= I.
\end{align}
Moreover, $\abs{X} \leq \norm{H}_{L_\infty} \leq  \frac{\Gamma\left(\frac{1}{4}\right)}{aA^{1/4}B} \norm{G}_{L_\infty} \leq  \frac{\Gamma\left(\frac{1}{4}\right)}{aA^{1/4}B} $.
\end{proof}

\section{Bosonic quantum information}\label{appsec:bosonic_prelims}
This paper deals with bosonic or continuous variable (CV) quantum information and this section aims to provide a brief overview of the relevant background information. We refer the reader to \cite{NielsenChuang} for further background on quantum information theory and to \cite{Braunstein2005,ferraro2005gaussian,Weedbrook2012,BUCCO} for detailed CV quantum information material. 

\smallskip

\noindent\textbf{Bosonic modes and quadratures.}
In CV quantum information, a mode refers to a degree of freedom associated with a specific quantum field of a CV quantum system, such as a single spatial or frequency mode of light, and is the equivalent of qubit in the CV regime. In this paper, $m \in \N^*$ denotes the number of modes in the system. Hereafter, $\{\ket{n} \}_{n\in\N}$ denote the single-mode Fock basis, with $\ket0$ being the vacuum state, and $\hat{a}$ and $\hat{a}^\dagger$ refer to the single-mode annihilation and creation operator, respectively, satisfying $[\hat{a},\hat{a}^\dagger] = \I$. These are related to the position and momentum quadrature operators as
\begin{equation}
    \hat{q} = \hat{a} + \hat{a}^\dagger, \hspace{5mm}\hat{p} = -i(\hat{a} - \hat{a}^\dagger),
\end{equation} 
with the convention $\hbar = 2$. Furthermore, $\hat{q}$ and $\hat{p}$ satisfy the commutation relation $[\hat{q},\hat{p}] = 2i\I$. The particle number operator is given as $\hat{n} = \hat{a}^\dagger \hat{a}$, whereas the energy operator is denoted by
\begin{equation}\label{appeq:energy_op}
    \hat N = \sum_{j=1}^m \hat N_j = \sum_{j=1}^m \frac{\hat q_j^2 + \hat p_j^2}{2}.
\end{equation}

\smallskip

\noindent\textbf{Gaussian unitaries and states.}
Product of unitary operations generated by Hamiltonians that are quadratic in the quadrature operators of the modes are called Gaussian unitary operations, and states produced by applying a Gaussian unitary operation to the vacuum state are Gaussian states. The action of an $m$-mode Gaussian unitary operation $\hat{G}$ on the vector of quadratures $\boldsymbol{\Gamma}= [\hat{q}_1,\dots,\hat{q}_m,\hat{p}_1,\dots,\hat{p}_m]$ is given by
\begin{equation}\label{eqn:symplectic_transform_quadrature}
    \hat{G}^\dagger \boldsymbol{\Gamma} \hat{G} = S \boldsymbol{\Gamma} + \boldsymbol{d},
\end{equation}
where $S$ is a $2m \times 2m$ symplectic matrix and $\boldsymbol{d} \in \mathbb{R}^{2m}$ is a displacement vector, associated to the Gaussian unitary gate $\hat G$. Single-mode displacement operator can be written in terms of creation operators as  
\begin{equation}
    \hat{D}(\alpha) = e^{\alpha \hat{a}^\dagger - \alpha^* a },
\end{equation}
where $\alpha \in \C$, or in terms of the position and momentum operators as
\begin{equation}
    \hat{D}(q, p) =e^{i p \hat q - i q \hat p} \, .
\end{equation}
where $q_i, p_i \in \R$, with $q_i = \Re(\alpha)$ and $p_i = \Im(\alpha)$ . Note that with our convention, we are writing the multi-mode displacement operator as 
\begin{equation}
    \hat D(\br) = \exp(i \hat R^T \Omega \br ), 
\end{equation}
with $\br = \{q_1,p_1,\dots,q_m,p_m\}$, $\alpha_j = (q_j + i p_j)/2, \forall j$, $\hat R = [\hat q_1,\hat p_1,\dots,\hat q_m,\hat p_m]$ and 
where
\begin{equation}
    \Omega =
    \bigoplus_{j=1}^m\begin{bmatrix}
            0 & 1 \\ -1 &0
        \end{bmatrix}.
\end{equation}
Finally, the single-mode states produced by the action of the displacement operator on the vacuum are called coherent states, $\ket{\alpha} = \hat D(\alpha)\ket 0$.

At this point, we note an important property of the symplectic matrix $S$ associated to a Gaussian unitary $\hat G$. Suppose we write $S$ in the block diagonal form in the quadrature basis $(\hat q_1,\hat q_2,\dots,\hat q_m, \hat p_1,\dots,\hat p_m)$ as
\begin{equation}
S = \begin{bmatrix}
    S_q && S_{qp} \\
    S_{pq} && S_p
\end{bmatrix},
\end{equation}
then for a symplectic matrix we have that
\begin{equation}
    S^{-1} = - \Omega S^T \Omega = \begin{bmatrix}
        S_p^T && - S_{qp}^T \\ - S_{pq}^T && S_q^T
    \end{bmatrix}.
\end{equation}
Therefore, if, for a Gaussian unitary $G_j$, we define symplectic coherence with respect to the first mode \cite{upreti2025interplay} as
\begin{equation}
    \sigma_j = |(S_j)_{q_1,p_1}|,
\end{equation}
then we have that
\begin{equation}\label{appeq:sigma_sigma_inv}
    (\sigma^{-1})_j = |(S^{-1}_j)_{q_1,p_1}| = |-(S_j)_{q_1,p_1}| = \sigma_j.
\end{equation}
We will use this result later.
\smallskip

\noindent\textbf{Displacement operator basis and characteristic functions.}
The set of displacement operators forms an operator basis, and therefore for any operator $\hat{O}$ living in the $m$-mode Hilbert space $\mathcal{H}^{\otimes m}$, we can write
\begin{equation}
    \hat{O} = \frac{1}{\pi^m} \int_{\bm\alpha \in \C^m} \mathrm d^{2m}\bm\alpha \chi_{\hat O} (\bm\alpha) \hat{D}^\dagger(\bm\alpha), \label{eq:fw-exp}
\end{equation}
where $\hat{D}(\bm\alpha) = \hat{D}(\bm\alpha_1)\otimes\dots\otimes\hat D(\bm\alpha_m)$ is the multi-mode displacement operator, and $\chi_{\hat O} (\bm\alpha)$ is referred to as the \textit{characteristic function} of the operator $\hat O$, given by
\begin{equation}
    \chi_{\hat O} (\bm\alpha) = \Tr[\rho \hat{D} (\bm\alpha)]. \label{eq:char-fun}
\end{equation}
The characteristic function of the single-mode coherent state $\ket \beta$ is given by
\begin{eqnarray}
    \chi_{\ket{\beta}\bra{\beta}}(\alpha) &=& \Tr[\ket{\beta}\bra{\beta}\hat D(\alpha)] = \bra{\beta} \hat D(\alpha) \ket{\beta} = \exp(i \Im(\alpha \beta^*)) \braket{\beta}{\alpha+\beta} \nonumber \\ &=& \exp(i \Im(\alpha \beta^*)) \exp\left(-\frac12 |\beta|^2 -\frac12 |\alpha +\beta|^2 + \beta^* (\alpha + \beta) \right).
\end{eqnarray}
Further, using the characteristic function,
\begin{equation}
    \Tr[\rho\, \hat{O}] = \frac{1}{\pi^m}\int_{\bm\alpha \in \C^m} \mathrm d^{2m}\bm\alpha \chi_\rho^*(\bm\alpha)\chi_{\hat O}(\bm\alpha),
\end{equation}
where the identity follows from combining Eqs.~\ref{eq:fw-exp} and~\ref{eq:char-fun} and interchanging the order of the trace and the integral by linearity.
Finally the magnitude of the characteristic function of a quantum state $\rho$,
\begin{equation}
    \left|\Tr[\rho \hat D(\br)]\right| \leq 1.
\end{equation}
$\forall \br \in \mathbb{R}^{2m}$. This follows since
\begin{equation}
    \left| \Tr[\rho \hat D(\br)]\right| \leq \|\rho\|_1 \|\hat D(\br)\|_\infty \leq 1 \times 1 = 1.
\end{equation}
The first inequality follows from the H\"older's inequality (Lemma \ref{lem:holder}) whereas the second inequality follows from the fact that for a quantum state $\|\rho\|_1 = 1$ and that $\forall \br \in \mathbb{R}^{2m}$, $\|\hat D(\br)\|_{\infty} = 1$ since the displacement operator is a unitary operator.

\smallskip

\noindent\textbf{Fourier norms and operator inequalities.}
In order to study the properties of (noisy) bosonic systems, we introduce the following notion, which serves as a metric for the complexity of bosonic operators and directly reflects the cost of classically simulating expectation values.
\begin{definition}[Fourier $p$-norms]
\label{def:fourier-p}
Let $\hat O$ be an operator whose characteristic function $\chi_{\hat O}$ is in $L^p$ (i.e.\ the $p$-th power of the absolute value of $\chi_{\hat O}$ is Lebesgue integrable).
Then the Fourier $p$-norm of $\hat{O}$ is the defined as the $p$-norm of $\chi_{\hat O}/\pi^m$:
\begin{align}
    \norm{\hat O}_{\mathbb{F}, p} \coloneqq \left(\int_{\mathbb{C}^{m}} \mathrm d^{2m}\boldsymbol{\alpha} \left\lvert \frac{\chi_{\hat{O}}(\boldsymbol{\alpha})}{\pi^m} \right\rvert^{p} \right)^{\frac{1}{p}}.
\end{align}
Moreover, if $\hat{O}$ has a discrete spectrum in the displacement operator basis, i.e. $\hat{O} \coloneqq \sum_{i \in \mathbb{N}} c_i \hat D(\boldsymbol{\alpha}_i)$, then the Fourier 1-norm of $\hat{O}$ is defined as follows:
\begin{align}
    \norm{O}_{\mathbb{F}, p} \coloneqq \left(\sum_{i \in \mathbb{N}} \abs{c_i}^p\right)^{\frac{1}{p}}. 
\end{align}
\end{definition}
We also prove the following relation between the Fourier 1-norm and the operator norm (i.e.\ the Schatten $\infty$-norm), which can be regarded as a quantum extension of the Hausdorff--Young inequality.
\begin{lemma}[Quantum Hausdorff–Young inequality]
\label{lem:haus-young}
Let $\hat{O}$ be an operator with finite Fourier 1-norm. Then it holds that
\begin{align}
    \norm{\hat{O}}_\infty \leq \fnormone{\hat O}.
\end{align}
\end{lemma}
\begin{proof}
The result readily follows from the definition of Fourier 1-norm. If $\chi_{\hat O}$ is $L^1$, we have that
\begin{align}
    \norm{\hat O}_\infty = &\sup_{\rho: \norm{\rho}_1=1} \abs{\Tr[\hat O\rho]}
    =  \sup_{\rho: \norm{\rho}_1=1}  \left|\frac{1}{\pi^m}  \int_{\boldsymbol{\alpha} \in \mathbb{C}^m} \mathrm d^{2m}\boldsymbol{\alpha} \chi_{\hat O}(\boldsymbol{\alpha}) \Tr[\hat D(-\boldsymbol{\alpha})\rho] \right| \\\leq&\left( \frac{1}{\pi^m} \int_{\boldsymbol{\alpha} \in \mathbb{C}^m} \mathrm d^{2m}\boldsymbol{\alpha} \left |\chi_{\hat O}\right|  \right)\sup_{\substack{\rho: \norm{\rho}_1=1,\\\boldsymbol{\alpha} \in \mathbb{C}^m}} \abs{\Tr[\hat D(-\boldsymbol{\alpha})\rho]} \leq \fnormone{\hat{O}}\sup_{\substack{\boldsymbol{\alpha} \in \mathbb{C}^m}}\norm{\hat D(\boldsymbol{\alpha})}_\infty = \fnormone{\hat{O}},
\end{align}
where we applied Minkowski's inequality in the second-to-last step and Hölder's inequality in the last step, and we noted that $\norm{\hat D(\boldsymbol{\alpha})}_\infty = 1$
Similarly, if $\hat{O}$ is a weighted sum of displacement operators, we have that
\begin{align}
    \norm{\hat O}_\infty = \sup_{\rho: \norm{\rho}_1=1} \abs{\Tr[\hat O\rho]}
    =  \sup_{\rho: \norm{\rho}_1=1} \left|\sum_{i \in \mathbb{N}} c_i \Tr[\hat D(\boldsymbol{\alpha}_i)\rho] \right| \leq &\left( \sum_{i\in \mathbb{N}} \abs{c_i}\right)  \sup_{\substack{\rho: \norm{\rho}_1=1 \\ \boldsymbol{\alpha} \in \mathbb{C}^m}}  \abs{\Tr[\hat D(\boldsymbol{\alpha})\rho]}
    \\ \leq&\fnormone{\hat{O}}\sup_{\substack{\boldsymbol{\alpha} \in \mathbb{C}^m}}\norm{\hat D(\boldsymbol{\alpha})}_\infty = \fnormone{\hat{O}},
\end{align}
where we applied again Minkowski's inequality in the second-to-last step and Hölder's inequality in the last step, and we noted again that $\norm{\hat D(\boldsymbol{\alpha})}_\infty = 1$.
\end{proof}

\smallskip

\noindent\textbf{Non-Gaussian operations.}
Non-Gaussian operations are necessary for enabling quantum advantage since Gaussian gates acting on Gaussian states and followed by Gaussian measurements can be classically simulated efficiently \cite{Bartlett2002}. One prominent example of a non-Gaussian gate is the cubic phase gate $e^{i\gamma\hat{q}^3}$ \cite{Gottesman2001}, whose action on the quadratures is given by \cite{Budinger2024}
\begin{equation}\label{appeq:cubic_action}
    e^{-i\gamma\hat{q}^3} \hat{q} e^{i\gamma\hat{q}^3}   = \hat{q}, \hspace{5mm} e^{-i\gamma\hat{q}^3}  \hat{p}  e^{i\gamma\hat{q}^3}    = \hat{p} + 6\gamma \hat{q}^2.
\end{equation}
A standard model of universal CV quantum computation (CVQC) is defined by a vacuum state input into a circuit with Gaussian unitaries and cubic phase gates \cite{lloyd1999quantum,Sefi2011,arzani2025}. 

\smallskip

\noindent\textbf{Thermal loss channel.}
Thermal loss arises from the interaction of a system with the thermal state of a finite-temperature environment. We model its effect using a quantum channel $\Lambda_{\bar n,\eta}(\cdot)$, where the action on a single-mode density matrix $\rho$ is given by:
\begin{equation}
    \Lambda_{\bar n,\eta}(\rho) = \Tr_E[\hat B_\eta \rho \otimes \rho_{\bar n} \hat B_\eta^\dagger],
\end{equation}
where $\rho_{\bar n}$ is the thermal state of the environment with mean photon number $\bar{n}$, given by
\begin{align}
        \rho_{\bar{n}} = &\frac{1}{1 + \bar{n}} \sum_{n=0}^\infty \left(\frac{\bar{n}}{\bar{n}+1} \right)^n \ketbra{n}{n} \nonumber \\ = &\frac{1}{\pi} \int_{\mathbb{C}} \mathrm d^2\alpha \, \exp\left(- \left(\bar{n} + \frac{1}{2}\right){|\alpha|^2}\right)\hat D(-\alpha),
\end{align}
whereas $\hat B_\eta$ is the two-mode Gaussian unitary gate beamsplitter whose associated displacement vector is zero and whose associated symplectic matrix is given by
\begin{equation}
    \begin{bmatrix}
        \sqrt{\eta} && \sqrt{1 - \eta} && 0 && 0 \\
        -\sqrt{1 - \eta} && \sqrt{\eta} && 0 && 0 \\
        0 && 0 && \sqrt{\eta} && \sqrt{1 - \eta} \\
        0 && 0 && -\sqrt{1 - \eta} && \sqrt{\eta}
    \end{bmatrix}.
\end{equation}
Note that the thermal loss channel for $\bar n = 0$ is known as the photon loss or pure loss channel.

\subsection{Evolution of displacement operators under bosonic channels}
\noindent\textbf{Evolution of displacement operators under Gaussian gates.}
The evolution of the displacement operator under the Gaussian channel ($\mathcal G_j (\cdot) = \hat G_j(\cdot) \hat G_j^\dagger$) is defined via its action on the displacement operator basis as follows \cite{BUCCO}
\begin{align}
    &\calG_j (\hat D(\br)) = e^{-i \bd_j^\top \Omega S_j \br}\hat D(S_j \br),
    \\& \calG_j^* (\hat D(\br)) = e^{i \bd_j^\top \Omega \br}\hat D(S_j^{-1} \br).
\end{align}

 \smallskip
 
\noindent\textbf{Evolution of displacement operator under the effect of cubic phase gate.}
The following result gives the evolution of the displacement operator under the effect of the cubic phase gate:
\begin{lemma}[Displacement operators under the action of a cubic phase gate]\label{applem:cubic_evo}
    Given the $m$-mode displacement operator $\hat D(\br)$ for $\br = \{r_1,\dots,r_m\} \in \mathbb{R}^{2m}$, with $r_i = \{q_i,p_i\}, \forall i \in \{1,\dots,m\}$, the evolution of the displacement operator under the action of the cubic phase gate $e^{i\gamma\hat q_1^3}$ is given by
    \begin{equation}
            e^{-i\gamma\hat q_1^3} \hat D(\br) e^{i\gamma \hat q_1^3} = e^{-i\frac{\pi}{4}}\sqrt{\frac{1}{24\pi\gamma q_1}} \exp\left(-4i \gamma q_1^3\right) \int \mathrm dp_1' \exp\left(i \frac{(p_1 - p_1')^2}{24  \gamma q_1} \right) \hat{D}(q_1,p_1') \otimes \hat{D}(r_2)\otimes \dots \otimes \hat{D}(r_m).
    \end{equation}
\end{lemma}
\begin{proof}
     We first focus on the effect of the cubic phase gate on a single mode displacement operator $\hat D(\alpha)$, with $r_1 = q_1 + i p_1$. Expanding $e^{-i\gamma\hat{q}^3} \hat{D} (r) e^{i\gamma\hat{q}^3}$ in the displacement operator basis,
\begin{equation}\label{eqn:displace_evolv_cubic}
    e^{-i\gamma\hat{q}^3} \hat{D} (r) e^{i\gamma\hat{q}^3} = \frac{1}{\pi}\int_{\mathbb{R}^2} \Tr[e^{-i\gamma\hat{q}^3} \hat{D}(r) e^{i\gamma\hat{q}^3} \hat{D} (\beta)] \hat{D}^\dagger(\beta) \mathrm d^2\beta \, ,
\end{equation}
where 
\begin{eqnarray}
    \beta := q_\beta + i p_\beta \, .
\end{eqnarray}
From \cite[Eq.~(54)]{frigerio2024},
\begin{equation}
    \Tr[e^{-i\gamma\hat{q}^3} \hat{D} (r) e^{i\gamma\hat{q}^3} \hat{D} (\beta)] = e^{-i \frac{\pi}{4}} \delta(q_1,-q_\beta)\sqrt{\frac{\pi}{24 \gamma q_1}} \exp\left(i \frac{(p_1 + p_\beta)^2}{24  \gamma q_1} - 4i \gamma q_1^3 \right) \, .
\end{equation}
Therefore,
\begin{eqnarray}
    e^{-i\gamma\hat{q}^3} \hat{D} (r) e^{i\gamma\hat{q}^3} &=& e^{-i\frac{\pi}{4}}\sqrt{\frac{1}{24\pi\gamma q_1}} \exp\left(-4i \gamma q_1^3\right) \int \mathrm dp_\beta \exp\left(i \frac{(p_1 + p_\beta)^2}{24  \gamma q_1} \right) \hat{D}(q_1,-p_\beta) \nonumber \\
    &=& e^{-i\frac{\pi}{4}}\sqrt{\frac{1}{24\pi\gamma q_1}} \exp\left(-4i \gamma q_1^3\right) \int \mathrm dp_\beta \exp\left(i \frac{(p_1 - p_\beta)^2}{24  \gamma q_1} \right) \hat{D}(q_1,p_\beta).
\end{eqnarray}

Now, for the $m$-mode displacement operator $\hat{D} (r_1) \otimes \dots \otimes \hat D(r_m)$:
\begin{eqnarray}
    e^{-i\gamma\hat{q}_1^3} \hat{D} (r_1) \otimes \dots \otimes \hat D(r_m) e^{i\gamma\hat{q}_1^3} = e^{-i\gamma\hat{q}_1^3} \hat{D} (r_1) e^{i\gamma\hat{q}_1^3} \otimes \dots \otimes \hat D (r_m) && \nonumber \\ &&  \hspace{-70mm} =e^{-i\frac{\pi}{4}}\sqrt{\frac{1}{24\pi\gamma q_1}} \exp\left(-4i \gamma q_1^3\right) \int \mathrm dp_1' \exp\left(i \frac{(p_1 - p_1')^2}{24  \gamma q_1} \right) \hat{D}(q_1,p_1') \otimes \hat{D}(r_2)\otimes \dots \otimes \hat{D}(r_m). \nonumber \\
\end{eqnarray}
\end{proof}

\smallskip

\noindent\textbf{Evolution of displacement operator under the effect of thermal loss.}
The following result gives the evolution of the displacement operator under the thermal loss channel and is a more general version of \cite[Appendix C]{Chabaud2024phasespace} where the evolution of the displacement operator under photon loss was computed. 
\begin{lemma}[Evolution of characteristic function under thermal loss channel] \label{applem:thermal_evo}
Given a thermal loss channel $\Lambda_{\bar n,\eta}$ with efficiency $\eta$ and thermal photon number $\bar n$, the evolution of the single-mode displacement operator $\hat D(r_1)$ with $r_1 = \{q_1,p_1\}$ under the effect of the thermal noise channel and its adjoint is given by:
\begin{eqnarray}
    \Lambda_{\bar n,\eta} (\hat D(r_1))&=& \frac{1}{\eta} \exp\left(-(1+2\bar{n})\left(\frac{1}{\eta} - 1\right)\frac{q_1^2+p_1^2}{2}\right)\hat{D}(r_1/\sqrt{\eta}), \nonumber \\
    \Lambda_{\bar n,\eta}^* (\hat D(r_1))&=&  \exp\left(-(1+2\bar{n})\left(1 - \eta\right)\frac{q_1^2+p_1^2}{2}\right)\hat D(r_1\sqrt{\eta}).
\end{eqnarray}
\end{lemma}
\begin{proof}
    First, we describe the effect of thermal loss channel on coherent states:
\begin{eqnarray}
    \Lambda_{\bar n,\eta}(\ket{\alpha}\bra{\alpha}) &=& \Tr_E[\hat B_\eta \ket{\alpha}\bra{\alpha} \otimes \lambda \hat B_\eta^\dagger] \nonumber \\
    &=& \int \mathrm d^2\beta \frac{e^{-|\beta|^2/\bar{n}}}{\pi \bar{n}} \Tr_E [\hat B_\eta \ket{\alpha}\bra{\alpha} \otimes \ket{\beta}\bra{\beta}_E \hat B_\eta^\dagger] \nonumber \\
    &=& \int \mathrm d^2\beta \frac{e^{-|\beta|^2/\bar{n}}}{\pi \bar{n}} \nonumber \\ && \hspace{10mm} \Tr_E \left[\ket{\alpha\sqrt{\eta} - \beta\sqrt{1-\eta}}\bra{\alpha\sqrt{\eta} - \beta\sqrt{1-\eta}} \otimes \ket{\alpha\sqrt{1-\eta}+\beta\sqrt{\eta}}\bra{\alpha\sqrt{1-\eta}+\beta\sqrt{\eta}}_E\right] \nonumber \\
    &=& \int \mathrm d^2\beta \frac{e^{-|\beta|^2/\bar{n}}}{\pi \bar{n}} \ket{\alpha\sqrt{\eta} - \beta\sqrt{1-\eta}}\bra{\alpha\sqrt{\eta} - \beta\sqrt{1-\eta}},
\end{eqnarray}
\noindent where we have used the fact that the thermal state $\lambda$ can be represented in the coherent state basis as
\begin{equation}
    \lambda = \int \mathrm d^2 \beta \frac{e^{-|\beta|^2/\bar{n}}}{\pi \bar{n}} \ket \beta \bra \beta
\end{equation}
and that a beamsplitter with efficiency $\eta$ acts on a tensor product of coherent states as
\begin{equation}
    \hat B_\eta (\ket{\alpha} \otimes \ket{\beta}) = \ket{\alpha \sqrt{\eta} - \beta \sqrt{1-\eta}}\otimes \ket{\alpha\sqrt{1-\eta}+\beta\sqrt{
    \eta
    }}.
\end{equation}
\noindent Therefore,
\begin{eqnarray}
    \Lambda_{\bar n,\eta}(\hat D(r_1)) &=& \frac{1}{\pi}e^{|r_1|^2/2}\int \mathrm d^2\alpha e^{r_1\alpha^* - r_1^* \alpha} \Lambda_{\bar n,\eta} (\ket{\alpha}\bra{\alpha}) \nonumber \\
    &=& \frac{1}{\pi}e^{r_1^2/2} \int \mathrm d^2\alpha e^{r_1\alpha^* - r_1^* \alpha} \int \mathrm d^2\beta \frac{e^{-|\beta|^2/\bar{n}}}{\pi \bar{n}} \ket{\alpha\sqrt{\eta} - \beta\sqrt{1-\eta}}\bra{\alpha\sqrt{\eta} - \beta\sqrt{1-\eta}}. \nonumber \\
\end{eqnarray}
Writing $\ket{\alpha\sqrt{\eta} - \beta\sqrt{1-\eta}} = e^{i\phi(\alpha,\beta)}\hat{D}(-\beta\sqrt{1-\eta})\ket{\alpha\sqrt{\eta}}$ and noting that phases cancel out we obtain
\begin{eqnarray}\label{eqn_thermal_transform}
\Lambda_{\bar n,\eta}(D(r_1)) &=& \int \mathrm d^2\beta \frac{e^{-|\beta|^2/\bar{n}}}{\pi \bar{n}} \hat{D}(-\beta\sqrt{1-\eta}) \left(\frac{1}{\pi}e^{|r_1|^2/2}  \int \mathrm d^2\alpha  e^{r_1 \alpha^* - r_1^* \alpha} \ket{\alpha\sqrt{\eta}}\bra{\alpha\sqrt{\eta}}\right)  \hat{D}^\dagger (-\beta\sqrt{1-\eta}) \nonumber \\
&=& \int \mathrm d^2\beta \frac{e^{-|\beta|^2/\bar{n}}}{\pi \bar{n}} \hat{D}(-\beta\sqrt{1-\eta}) \frac{1}{\eta} \exp\left(-\left(\frac{1}{\eta} - 1\right)\frac{q_1^2 + p_1^2}{2}\right)\hat{D}(r_1/\sqrt{\eta})\hat{D}^\dagger(-\beta\sqrt{1-\eta}) \nonumber \\
&=& \frac{1}{\eta} \exp\left(-\left(\frac{1}{\eta} - 1\right)\frac{q_1^2 + p_1^2}{2}\right) \int \mathrm d^2\beta \frac{e^{-|\beta|^2/\bar{n}}}{\pi \bar{n}} \hat{D}(-\beta\sqrt{1-\eta}) \hat{D}(r_1/\sqrt{\eta})\hat{D}^\dagger(-\beta\sqrt{1-\eta}) \nonumber \\
&=& \frac{1}{\eta} \exp\left(-\left(\frac{1}{\eta} - 1\right)\frac{q_1^2 + p_1^2}{2}\right) \nonumber \\ && \hspace{10mm} \times \int \mathrm d^2\beta \frac{e^{-|\beta|^2/\bar{n}}}{\pi \bar{n}} \exp\left(i\Im \left(-\beta\sqrt{1-\eta}\frac{(q_1 - i p_1)}{\sqrt{\eta}}\right)\right)\hat{D}\left(-\beta\sqrt{1-\eta}+\frac{q_1 +i p_1}{\sqrt{\eta}}\right)\hat{D}^\dagger(-\beta\sqrt{1-\eta}) \nonumber \\
&=& \frac{1}{\eta} \exp\left(-\left(\frac{1}{\eta} - 1\right)\frac{q_1^2 + p_1^2}{2}\right) \hat{D}(r_1/\sqrt{\eta}) \nonumber \\ && \hspace{0mm} \times \int \mathrm d^2\beta \frac{e^{-|\beta|^2/\bar{n}}}{\pi \bar{n}} \exp\left(i\Im \left(-\beta\sqrt{1-\eta}\frac{(q_1-ip_1))}{\sqrt{\eta}}\right)\right) \exp\left(i\Im\left(\left(-\beta\sqrt{1-\eta}+\frac{q_1 + i p_1}{\sqrt{\eta}}\right)(\beta^*\sqrt{1-\eta})\right)\right) \nonumber \\
&=& \frac{1}{\eta} \exp\left(-\left(\frac{1}{\eta} - 1\right)\frac{q_1^2+p_1^2}{2}\right) \hat{D}(r_1/\sqrt{\eta}) \int \mathrm d^2\beta \frac{e^{-|\beta|^2/\bar{n}}}{\pi \bar{n}} \exp\left(2i\Im\left(\frac{\sqrt{1-\eta}}{\sqrt{\eta}} \beta^* (q_1+ip_1))\right)\right) \nonumber \\
&:=& \frac{1}{\eta} \exp\left(-\left(\frac{1}{\eta} - 1\right)\frac{q_1^2 + p_1^2}{2}\right) \hat{D}(r_1/\sqrt{\eta}) f(\eta,\bar{n},q_1,p_1),
\end{eqnarray}
where the second line follows from \cite[Eq.\ 55]{Chabaud2024phasespace}. To find the value of $f(\eta,\bar{n},\xi)$, substituting $\beta^* = x + i y$ and $a = 2\sqrt{\frac{1-\eta}{\eta}}(q_1+ip_1)=a_R+ia_I$:
\begin{eqnarray}
    f(\eta,\bar{n},\xi) &=& \frac{1}{\pi \bar{n}} \int \mathrm dx \mathrm dy \exp\left(- \frac{x^2+y^2}{\bar{n}}\right) \exp\left(i\Im((a_R + i a_I)(x+iy))\right) \nonumber \\
    &=& \frac{1}{\pi \bar{n}} \int \mathrm dx \exp\left(-\left(\frac{x^2}{ \bar{n}} - i a_I x\right)\right) \int \mathrm dy \exp\left(-\left(\frac{y^2}{\bar{n}} - i a_R y\right)\right) \nonumber \\
    &=& \frac{1}{\pi \bar{n}} \exp\left(\frac{1}{\bar{n}}\frac{-\bar{n}^2 a_I^2}{4}\right) \sqrt{\pi\bar{n}} \exp\left(\frac{1}{\bar{n}}\frac{-\bar{n}^2 a_R^2}{4}\right)\sqrt{\pi\bar{n}} \nonumber \\
    &=& \exp\left(-\frac{\bar{n}}{4}|a|^2\right) \nonumber \\
    &=& \exp\left(-\bar{n}\left(\frac{1}{\eta} - 1\right)(q_1^2 + p_1^2)\right).
\end{eqnarray}
\noindent Putting the value of $f(\eta,\bar n,\xi)$ back into equation \ref{eqn_thermal_transform},
\begin{equation}
    \Lambda_{\bar n,\eta}(\hat D(r_1)) = \frac{1}{\eta} \exp\left(-(1+2\bar{n})\left(\frac{1}{\eta} - 1\right)\frac{q_1^2+p_1^2}{2}\right)\hat{D}(r_1/\sqrt{\eta}).
\end{equation}
Now,
\begin{eqnarray}
    \Tr[\Lambda_{\bar n,\eta}(\hat{D}(r_1))\hat{D}(\alpha)] &=&  \frac{1}{\eta} \exp\left(-(1+2\bar{n})\left(\frac{1}{\eta} - 1\right)\frac{q_1^2+p_1^2}{2}\right) \Tr[\hat D(r_1/\sqrt{\eta})\hat{D}(\alpha)] \nonumber \\
    &=& \frac{1}{\eta} \exp\left(-(1+2\bar{n})\left(\frac{1}{\eta} - 1\right)\frac{q_1^2+p_1^2}{2}\right) \pi \delta^2(r_1/\sqrt{\eta} + \alpha) \nonumber \\
    &=& \exp\left(-(1+2\bar{n})\left(\frac{1}{\eta} - 1\right)\frac{q_1^2+p_1^2}{2}\right) \pi \delta^2(r_1 + \alpha\sqrt{\eta}) \nonumber \\
    &=& \exp\left(-(1+2\bar{n})\left(1 - \eta\right)\frac{|\alpha|^2}{2}\right) \pi \delta^2(r_1 + \alpha\sqrt{\eta}) \nonumber \\
    &=& \exp\left(-(1+2\bar{n})\left(1 - \eta\right)\frac{|\alpha|^2}{2}\right) \Tr[\hat{D}(r_1)\hat{D}(\alpha\sqrt{\eta})]
\end{eqnarray}
Since $\Tr[\Lambda_{\bar n,\eta}(\hat{D}(\xi))\hat{D}(\chi)] = \Tr[\hat{D}(\xi)\Lambda_{\bar n,\eta}^*(\hat{D}(\chi))]$, therefore
\begin{equation}
    \Lambda_{\bar n,\eta}^*(\hat{D}(\alpha)) = \exp\left(-(1+2\bar{n})\left(1 - \eta\right)\frac{|\alpha|^2}{2}\right)\hat{D}(\alpha\sqrt{\eta}).
\end{equation}
\end{proof}
\noindent Note that for $\bar n = 0$, the thermal loss channel corresponds to the photon loss, and we recover the expression of the action on the photon loss channel and its adjoint given in \cite[Appendix C]{Chabaud2024phasespace}. Before going to the next section, we give a mathematical result on how the adjoint of the thermal loss channel affect the Schatten 1-norm of operator. This will be useful in Theorem \ref{apptheo:global_noise_induced conc}:

\begin{lemma}[Effect of the adjoint of the thermal noise channel on the Schatten 1-norm]\label{applem:thermal_Schatten_one}
    Given an operator $\hat O$ with finite Schatten 1-norm and the thermal loss channel $\Lambda_{\bar n,\eta}$:
    \begin{equation}
        \|\Lambda_{\bar n,\eta}^*(\hat O)\|_1 \leq \frac{1}{\eta} \|\hat O\|_1.
    \end{equation}
\end{lemma}
\begin{proof}
    Let us first assume that $\hat O$ is positive semidefinite (PSD), then
    \begin{equation}
        \|\hat O\|_1 = \Tr[\hat O]
    \end{equation}
    and since $\Lambda_{\bar n,\eta}^*$ is a completely positive map
    \begin{eqnarray}
      \|\Lambda_{\bar n,\eta}^*(\hat O)\|_1 =   \Tr[\Lambda_{\bar n,\eta}^*(\hat O)] &=& \frac{1}{\pi^m} \int_{\br \in \R^{2m}} \mathrm d\br \Tr[\Lambda^{*}_{\bar n, \eta}(\hat O) D(\br)] \Tr[\hat{D}(-\br)] = \frac{1}{\pi^m} \int_{\br \in \R^{2m}} d\br \Tr[\hat O \Lambda_{\bar n, \eta}(D(\br))] \Tr[D(-\br)] \nonumber \\
      &=& \frac{1}{\eta\pi^m} \int_{\br \in \R^{2m}} d\br e^{-(1/2 + \bar n)(1/\eta - 1)(q_1^2 + p_1^2)} \Tr[\hat O D(q_1/\sqrt{\eta},p_1/\sqrt{\eta},q_2,p_2,\dots,q_m,p_m)] \pi^m \delta^{2m}(-\br,\bm{0})\nonumber \\ 
      &=& \frac{1}{\eta} \times 1  \times \Tr[\hat O\hat D(\bm{0})] = \frac{1}{\eta}\|\hat O\|_1,
    \end{eqnarray}
    where in the second line, we have used the fact that $\Tr[\hat D(-\br)] = \pi^m \delta^{2m}(-\br,\bm{0})$. When  $\hat O$ is not PSD, we decompose it into PSD operators $\hat O_{+}$ and $\hat O_{-}$ such that
    \begin{equation}
        \hat O = \hat O_{+} - \hat O_{-}
    \end{equation}
    and
    \begin{equation}
        \|\hat O\|_1 = \|\hat O_+\|_1 + \|\hat O_-\|_1.
    \end{equation}
    Therefore,
    \begin{eqnarray}
        \|\Lambda_{\bar n,\eta}^*(\hat O)\|_1 &=& \|\Lambda_{\bar n,\eta}^*(\hat O_+ - \hat O_-)\|_1 \nonumber \\
        &\leq& \|\Lambda_{\bar n,\eta}^*(\hat O_+)\|_1 + \|\Lambda_{\bar n,\eta}^*(\hat O_-)\|_1 \nonumber \\
        &=& \frac{1}{\eta} (\|\hat O_+\|_1 + \|\hat O_-\|_1) = \frac{1}{\eta}\|\hat O\|_1,
    \end{eqnarray}
    where in the second line, we have used Minkowski's inequality.
\end{proof}

\subsection{Coherent state projectors as probability distributions}
When trying to estimate the overlap of the output of noisy bosonic circuit with the local coherent state projectors, we will use the (renormalized) magnitude of the characteristic functions of these local coherent state projectors as probability distributions. In this section, we explain how this can be done. Given the output operator
\begin{equation}
    \hat O = \left(\bigotimes_{i=1}^k \ket{\alpha_i}\bra{\alpha_i}\right)\otimes \mathbb{I}^{\otimes m-k},
\end{equation}
for $\br \in \mathbb{R}^{2m}$, its characteristic function is given as
\begin{equation}
    \chi_{\hat O}(\br) = \Tr[\hat O \hat D(\br)] = \left(\Pi_{i=1}^k (\Tr[\ket {\alpha_i} \bra{\alpha_i} \hat D(r_i)])\times \Pi_{j=k+1}^m (\Tr[\hat D(r_j)])\right).
\end{equation}
Now
\begin{equation}
    \Tr[\hat D(r_j)] = \pi \delta^2(r_j)
\end{equation}
$\forall j \in \{k+1,\dots,m\}$ and $r_j = (q_j,p_j)$ Further, for two single coherent states $\ket \alpha$ and $\ket \beta$,
\begin{equation}
    \braket{\alpha}{\beta} = \exp\left(-\frac12 |\alpha|^2 - \frac12 |\beta|^2 + \alpha^* \beta\right), 
\end{equation}
and therefore
\begin{equation}
    \Tr[\ket {\alpha_i} \bra{\alpha_i} \hat D(r_i)] = \exp(i \Im(\alpha_i r_i^*)) \braket{\alpha_i}{\alpha_i + r_i} = \exp(- \frac{1}{2}|\alpha_i|^2 - \frac{1}{2}|\alpha_i + r_i|^2 + \alpha_i^{*}(\alpha_i + r_i)).
\end{equation}
This gives
\begin{equation}
    \chi_{\hat O}(\br) = \left(\Pi_{i=1}^k \exp(- \frac{1}{2}|\alpha_i|^2 - \frac{1}{2}|\alpha_i + r_i|^2 + \alpha_i^{*}(q_i + i p_i))\right) \times \pi^{m-k} (\Pi_{j=k+1}^m \delta^2(r_j) ).
\end{equation}
For Dirac delta functions (that appear whenever we are dealing with operators which are not acting on all the modes) we use the convention $|\delta|\coloneqq \delta$. With
\begin{equation}
    |\braket{\alpha_i}{\alpha_i + r_i}| = \exp\left(-\frac{1}{2}|\alpha_i + r_i - \alpha_i|^2\right) = \exp\left(-\frac{q_i^2 + p_i^2}{2}\right),
\end{equation}
we get that
\begin{equation}\label{appeq:Fourier_1_norm_local_projectors}
    \fnormone{\hat O} = \frac{1}{\pi^m} \int_{\mathbb{R}^{2m}} \mathrm d^{2m}\br |\chi_{\hat O}(\br)| = \frac{1}{\pi^k} \Pi_{i=1}^k \int_{\mathbb{R}^2} \mathrm dq_i \mathrm dp_i \exp\left(-\frac{q_i^2 + p_i^2}{2} \right) = 2^k.
\end{equation}
Therefore the magnitude of the characteristic function of $k$-local coherent state projectors can be used as a probability distribution by dividing it by the normalizing constant $2^k$.

\section{Technical lemmas for bosonic operators}
\noindent This section details additional technical tools that are specific to our expectation value estimation problems. The results mentioned in this section will be referenced later when we detail our estimation algorithms and their running times.

\subsection{Derivative bound for the noisy characteristic function}

In this work, we consider noisy bosonic circuits such that the magnitude of first derivative of the expectation value of the displacement operator with respect to the position coordinate of the first mode is bounded throughout the computation. The following result connects this rather non-intuitive quantity to a much more commonly considered quantity, the quadrature moments:
\begin{lemma}[Derivative bound for noisy characteristic function]\label{lem:curvature_bound}
Given a quantum state $\rho$ and displacement operator $\hat D(\br)$, with $r_1 = q_1 + i p_1$, we have that
  \begin{align}
    \max_{q_{1} \in \mathbb{R}}\bigg|{\frac{\partial}{\partial q_{1}} \Tr[\hat D(\br)\Lambda_{\bar n, \eta}(\rho)]} \bigg| \leq f_{\rho,\br}(\hat p),
\end{align} 
where $f_{\rho,\alpha}(\hat p)$ is a function depending on the expectation value of the first two moments of the momentum quadrature of the first mode $\hat{p}_1$.
\end{lemma}
\begin{proof}
First we write,
\begin{equation}
    \Tr[\hat D(\alpha)\Lambda_{\bar n, \eta}(\rho)] = \Tr[\Lambda_{\bar n, \eta}^* (\hat D(\alpha)) \rho].
\end{equation}
Now, from Lemma \ref{applem:thermal_evo}, for a single-mode displacement operator
\begin{eqnarray}
    \Lambda_{\bar n, \eta}^* (\hat D(r_1)) &=& \exp\left(-(1/2 + \bar n)(1 - \eta)(q^{2}_{1} + p_{1}^2)\right)\hat D(\sqrt \eta r_1) \nonumber \\
    &=&\exp\left(-(1/2 + \bar n)(1 - \eta)(q^{2}_{1} + p_{1}^2)\right) \exp\left(i\sqrt \eta p_{1} \hat q_1 - i \sqrt \eta q_{1} \hat{p}_1\right).
\end{eqnarray}
Therefore, taking the derivative with respect to the position coordinate,
\begin{eqnarray}
    \frac{\partial}{\partial q_1} \Lambda_{\bar n, \eta}^* (\hat D(r_1)) &=&  \exp\left(i\sqrt \eta p_{1} \hat q_1 - i \sqrt \eta q_{1} \hat{p}_1\right) \frac{\partial}{\partial q_1} \exp\left(-(1/2 + \bar n)(1 - \eta)(q^{2}_{1} + p_{1}^2)\right) \nonumber \\&& + \exp\left(-(1/2 + \bar n)(1 - \eta)(q^{2}_{1} + p_{1}^2)\right) \frac{\partial}{\partial q_1}\exp\left(i\sqrt \eta p_{1} \hat q_1 - i \sqrt \eta q_{1} \hat{p}_1\right).
\end{eqnarray}
Now, using Duhamel's formula \cite[Eq.\ 12]{Bauer2013TimeOrdering},
\begin{equation}
    \frac{\partial}{\partial q_1} \exp\left(i\sqrt \eta p_{1} \hat q_1 - i \sqrt \eta q_{1} \hat{p}_1\right) = (i \eta p_1 - i \sqrt{\eta} \hat p_1) \exp\left(i\sqrt \eta p_{1} \hat q_1 - i \sqrt \eta q_{1} \hat{p}_1\right).
\end{equation}
Therefore,
\begin{eqnarray}
     \frac{\partial}{\partial q_1} \Lambda_{\bar n, \eta}^* (\hat D(r_1)) &=& \exp\left(-(1/2 + \bar n)(1 - \eta)(q^{2}_{1} + p_{1}^2)\right) \times \bigg( - (1+ 2\bar n)(1-\eta) q_1 + (i\eta p_1 - \sqrt{\eta} \hat p_1)\bigg)\times \hat{D}(\sqrt{\eta}r_1).
\end{eqnarray}
 And we get
 \begin{eqnarray}
     \left|\Tr\left[\frac{\partial}{\partial q_1} \Lambda_{\bar n, \eta}^* (D(r_1)) \rho\right]\right| &\leq& \bigg((1+2\bar n)(1-\eta)|q_1| + \eta |p_1| + \sqrt{\eta}\left|\Tr[\hat p_1 \hat D(\sqrt{\eta} r_1) \rho]\right|\bigg) \times \exp\left(-(1/2 + \bar n)(1 - \eta)(q^{2}_{1} + p_{1}^2)\right). \nonumber \\
 \end{eqnarray}
 Where we have used the fact that the magnitude of the characteristic function of a state $|\Tr[\rho \hat D(\br)] |\leq 1, \forall \br$. To bound the third term, we use the inequality (Lemma \ref{lem:trace_second_moment})
\begin{equation}
        \abs{\Tr[A\rho]}^2 \leq \frac{1}{2} \Tr[\{A,A^\dag\}\rho],
\end{equation}
with $A = \hat p_1\hat{D}(\sqrt{\eta}r_1)$. Now,
\begin{equation}
    \{A,A^\dag\}  = \hat p_1^2 + \hat p_1^2 + 4 \eta p_1^2 + 4 \sqrt{\eta} p_1 \hat p_1 = 2 \hat p_1^2 + 4\eta p_1^2 + 4\sqrt{\eta} p_1 \hat p_1.
\end{equation}
Therefore,
\begin{eqnarray}
    \left|\Tr\left[\frac{\partial}{\partial q_1} \Lambda_{\bar n, \eta}^* (D(\alpha)) \rho\right]\right| &\leq& \bigg((1+2\bar n)(1-\eta)|q_1| + \eta |p_1| + \sqrt{\Tr[ (\hat p_1^2 + 2\eta p_1^2 + 2\sqrt{\eta} p_1 \hat p_1) \rho]}\bigg) \nonumber \\ &&\hspace{10mm}\times\exp\left(-(1/2 + \bar n)(1 - \eta)(q^{2}_{1} + p_{1}^2)\right).
\end{eqnarray}
\end{proof}
Therefore, bounding the second  moment of the position quadrature throughout the circuit guarantees that $\left|\Tr[\frac{\partial}{\partial q_1} \Lambda_{\bar n, \eta}^* (D(\br)) \rho]\right|$ is bounded throughout the circuit, $\forall \br \in \mathbb{R}^{2m}$. \noindent
Further, recalling the total energy operator
\begin{equation}
    \hat N \;=\; \sum_{j=1}^m \hat N_j
    \;=\; \sum_{j=1}^m \frac{\hat q_j^2 + \hat p_j^2}{2},
\end{equation}
we have for each mode \(j\)
\begin{equation}
     N_j \succeq \frac{1}{2}\,\hat p_j^2,
\end{equation}
and therefore for any state \(\rho\) (with \(\Tr\rho=1\))
\begin{equation}
    \Tr(\rho\,\hat p_j^2) \le 2\,\Tr(\rho\,\hat N).
\end{equation}
Hence a uniform bound on the energy operator \(\Tr(\rho\,\hat N)\) throughout the circuit implies uniform bounds on the first two moments of each quadrature \(\hat q_j\).

\subsection{Estimating the first and second quadrature moments using the characteristic function}
In this section, we explain how estimating the characteristic function of a quantum state enables the estimation of quadrature moments. The key observation is that position and momentum operators can be written as partial derivatives of displacement operators and can therefore be approximated by finite differences of displacement operators through a Taylor expansion.
However, this approximation is valid only under suitable energy bounds, since position and momentum operators (and their powers) are unbounded operators. These assumptions ensure that Taylor’s theorem can be applied consistently and that the corresponding expectation values remain finite.
\begin{lemma}\label{applem:quad_finite_diff}
     Given a quantum state $\rho$ with finite quadrature moments and $\delta>0$, it holds that $\forall j \in \{1,\dots,m\}$:
     \begin{eqnarray}
        \left| \Tr[\rho \hat q_j] - \frac{\Tr[\rho \hat D_j(0,\delta)] - 1}{i\delta} \right| &\leq& \delta \Tr[\rho \hat q_j^4] ,\nonumber \\
       \left| \Tr[\rho \hat q_j^2] - \frac{2}{\delta^2}\left((1 - \Tr[\rho \hat D_j(0,\delta)]) + \frac{1}{\delta}(\Tr[\rho \hat D_j(0,\delta^2)] - 1)\right)\right| &\leq& \frac{2}{3} \delta \Tr[\rho \hat q_j^6] + 2 \delta \Tr[\rho \hat q_j^4], \nonumber \\
       \left| \Tr[\rho \hat p_j] - \frac{1-\Tr[\rho \hat D_j(\delta,0)]}{i\delta} \right| &\leq& \delta \Tr[\rho \hat p_j^4], \nonumber \\
       \left| \Tr[\rho \hat p_j^2] - \left(\frac{2}{\delta^2}(1 - \Tr[\rho \hat D_j(\delta,0)]) + \frac{2}{\delta^3}(\Tr[\rho \hat D_j(\delta^2,0)] - 1)\right)\right| &\leq& \frac{2}{3} \delta \Tr[\rho \hat p_j^6] + 2 \delta \Tr[\rho \hat p_j^4],
     \end{eqnarray}
     where $\hat D_j(q_j,p_j)$ is the $j$th mode displacement operator.
 \end{lemma}
\begin{proof}
    We will prove the result for the position quadratures, and the results for the momentum quadratures can be similarly proven. The $j$th mode displacement operator of the first mode is given by
    \begin{equation}
        \hat D_j(q,p) = \exp(i p \hat q_j - i q \hat p_j).
    \end{equation}
    Therefore,
    \begin{equation}
        \hat D_j(0,p) = \exp(i p \hat q_j),
    \end{equation}
    and hence
    \begin{equation}
        \frac{\partial}{\partial p} \hat D_j(0,p) \bigg|_{p=0} = i \hat{q}_j.
    \end{equation}
    By Taylor's remainder theorem (Theorem \ref{apptheo:taylor_rem}),
    \begin{eqnarray}
       \Tr[ \rho \hat D_j(0,\delta)] &=& \Tr[\rho\hat D_j(0,0)] + \delta\Tr\left[\rho \frac{\partial}{\partial p} \hat D_j(0,p)\bigg|_{p=0}\right] + \frac{\delta^2}{2!} \Tr\left[\rho \frac{\partial^2}{\partial p^2} \hat D_j(0,p)\bigg|_{p=\xi}\right] \nonumber \\
        &=& \Tr[\rho\hat D_j(0,0)] + i \delta \Tr[\rho \hat q_j] + + \frac{\delta^2}{2!} \Tr\left[\rho \frac{\partial^2}{\partial p^2} \hat D_j(0,p)\bigg|_{p=\xi}\right],
    \end{eqnarray}
    for some $\xi \in (0,\delta)$. Therefore,
    \begin{equation}
        \Tr[\rho \hat q_j] = \frac{\Tr[\rho \hat D_j(0,\delta)] - 1}{i
        \delta} + i\frac{\delta}{2}\Tr\left[\rho \frac{\partial^2}{\partial p^2} \hat D_j(0,p)\bigg|_{p=\xi} \right].
    \end{equation}
    Now
    \begin{equation}
       \left| \Tr\left[\rho \frac{\partial^2}{\partial p^2} \hat D_j(0,p) \right] \right| = \left|\Tr[\rho \hat q_j^2 \hat{D}_j(0,p)]\right| \nonumber \\
       \leq \left|\Tr[\rho \{\hat q_j^2 \hat{D}_j(0,p), \hat{D}_j^\dagger(0,p), \hat q_j^2\}]\right| \nonumber \\
       = 2 \Tr[\rho \hat q_j^4],
    \end{equation}
    where the inequality follows from Lemma \ref{lem:trace_second_moment}. Therefore, we get
    \begin{equation}
        \left| \Tr[\rho \hat q_j] - \frac{\Tr[\rho \hat D_j(0,\delta)] - 1}{i\delta} \right| \leq \delta \Tr[\rho \hat q_j^4].
    \end{equation}
To estimate $\Tr[\rho \hat q_j^2]$, notice that
\begin{equation}
    \frac{\partial^2}{\partial p^2} \hat D_j(0,p) \bigg|_{p=0} = -\hat{q}_j^2.
\end{equation}
Therefore, using Taylor's remainder theorem (Theorem \ref{apptheo:taylor_rem}) and expanding up to the third order term,
\begin{eqnarray}
        \Tr[\rho\hat D_j(0,\delta)] &=& \Tr[\rho\hat D_j(0,0)] + \delta \Tr\left[\rho\frac{\partial}{\partial p} \hat D_j(0,p)\bigg|_{p=0}\right] + \frac{\delta^2}{2!}\Tr\left[ \rho \frac{\partial^2}{\partial p^2} \hat D_j(0,p)\bigg|_{p=0}\right] +  \frac{\delta^3}{3!} \Tr\left[ \rho \frac{\partial^3}{\partial p^3} \hat D(0,p)\bigg|_{p=\xi}\right]\nonumber \\
        &=& \Tr[\rho\hat D_j(0,0)] + i \delta \Tr[ \rho \hat q_j] - \frac{\delta^2}{2!} \Tr[ \rho\hat q_j^2] +\frac{\delta^3}{3!} \Tr\left[ \rho \frac{\partial^3}{\partial p^3} \hat D(0,p)\bigg|_{p=\xi}\right]
    \end{eqnarray}
    for some $\xi \in (0,\delta)$. Therefore,
    \begin{equation}
        \Tr[\rho \hat q_j^2] = \frac{2}{\delta^2}(1- \Tr[\rho \hat D_j(0,\delta)]) + \frac{2i}{\delta} \Tr[\rho \hat q_j] + \frac{\delta}{3} \Tr\left[\rho\frac{\partial^3}{\partial p^3} \hat D_j(0,p)\bigg|_{p=\xi}\right].
    \end{equation}
    Now,
    \begin{equation}
        \left| \Tr\left[\rho \frac{\partial^3}{\partial p_1^3} \hat D(0,p_1) \right] \right| = \left|\Tr[\rho \hat q_1^3 \hat{D}(0,p_1)]\right| \nonumber \\
       \leq \left|\Tr[\rho \{\hat q_1^3 \hat{D}(0,p_1), \hat{D}^\dagger(0,p_1), \hat q_1^3\}]\right| \nonumber \\
       = 2 \Tr[\rho \hat q_1^6].
    \end{equation}
    Where the inequality follows from Lemma \ref{lem:trace_second_moment}. Also from the previous discussion,
    \begin{equation}
        \Tr[\rho \hat q_j] = \frac{\Tr[\rho \hat D_j(0,\delta^2)] - 1}{i\delta^2} \pm \epsilon,
    \end{equation}
    where $|\epsilon| \leq \delta^2 \Tr[\rho \hat q^4]$. Therefore, using the triangle inequality, we obtain
    \begin{equation}
        \left| \Tr[\rho \hat q_j^2] - \frac{2}{\delta^2}\left((1 - \Tr[\rho \hat D_j(0,\delta)]) + \frac{1}{\delta}(\Tr[\rho \hat D_j(0,\delta^2)] - 1)\right)\right| \leq \frac{2}{3} \delta \Tr[\rho \hat q_j^6] + 2 \delta \Tr[\rho \hat q_j^4].
    \end{equation}
\end{proof}
Therefore, given guarantees on the moments of quadrature up to the sixth order on the evolved output state, we can approximate the first two quadrature moments with the characteristic function of the output state up to arbitrary precision. This result will be useful after we detail simulation algorithm to estimate the characteristic function. Since $\hat N_j^2 \succeq \hat q_j^4, \hat N_j^3 \succeq \hat q_j^6, \forall j \in (1,\dots,m)$ for the energy operator $\hat N$ (Eq.\ \ref{appeq:energy_op}), the condition of bounded fourth and sixth quadrature moments on the output state can be expressed as condition on boundedness of second and third moments of $\hat N_j$, as is done in the main text. 

Note that similar techniques allow us to estimate higher quadrature moments in terms of characteristic functions. However, this also requires guarantees on magnitudes of higher moments because we need to expand the Taylor's remainder theorem up to higher order terms.


\subsection{Useful decompositions of bosonic channels}
\label{appsec:decompositions}
As detailed in the preliminaries in Section \ref{appsec:bosonic_prelims}, this work deals with noisy bosonic circuits consisting of Gaussian unitary gates and cubic phase gates and we model the interaction with the environment using a thermal loss channel.

In what follows, we denote cubic phase gates ($e^{i\gamma \hat q^3}$) and Gaussian gates ($\hat G$) by the quantum channels $\mathcal C$ and $\mathcal G$, such that their action is given by:
\begin{eqnarray}
    \mathcal C(\cdot) &=& e^{i\gamma \hat q^3} (\cdot) e^{-i\gamma \hat q^3} \nonumber, \\
    \mathcal C^*(\cdot) &=& e^{-i\gamma \hat q^3} (\cdot) e^{i\gamma \hat q^3}, \nonumber \\
    \mathcal G(\cdot) &=& \hat G (\cdot) \hat G^\dagger, \nonumber \\
    \mathcal G^*(\cdot) &=& \hat G^\dagger (\cdot) \hat G.
\end{eqnarray}

\noindent\textbf{Cubic phase gates.} For convenience, we decompose each cubic phase gate $\calC$ with cubicity $\gamma$ into two consecutive linear maps:
\begin{align}
    \calC^* = \calC^*_{\mathfrak{C}} \circ \calC^*_{\mathfrak{D}}, \qquad 
    \calC = \calC_{\mathfrak{D}} \circ \calC_{\mathfrak{C}},
\end{align}
where $\mathfrak{D}$ stands for ``discrete'' and $\mathfrak{C}$ stands for ``continuous''. The actions of the maps in the displacement operator basis are given by:
\begin{align}
    &\calC^*_{\mathfrak{D}} (D(\br))  \coloneqq 
    \begin{cases}
      \sqrt{\tfrac{1}{24\pi \gamma q_1}}\, D(\br)  &\text{if $q_1 \neq 0$},\\
       D(\br)  &\text{if $q_1 = 0$},
    \end{cases} \\
    &\calC^*_{\mathfrak{C}} (D(\br))  \coloneqq 
    \begin{cases}
        \int_{\mathbb{R}} d\bar{p}_{1} \, 
        e^{i \left[\tfrac{(p_{1}-\bar{p}_{1})^2}{24  \gamma q_{1}} - 4 \gamma q_{1}^3 - \tfrac{\pi}{4}\right]} 
        D(\bar{\br}), &\text{if $q_1 \neq 0$},\\
       D(\br)  &\text{if $q_1 = 0$},
    \end{cases}
\end{align}
where $\bar{\br} = (q_1, - \bar p_1, q_2, p_2 \dots, q_m, p_m)$. Note that the composition correctly describes the action of the cubic phase gate on the displacement operator according to Lemma \ref{applem:cubic_evo}.

\noindent We also define an analogous decomposition for the channel itself:
\begin{align}
    &\calC_{\mathfrak{D}} (D(\br))  \coloneqq 
    \begin{cases}
      \sqrt{\tfrac{1}{24\pi \gamma q_1}}\, D(\br)  &\text{if $q_1 \neq 0$},\\
       D(\br)  &\text{if $q_1 = 0$},
    \end{cases} \\[1ex]
    &\calC_{\mathfrak{C}} (D(\br))  \coloneqq 
    \begin{cases}
        \int_{\mathbb{R}} d\bar{p}_{1} \, 
        e^{i \left[-\tfrac{(p_{1}-\bar{p}_{1})^2}{24  \gamma q_{1}} + 4 \gamma q_{1}^3 + \tfrac{\pi}{4}\right]} 
        D(\bar{\br}), &\text{if $q_1 \neq 0$},\\
       D(\br)  &\text{if $q_1 = 0$},
    \end{cases}
\end{align}
with $\bar{\br} = (q_1, q_2, \dots, q_m, \bar p_1, p_2, \dots, p_m )$.  
\vspace{1em}

\noindent \textbf{Thermal channels.}
 We introduce the following linear map, which we dub \emph{depolarizing map} with noise rate $p$:
\begin{align}\label{def:depomap}
    \mathcal{N}_{p}(\hat D(\xi)) \coloneqq \exp\left(- p{\abs{\xi}^2}\right) \hat D(\xi).
\end{align}
The thermal channel can be re-expressed as follows:
\begin{align}
    &\Lambda_{\bar{n},\eta}(\hat D(\xi)) = \Lambda_{0,\eta} \circ \mathcal{N}_{\bar n(1/\eta - 1)} (\hat D(\xi)),
    \\ &\Lambda^*_{\bar{n},\eta}(\hat D(\xi)) = \Lambda_{0,\eta}^* \circ \mathcal{N}_{\bar n(1-\eta)} (\hat D(\xi)),
\end{align}
such that the composition of these maps describe the action of thermal loss channel and its adjoint on the displacement operator according to Lemma \ref{applem:thermal_evo}. This decomposition is used in the warm-up demonstration hereafter, in which we explain how a strong noise model of uniform thermal losses can render the overlap of two quantum states trivial, in the sense of concentrating to zero.

\section{Warm up: Overlap decay under uniform noise}\label{appsec:purity_uniform}
We analyze the effect of a uniform layer of thermal noise acting on all modes of a bosonic circuit (Figure \ref{fig:uniform_layer}). This setting, while strong, provides a clean and illustrative case where overlap quantum states decays rapidly under noise.
As a preliminary technical step, we show that the purity of an operator evolved under the (unphysical) depolarizing map introduced in the previous section decay exponentially with the number of modes.
\begin{lemma}[Purity decay under depolarizing map]\label{lem:depolarizing_2_norm}
Let $H$ be a trace-class operator on $m$ modes, and $\mathcal N_p^{\otimes m}$ denote the depolarizing map with parameter $p$ (see Eq.\ \ref{def:depomap}).  For any $\tau\geq 0$, it holds that
\begin{align}
    \norm{\calN_p^{\otimes m}(H)}_2^2 \leq \left(\frac{\tau^m}{m!} \right)\norm{H}_1^2 + \exp(- 2p\tau )  \norm{H}_1^2. 
\end{align}
In particular, it holds that
\begin{align}
    \norm{\calN_p^{\otimes m}(H)}_2^2 \in  \exp\left(- \Omega(m\cdot\min\{1,p\})\right)\norm{H}_1^2.
\end{align}
\end{lemma}
\begin{proof}
Expressing the Schatten 2-norm of $\calN^{\otimes m}_p(H)$ in the displacement operator basis, we find that
\begin{align}
    \norm{\calN_p^{\otimes m}(H)}_2^2
    = &\frac{1}{\pi^m}\int_{\mathbb{C}^m} \mathrm d^{2m}\boldsymbol{\alpha} \left\lvert{\Tr[H\calN_p^{\otimes m}(\hat D(\boldsymbol{\alpha}))]}\right\rvert^2
    \\ = &\underbrace{\frac{1}{\pi^m}\int_{\substack{\mathbb{C}^m:\\ \norm{\boldsymbol{\alpha}}_2^2\leq \tau}} \mathrm d^{2m}\boldsymbol{\alpha} \left\lvert\Tr[H\calN_p^{\otimes m}(\hat D(\boldsymbol{\alpha}))]\right\rvert^2}_{\coloneqq I_{\mathrm{low}}}
    + \underbrace{\frac{1}{\pi^m}\int_{\substack{\mathbb{C}^m:\\ \norm{\boldsymbol{\alpha}}_2^2\geq \tau}} \mathrm d^{2m}\boldsymbol{\alpha} \left\lvert{\Tr[H\calN_p^{\otimes m}(\hat D(\boldsymbol{\alpha}))]}\right\rvert^2}_{\coloneqq I_{\mathrm{high}}},
\end{align}
where in the first step we used the fact that the depolarizing map and its adjoint coincide, i.e.\ $\calN_p^* = \calN_p$.
We now upper bound $I_{\mathrm{low}}$ and $I_{\mathrm{high}}$ separately. The integral $I_{\mathrm{low}}$ can be upper bounded as follows
\begin{align}
    I_{\mathrm{low}}=&\frac{1}{\pi^m}\int_{\substack{\mathbb{C}^m:\\ \norm{\boldsymbol{\alpha}}_2^2\leq \tau}} \mathrm d^{2m}\boldsymbol{\alpha} \left\lvert\Tr[H\calN_p^{\otimes m}(\hat D(\boldsymbol{\alpha}))]\right\rvert^2 \\\leq &\max_{\boldsymbol{\alpha}}\left\{\exp\left(-2p\abs{\alpha}^2\right)\left\lvert \Tr[H\hat D(\boldsymbol{\alpha})]\right\rvert^2\right\} \times \left(\frac{1}{\pi^m}\int_{\substack{\mathbb{C}^m:\\ \norm{\boldsymbol{\alpha}}_2^2\leq \tau}} \mathrm d^{2m} \boldsymbol{\alpha}\right)
    \\\leq &\norm{H}_1^2 \max_{\boldsymbol{\alpha}}\norm{\hat D(\alpha)}_\infty^2  \frac{\tau^m}{m!} = \frac{\tau^m}{m!} \norm{H}_1^2,
\end{align}
where in the second inequality we used Hölder's inequality and noted that $\norm{D(\boldsymbol{\alpha})}_\infty = 1$ since displacement operators are unitary. Moreover, we also used the fact that the volume of the $2m$-dimensional ball of radius $\tau$ is $\pi^{m}\tau^m/m!$ (Fact\ \ref{fact:ball}).
As for $I_{\mathrm{high}}$, we have that
\begin{align}
    I_{\mathrm{high}}=&\frac{1}{\pi^m}\int_{\substack{\mathbb{C}^m:\\ \norm{\boldsymbol{\alpha}}_2^2\geq \tau}} \mathrm d^{2m}\boldsymbol{\alpha} \Tr[H\mathcal{N}_p^{\otimes m}(\hat D(\boldsymbol{\alpha}))]^2
   \\ = &\frac{1}{\pi^m}\int_{\substack{\mathbb{C}^m:\\ \norm{\boldsymbol{\alpha}}_2^2\geq \tau}} \mathrm d^{2m}\boldsymbol{\alpha} \exp\left(- 2p {\norm{\boldsymbol{\alpha}}_2^2}\right)\abs{\Tr[H\hat D(\boldsymbol{\alpha})]}^2 \\\leq& \exp(- 2 p\tau ) \int_{\substack{\mathbb{C}^m:\\ \norm{\boldsymbol{\alpha}}_2^2\geq \tau}} \mathrm d^{2m}\boldsymbol{\alpha} \abs{\Tr[H\hat D(\boldsymbol{\alpha})]}^2   \\ \leq & \exp(- 2 p\tau ) \int_{\mathbb{C}^m} \mathrm d^{2m}\boldsymbol{\alpha} \abs{\Tr[H\hat D(\boldsymbol{\alpha})]}^2 = \exp(- 2 p\tau )  \norm{H}_2^2 \\ \leq & \exp(- 2 p\tau ) \norm{H}_1^2,
\end{align}
where for the last inequality, we have used the property of Schatten norms $\|H\|_2 \leq \| H\|_1$. Putting everything  together, we obtain the first part of the Lemma:
\begin{align}
    \norm{\calN_p^{\otimes m}(H)}_2^2 \leq \left(\frac{\tau^m}{m!} + \exp(- 2p\tau ) \right)\norm{H}_1^2.
\end{align}
In order to prove the second part of the Lemma, we employ the inequality $\frac{1}{m!} \leq \frac{e^{m}}{m^m}$  and we choose $\tau = 0.9 e^{-1} m$:
\begin{align}
    \norm{\calN_p^{\otimes m}(H)}_2^2 \leq &\left(\frac{e^m\tau^m}{m^m} + \exp(- 2p\tau ) \right)\norm{H}^2_1
   \\ =  &2\left({0.9^m} + \exp(- m (1.8 e^{-1} p ) ) \right)\norm{H}^2_1 \\
   \in &\exp\left(-\Omega(m) \right)\norm{H}^2_1  + \exp\left(-\Omega(mp) \right)\norm{H}^2_1 
   \\\subseteq &\exp(-\Omega(m\cdot\min\{1,p\}))\norm{H}^2_1,
\end{align}
where in the second-to-last step we noted that $1.8 e^{-1} p\in \Omega(p)$ and $0.9^m \in \exp\left( -\Omega(m)\right)$.
\end{proof}
Building on Lemma\ \ref{lem:depolarizing_2_norm}, which establishes exponential purity decay under the depolarizing map, we can now lift this result to the physical thermal loss channel, leading to the following Theorem.
\begin{theorem}[Overlap decay under thermal channel] Given two quantum states $\rho$ and $\sigma$ and the $m$-mode thermal loss channel $\Lambda_{\bar n,\eta}^{\otimes m}$ satisfying $\bar n(1-\eta)\in \Omega(1)$
\begin{align}
   \left|\Tr[\Lambda_{\bar n,\eta}^{\otimes m}(\rho)\sigma] \right|= \left|\Tr[\rho\Lambda_{\bar n,\eta}^{\otimes m *}(\sigma)] \right| \in \exp(-\Omega(\bar n(1-\eta)m)).
\end{align}   
\end{theorem}
\begin{proof}
    Decomposing each of the thermal loss channels as $\Lambda_{\bar n,\eta}^{*}$ as $\Lambda_{\bar n,\eta}^{*} = \Lambda_{0,\eta}^* \circ \mathcal N_{\bar n(1-\eta)}$, we get
    \begin{align}
        \left|\Tr[\rho\Lambda_{\bar n,\eta}^{\otimes m *}(\sigma)] \right| = &\left|\Tr[\Lambda_{0,\eta}^{\otimes m}(\rho) \mathcal N_{\bar n(1-\eta)}^{\otimes m}(\sigma) ] \right| \leq \|\Lambda_{0,\eta}^{\otimes m}(\rho)\|_2 \|\calN_{\bar n(1-\eta)}^{\otimes m}(\sigma)\|_2 \\\in &\exp(-\Omega(m\cdot \min\{1,\bar n(1-\eta))\}) 
        \subseteq \exp(-\Omega(\bar n(1-\eta)m)),
    \end{align}
Here we first used Hölder’s inequality (Lemma~\ref{lem:holder}) with $p=q=2$, the facts that $\Lambda_{0,\eta}^{\otimes m}(\rho)$ is a quantum state and therefore its purity $ \|\Lambda_{0,\eta}^{\otimes m}(\rho)\|_2^2$ is at most 1, together with Lemma~\ref{lem:depolarizing_2_norm}. Finally, in the last step we used that $\bar n(1-\eta)\in\Omega(1)$.
\end{proof}
This Theorem demonstrates that for a uniform layer of thermal loss channel, the overlap of two states exponentially in $m$ concentrates to zero and hence the simulation of such overlaps becomes trivial.

Before establishing classical simulation results for a weaker noise model, we need a few more technical results, which we detail in the upcoming section.

\section{Sampling from the Fourier-Weyl spectrum}
\noindent As stated earlier, our simulation algorithm is an instance of Markov chain Monte Carlo methods, where we are sampling a point from an appropriate probability distribution and evaluating the estimator function at that point. In order to formalize our algorithm and our results, we first define a sampling oracle:

\begin{definition}[Sampling Oracle]
Let $O$ be an $m$-mode operator and let $\calS$ be a set of $m$-mode operators. 
Given $\epsilon, A\geq 0$, we define the  \emph{Sampling Oracle} $\mathsf{Sample}^{(\mathcal{S})}(O)$ with bias $\epsilon$ and magnitude $A$.
  When queried,  $\mathsf{Sample}^{(\calS)}(O)$ returns independent draws $(Z,\br) \in \mathbb{C} \times \mathbb{R}^{2m}$ such that for any fixed operator $\sigma \in \mathcal{S}$ independent of $(Z, \br)$, it holds that:
\begin{align}
    |Z|\le A
    \quad \text{and} \quad
   \bigl| \mathbb{E}\bigl[Z\Tr\bigl[D(\br)\sigma\bigr]\bigr] -\Tr\bigl[O\sigma\bigr]\bigr|\leq \epsilon.
\end{align}
When $\mathcal{S}$ is the entire set of $m$-mode operators, we drop the superscript and simply write $\mathsf{Sample}(O)$.
\end{definition}
\noindent The sampling oracle can be thought of as a randomized procedure that outputs a displacement operator (weighted by a prefactor), allowing the estimation of expectation values. The \emph{bias} quantifies the systematic error introduced by the oracle, while the \emph{magnitude} controls the maximum variance of its outputs. Together, these parameters allow us to rigorously track how approximation errors accumulate when combining many oracles in a simulation algorithm, as detailed in the following lemma.
\begin{lemma}[Chained Sampling]
\label{applem:CHained-sampling}
Let $\calL$ be a linear map and $O$ be an operator.
Assume we have access to the oracle $\mathsf{Sample}^{(\calS)}(O)$ with bias $\epsilon_1$ and magnitude $A_1$ and also to the oracle $\mathsf{Sample}^{(\calS)}(\calL(\hat D(\br)))$ with bias $\epsilon_2$ and magnitude $A_2$ for all $\br \in \mathbb{R}^{2m}$.
Moreover, we assume that $\sigma \in \calS \implies \calL^*(\sigma) \in \calS$.
Then we can simulate the oracle $\mathsf{Sample}^{(\calS)}(\calL(O)) $ with bias $\epsilon_1 + A_1\epsilon_2$ and magnitude $A_1A_2$.
\end{lemma}
\begin{proof}
   Let $(Z_1,\br_1)$ be drawn from the oracle  $\mathsf{Sample}^{(\calS)}(O)$. Then it holds that
    \begin{equation}
       \forall \sigma \in \calS : \bigl| \mathbb{E}_{Z_1,\br_1}\bigl[Z_1\Tr\bigl[\hat D(\br_1)\calL^*(\sigma)\bigr]\bigr] -\Tr\bigl[O\calL^*(\sigma)\bigr]\bigr|\leq \epsilon_1.
    \end{equation}
By the definition of adjoint map, it follows that
    \begin{equation}
        \bigl| \mathbb{E}_{Z_1,\br_1}\bigl[Z_1\Tr\bigl[\mathcal L(\hat D(\br_1))\sigma\bigr]\bigr] -\Tr\bigl[\mathcal L(O)\sigma\bigr]\bigr|\leq \epsilon_1.
    \end{equation}
Moreover, given $(Z_2, \br_2)$ drawn from the oracle $\mathsf{Sample}^{(\calS)}(\calL(\hat D(\br_1)))$, it holds that
    \begin{equation}
       \forall \sigma \in \calS : \bigl| \mathbb{E}_{Z_2,\br_2}\bigl[Z_2\Tr\bigl[\hat D(\br_2)\sigma\bigr]\bigr] -\Tr\bigl[\calL(\hat D(\br_1))\sigma\bigr]\bigr|\leq \epsilon_2.
    \end{equation}
Thus, for all $\sigma \in \calS$ it holds that
    \begin{align}
    &\mathbb{E}_{\substack{Z_1,\br_1, Z_2, \br_2}}\left\{Z_1 Z_2  \Tr\bigl[\hat D(\br_2)\sigma\bigr] \right\} \\=
    &\mathbb{E}_{Z_1,\br_1}\left\{ Z_1\mathbb{E}_{Z_2,\br_2 | Z_1,\br_1}\bigl[Z_2\Tr\bigl[\hat D(\br_2)\sigma\bigr] \bigr] \right\} \\= &\mathbb{E}_{Z_1,\br_1}\left\{Z_1\Tr\bigl[\calL(\hat D(\br_1))\sigma\bigr] \bigr] \right\} + Z_1\alpha_2
    \\ =&\Tr[O\sigma] + \alpha_1 + Z_1\alpha_2,
    \end{align}
   where $\alpha_1, \alpha_2 \in \mathbb{C}$ satisfy $\abs{\alpha_1} \leq \epsilon_1$ and $\abs{\alpha_2} \leq \epsilon_2$, reflecting the fact that the corresponding oracles have biases $\epsilon_1$ and $\epsilon_2$, respectively.
    Rearranging the above equation and using the triangle inequality we find that
    \begin{align}
         \left \lvert \mathbb{E}_{\substack{Z_1,\br_1, Z_2, \br_2}}\left\{Z_1 Z_2  \Tr\bigl[\hat D(\br_2)\sigma\bigr] \right\}  - \Tr[O\sigma] \right\rvert \leq  \epsilon_1 + A_1\epsilon_2.
    \end{align}
    As $Z_1Z_2 \leq \abs{A_1 A_2}$, the above inequality implies that we can simulate the oracle $\mathsf{Sample}^{(\calS)}(\calL(O)) $ with bias $\epsilon_1 + A_1\epsilon_2$ and magnitude $A_1A_2$.
\end{proof}
We briefly comment on how errors propagate when simulating circuits layer by layer. At each step, the bias adds linearly, while the magnitude multiplies across layers. Consequently, repeated chaining of oracles can cause the simulation error to grow exponentially with depth, unless the magnitude remains strictly smaller than one. This contractive property will therefore be crucial in identifying efficiently simulable regimes.

Note that when estimating $\mathbb{E}\bigl[Z\Tr\bigl[\hat D(\br)\sigma\bigr]\bigr]$ via statistical sampling, we are effectively drawing points from phase space and evaluating the characteristic function at these points. Extending this to a Markov chain by chaining several such sampling oracles corresponds to sampling sequences of displacement operators, or ``paths,'' and evaluating the characteristic function at the endpoint of each path. This perspective motivates the terminology \emph{displacement propagation} for the simulation algorithm.

\subsection{Unbiased oracles}
\label{sec:unbiased-oracles}
To develop efficient simulation algorithms, we require unbiased sampling procedures that accurately reproduce the action of noisy circuit layers on displacement operators. In the following Lemmas, we create sampling oracles whose statistical averages coincide exactly with the desired quantum expectation values, for each layer of noisy cubic phase gates plus Gaussian gates, and for layers prepared by decomposing the noisy cubic gates and thermal channels into simpler linear maps, as detailed in Section~\ref{appsec:decompositions}. These Lemmas provide the technical foundation for this construction, presenting explicit recipes for implementing these unbiased estimators.



Our first oracle is defined with respect to a linear map which involves split cubic gates $\calC_{{j-1}, \mathfrak{D}}$ and $\calC_{j, \mathfrak{C}}$:

\begin{lemma}[Split oracle]
\label{lem:oracleAj}
    For $j=2,3,\dots,L$, let $\mathcal A_{j}$ be defined as follows:
    \begin{equation}
        \mathcal A_{j} =  \calC_{{j-1}, \mathfrak{D}}^*
    \circ \Lambda^{*}_{\bar n, \eta} \circ \calG_{j}^*\circ \Lambda^{*}_{\bar n, \eta}\circ \calC_{j, \mathfrak{C}}^*,
    \end{equation}
 Then for all $\br = (q_1,p_1,\dots, q_m,p_m) \in \mathbb{R}^{2m}$ with $q_1\neq0$ we have that
\begin{align}
       \fnormone{\calA_{j}(\hat D(\br))} \leq \frac{\Gamma\!\bigl(\tfrac14\bigr)}{\sigma_{j}\ \sqrt{24\pi}\eta^{3/4}|\gamma_{j-1}|^{1/2}}\times  \left\{\left(\frac{1}{2}+\bar n\right)(1-\eta)\right\}^{-1/4},
\end{align}
where
    \begin{equation}
        (\sigma)_{j} = |(S_{j})_{q_1,p_1}|
    \end{equation}
    and $S_{j}$ is the symplectic matrix associated with the Gaussian channel $\mathcal G_{j}$. Moreover, for all $\br = (q_1,p_1,\dots, q_m,p_m) \in \mathbb{R}^{2m}$ with $q_1\neq0$, there is a classical randomized algorithm that implements the oracle $\mathsf{Sample}(\calA_j(\hat D(\br)))$ with zero bias and magnitude
\begin{align}
 \frac{\Gamma\!\bigl(\tfrac14\bigr)}{\sigma_{j}\ \sqrt{24\pi}\eta^{3/4}|\gamma_{j-1}|^{1/2}}\times  \left\{\left(\frac{1}{2}+\bar n\right)(1-\eta)\right\}^{-1/4}.
\end{align}
\end{lemma}
\begin{proof}
We start by decomposing $\calA_j$ into simpler linear maps, and considering its action on a displacement operator: 
    \begin{eqnarray}
        \calA_j (\hat D(\br)) &=& \calC_{{j-1}, \mathfrak{D}}^*
    \circ \Lambda^{*}_{\bar n, \eta} \circ \calG_{j}^*\circ \Lambda^{*}_{\bar n, \eta}\circ \calC_{j, \mathfrak{C}}^*(\hat D(\br)) \nonumber \\
    &=& \int_\mathbb{R} \mathrm d\bar{p}_1 \exp \left[i\left(\frac{(p_1 - \bar{p}_1)^2}{24\gamma_{j}q_1} - 4\gamma_{j} q_1^3 - \frac{\pi}{4} + \bm{d}_{j}^T\Omega\bar{\br}\right)\right] \nonumber \\
    &&\hspace{5mm}\times \sqrt{\mathrm{sgn}(\tilde q_1 \gamma_{j-1})} \times  \frac{\exp(-(1/2 + \bar n)(1-\eta)(q_1^2 + \bar p_1^2 + \tilde q_1^2 + \tilde p_1^2))}{\eta^{1/4}\sqrt{24\pi|\tilde{q}_1 \gamma_{j-1}|}} \hat D(\sqrt{\eta}\tilde q_1,\sqrt{\eta} \tilde p_1,\tilde q_2, \tilde p_2,\dots, \tilde q_m,\tilde p_m), \nonumber \\
    \end{eqnarray}
    where 
    \begin{eqnarray}
    \bar {\bm{r}} &=& (\sqrt{\eta}q_1,\sqrt{\eta}\bar p_1,\dots,q_m,p_m), \nonumber \\
        \tilde{\bm{r}} =(\tilde{q}_1,\tilde{p}_1,\dots, \tilde{q}_m,\tilde{p}_m)   &=& S_{j}^{-1}(\sqrt{\eta}q_1,\sqrt{\eta}\bar p_1,q_2,p_2,\dots,q_m,p_m). \nonumber \\
    \end{eqnarray}
    Defining the function $\Phi_{j}: \mathbb{R}^{2m+1} \rightarrow \mathbb{C}$ such that $\max_{\br,\bar p_1}\abs{\Phi_j(\br,\bar p_1)} \leq 1$,
    \begin{eqnarray}
        \Phi_{j} (\br,\bar{p}_1) &:=& \exp \left[i\left(\frac{(p_1 - \bar{p}_1)^2}{24\gamma_{j}q_1} - 4\gamma_{j} q_1^3 - \frac{\pi}{4} + \bm{d}_{j}^T\Omega\bar{\br}\right)\right] \nonumber \\
        &\times& \sqrt{\mathrm{sgn}(\tilde q_1 \gamma_{j-1})}\times \exp\left(-(1/2 + \bar n)(1-\eta)(q_1^2 + \bar p_1^2 + \tilde p_1^2)\right) \times  \Tr[\rho \hat D(\bm{\tilde r})]
    \end{eqnarray}
    and the unnormalized density function  $\nu_{j}:\mathbb{R} \rightarrow \mathbb{R}^+$
    \begin{equation}
        \nu_{j}(\bar p_1) := \frac{\exp(-(1/2 + \bar n)(1-\eta)\tilde q_1^2)}{\eta^{1/4}\sqrt{24\pi |\tilde q_1 \gamma_{j-1}|}},
    \end{equation}
    where 
    \begin{equation}
       \tilde q_1 = -\sqrt{\eta} (S^{-1}_{j})_{q_1,p_1} \bar p_1 + \text{terms depending on } (q_1,p_1,\dots,q_m,p_m).
    \end{equation}
    Noting that from \ref{appeq:sigma_sigma_inv}, $|(S_j^{-1})_{q_1,p_1}| = |(S_j)_{q_1,p_1}| = \sigma_j$, we have
    \begin{equation}
        \Tr[\rho \calA_j(\hat D(\br))] = \int_{\bar p_1}\mathrm d\bar p_1 \Phi_j (\br,\bar p_1) \nu_j(\bar p_1).
    \end{equation}
    We note that the Fourier $1$-norm of $\calA_{j}(\hat D(\br))$ can be computed via a Gamma function, as detailed in Lemma\ \ref{applem:gamma_integral}:
    \begin{equation}
        \fnormone{\calA_{j}(\hat D(\br))} = \int_{\mathbb{R}} \mathrm d\bar{p}_1 \nu(\bar p_1) |\Phi_{j} (\br,\bar{p}_1)| \leq  \int_{\mathbb{R}} \mathrm d\bar{p}_1 \nu(\bar p_1)  = \frac{\Gamma\!\bigl(\tfrac14\bigr)}{\sigma_{j}\ \sqrt{24\pi}\eta^{3/4}|\gamma_{j-1}|^{1/2}}\times  \left\{\left(\frac{1}{2}+\bar n\right)(1-\eta)\right\}^{-1/4}.
    \end{equation}
 Therefore, we can invoke Lemma \ref{lem:MC-gamma} to implement $\mathsf{Sample}(\calA_j(\hat D(\br)))$ with zero bias and magnitude
    \begin{equation}
        \frac{\Gamma\!\bigl(\tfrac14\bigr)}{\sigma_{j}\ \sqrt{24\pi}\eta^{3/4}|\gamma_{j-1}|^{1/2}}\times  \left\{\left(\frac{1}{2}+\bar n\right)(1-\eta)\right\}^{-1/4}.
    \end{equation}
\end{proof}

\noindent The split oracle from Lemma\ \ref{lem:oracleAj}  addresses the general case for $\calA_j$ with $j\geq 2$. We then define a linear map $\calA_1$ with a slightly simpler structure which appears at the input of the noisy circuits we simulate (in the Heisenberg picture), and is therefore handled separately in the next result.
\begin{lemma}[Input oracle]
\label{lem:oracleA1}
Given the linear map $\calA_1$ defined as 
\begin{equation}
  \calA_1 \coloneqq  \mathcal G_1^*\circ  \Lambda^{*}_{\bar n, \eta}\circ \calC_{1, \mathfrak{C}}^*,
\end{equation}
then for all $\br \in \mathbb{R}^{2m}$, we have that
\begin{align}
        \fnormone{\calA_1(\hat D(\br))} = \sqrt{\pi}((1/2+\bar n)(1-\eta))^{-1/2}.
\end{align}
Moreover, for all $\br \in \mathbb{R}^{2m}$, there is a classical randomized algorithm that implements the oracle $\mathsf{Sample}(\calA_1(\hat D(\br)))$ with zero bias and magnitude
\begin{align}
    \sqrt{\pi}((1/2+\bar n)(1-\eta))^{-1/2}.
\end{align}
\end{lemma}
\begin{proof}
We start by decomposing $\calA_1$ into simpler linear maps, and considering its action on a displacement operator: 
    \begin{eqnarray}
        \calA_1 (\hat D(\br)) &=&  \mathcal G_1^*\circ  \Lambda^{*}_{\bar n, \eta}\circ \calC_{1, \mathfrak{C}}^*(\hat D(\br)) \nonumber \\
        &=& \int_{\mathbb{R}} \mathrm d\bar{p}_1 \exp\left(i\left(\frac{(p_1 - \bar{p}_1)^2}{24\gamma_1q_1} - 4 \gamma_1 q_1^3 - \frac{\pi}{4} + \bm{d}_1^T \Omega \bm{\bar r}\right)\right) \nonumber \\
        &&\hspace{5mm}\times \exp(-(1/2 + \bar n)(1-\eta)(q_1^2 + \bar{p}_1^2))\hat D(\tilde {\br}),
    \end{eqnarray}
    where
    \begin{equation}
        \tilde {\bm r} = S_1^{-1}(q_1 \sqrt{\eta},\bar{p}_1 \sqrt{\eta},q_2,p_2,\dots,q_m,p_m) = S_1^{-1} \bm{\bar r}.
    \end{equation}
Defining the function $\Phi_1: \mathbb{R}^{2m+1} \rightarrow \mathbb{C}$ such that $\max_{\br, \bar p_1} |\Phi_1(\br, \bar p_1)| \leq 1$,
\begin{eqnarray}
    \Phi_1(\br,\bar p_1) = \exp\left(i\left(\frac{(p_1 + \bar{p}_1)^2}{24\gamma_1q_1} - 4 \gamma_1 q_1^3 - \frac{\pi}{4} + \bm{d}_1^T \Omega \bm{\bar r}\right)\right) \exp(-(1/2 + \bar n)(1-\eta)q_1^2)\Tr[\rho \hat D(\tilde{\br})]
\end{eqnarray}
and the unnormalized density function $\nu_1: \mathbb{R} \rightarrow \mathbb{R}^+$
\begin{equation}
    \nu_1(\bar p_1) = \exp(-(1/2+\bar n)(1-\eta)\bar{p}_1^2),
\end{equation}
we have
\begin{equation}
    \Tr[\rho \calA_1 (\hat D(\br))] = \int_{\mathbb{R}} \mathrm d\bar p_1 \nu_1 (\bar{p}_1) \Phi_1(\br,\bar{p}_1).
\end{equation}
Then the Fourier one norm of $\calA_1(\hat D(\br))$ is given by
\begin{equation}
    \fnormone{\calA_1(\hat D(\br))} = \sqrt{\pi}((1/2+\bar n)(1-\eta))^{-1/2}.
\end{equation}
Moreover, as $|\Phi_1(\br, \bar p_1)| \leq 1$ and $\nu_1$ is a Gaussian probability density function, we can invoke Lemma \ref{lem:MC-Gaussian} to implement $\mathsf{Sample}(\calA_1 (\hat D(\br)))$ with zero bias and magnitude
\begin{equation}
    \sqrt{\pi}((1/2+\bar n)(1-\eta))^{-1/2}.
\end{equation}
\end{proof}
\noindent Now we define the sampling oracle for the following physical quantum channel:
\begin{lemma}[Physical oracle]\label{applem:physical_oracle}
Given the quantum channel defined as 
\begin{equation}
  \calU_j^* \coloneqq  \mathcal G_j^*\circ  \Lambda^{*}_{\bar n, \eta}\circ \calC_{1}^* \circ \Lambda_{\bar n,\eta}^*,
\end{equation}
then for all $\br \in \mathbb{R}^{2m}$, we have that
\begin{align}
        \fnormone{\calU_j^*(\hat D(\br))} \leq \frac{1}{\sqrt{24\sqrt{\eta}\abs{\gamma} |q_1|}} \exp\left[- \left(\frac{1}{2} + \bar n\right)(1-\eta^2)q_1^2\right] \sqrt{1/((1/2 + \bar n)(1-\eta))}.
\end{align}
Moreover, for all $\br \in \mathbb{R}^{2m}$, there is a classical randomized algorithm that implements the oracle $\mathsf{Sample}(\calU_j^*(\hat D(\br)))$ with zero bias and magnitude
\begin{align}\label{appeq:mag_physical_oracle}
  \frac{1}{\sqrt{24\sqrt{\eta}\abs{\gamma} |q_1|}} \exp\left[- \left(\frac{1}{2} + \bar n\right)(1-\eta^2)q_1^2\right] \sqrt{1/((1/2 + \bar n)(1-\eta))}.
\end{align}
\end{lemma}
\begin{proof}
    \begin{eqnarray}
    \calU_j^*(\hat D(\br)) &\coloneqq&  \frac{\exp(-(1/2 + \bar n)(1-\eta)(q_1^2 + p_1^2))}{\sqrt{24 \pi \gamma \sqrt{\eta} q_1}} \int \mathrm d\tilde{p}_1 \exp\left(i \left(\frac{\left(\sqrt{\eta} p_1+\bar{p}_1\right)^2}{24  \gamma \sqrt{\eta} q_1} - 4 \gamma \eta^{3/2} q_1^3 - \frac{\pi}{4}\right)\right) \nonumber \\&& \hspace{60mm} \times \exp\left(-(1/2 + \bar n)(1-\eta)(\eta q_1^2 + \bar p_1^2) + i \bm{d}_j^T \Omega \tilde{\br}\right) D(S_j^{-1}\bm{\tilde r}),  \nonumber \\
\end{eqnarray}
where
\begin{equation}
    \bm{\tilde r} = (\eta q_1, \sqrt{\eta} p_1, q_2,p_2,\dots,q_m,p_m).
\end{equation}
Therefore,
\begin{equation}
\fnormone{\calU_j^*(\hat D(\br))} \leq   \frac{1}{\sqrt{24\sqrt{\eta}\abs{\gamma} |q_1|}} \exp\left[- \left(\frac{1}{2} + \bar n\right)(1-\eta^2)q_1^2\right] \sqrt{1/((1/2 + \bar n)(1-\eta))}  
\end{equation}
Further,
\begin{eqnarray}
    \Tr[\mathcal{U}_j^*(D(\br))] &= & \int \mathrm d\bar{p}_1 \frac{\exp\left(- \left(\frac{1}{2} + \bar n\right)(1-\eta)\left(\bar{p}_1^2  \right)\right)}{\sqrt{\pi/((1/2 + \bar n)(1-\eta))}} f(\bar p_1),
\end{eqnarray}
where
\begin{eqnarray}
    f(\bar p_1) &\coloneqq& \frac{\exp(-(1/2 + \bar n)(1-\eta)(q_1^2 + p_1^2))}{\sqrt{24\pi \gamma \sqrt{\eta} q_1}} \exp\left(i \left(\frac{\left(\sqrt{\eta} p_1+\bar{p}_1\right)^2}{24  \gamma \sqrt{\eta} q_1} - 4 \gamma \eta^{3/2} q_1^3 - \frac{\pi}{4}\right)\right) \nonumber \\
    && \times \exp\left(-(1/2 + \bar n)(1-\eta)\eta q_1^2 + i \bm{d}_j^T \Omega \tilde{\br}\right) \Tr[D(S_j^{-1}\bm{\tilde r}) \rho] \times \sqrt{\pi/((1/2 + \bar n)(1-\eta))}.
\end{eqnarray}
The magnitude of $f(\bar p_1) $ and consequently, the magnitude of sampling oracle can be bounded as follows:
\begin{align}
    \abs{f(\bar p_1)} \leq &\frac{1}{\sqrt{24\pi \sqrt{\eta}\abs{ \gamma  q_1}}} \exp\left[- \left(\frac{1}{2} + \bar n\right)(1-\eta)(1 + \eta)q_1^2\right] \sqrt{\pi/((1/2 + \bar n)(1-\eta))}.
    \end{align}
\end{proof}
\noindent These sampling oracles will be used in Section \ref{appsec:unbiased_estimator} where we employ unbiased estimators to estimate characteristic functions. Before that, however, we want to bound the error that occurs when we approximate the cubic phase gate as identity, when calculating characteristic functions. This is covered in the next section.

\subsection{Approximating the cubic phase gate as identity}
\label{section:biased-oracle}
Noting the expression of displacement operator $\hat D(q_1,p_1) = \exp(i p_1 \hat q - i q_1 \hat p)$ and the action of the cubic phase gate on the position and momentum quadrature operators (Eq.~\ref{appeq:cubic_action}), we see that for $q_1 \rightarrow 0$ or $\gamma \rightarrow 0$, the action of the cubic phase gate is the identity operation. Whereas, we cover the case of $\gamma$  approaching zero in Section \ref{appsec:near-Gaussian_approx}, this section is based on the intuition that when $|q_1|$ is very close to zero, we can approximate the cubic phase gate with the identity operator by introducing a small error. This is formalized in the following Lemma:

\begin{lemma}[Cubic gate approximation scheme]\label{applem:cubic_approx_small_q}
Given $M \geq 0$, let $\calS_M$ be a set of $m$-mode quantum states such that
\begin{align}
\max_{\substack{ \rho \in \calS_M\\ \br' \in \mathbb{R}^{2m} }} \bigg |\left\{\frac{\partial}{\partial q_1}\Tr[\hat D(\br)\Lambda_{\bar n , \eta} \circ \mathcal C_\gamma \circ \Lambda_{\bar n,\eta} \circ \mathcal G_j(\rho)]\right\}_{\br =\br'} \bigg |\leq M.
\end{align}
For any $\br \in \mathbb{R}^{2m}$ and for any $\rho \in \calS_M$,
    \begin{align}
         \left\lvert \Tr[\Lambda_{\bar n, \eta}^*(\hat D(\br))\calC_\gamma \circ \Lambda_{\bar n, \eta} \circ \mathcal G(\rho)] -  \Tr[\Lambda_{\bar n, \eta}^*(\hat D(\br_0))\Lambda_{\bar n, \eta} \circ \mathcal G(\rho)]  \right\rvert
         \leq |q_1| M,
    \end{align}
    where $\br_0 = (0,p_1,q_2,p_2,\dots,q_m,p_m)$. 
\end{lemma}
 \begin{proof}
We first note that for a differentiable function $f$ defined over $\R$, we have
\begin{equation}
    f(b) - f(a) = \int_{x= a}^b f'(x) dx,
\end{equation}
and therefore, by the triangle inequality,
\begin{equation}
    |f(b) - f(a)| \leq |b-a| \max_{x\in[a,b]} |f'(x)| \leq |b-a| \max_{x\in\R} |f'(x)|.
\end{equation}
Therefore, for $\bm r=(q_1,p_1,\dots,q_m,p_m)])$ and $\bm r_0=(0,p_1,\dots,q_m,p_m)])$, taking $f(q_1)=\Tr[\Lambda_{\bar n, \eta}^*(\hat D(\br))\calC_\gamma \circ \Lambda_{\bar n, \eta} \circ \mathcal G(\rho)]$ we have
\begin{eqnarray}
     \left\lvert \Tr[\Lambda_{\bar n, \eta}^*(\hat D(\br))\calC_\gamma \circ \Lambda_{\bar n, \eta} \circ \mathcal G(\rho)] -  \Tr[\Lambda_{\bar n, \eta}^*(\hat D(\br_0))\mathcal C_\gamma \circ \Lambda_{\bar n, \eta} \circ \mathcal G(\rho)]  \right\rvert&& \nonumber \\ 
        && \hspace{-70mm} \leq \abs{ q_1} \max_{q_1\in \R} \left|\Tr\left[\left(\frac{\partial}{\partial q_1}\Lambda^{*}_{\bar n,\eta}(\hat D(\br))\right) C_\gamma \circ \Lambda_{\bar n, \eta} \circ \mathcal G( \rho)\right]\right| \leq |q_1|M.
\end{eqnarray}
Now, if the position coordinate of the displacement operator of the first mode is zero, the cubic phase gate acts as identity on this displacement operator, since $\hat D(0,p_1) = \exp(i p_1 \hat q_1)$, so
\begin{equation}
   \Tr[\Lambda_{\bar n, \eta}^*(\hat D(\br_0))\calC_\gamma \circ \Lambda_{\bar n, \eta} \circ \mathcal G(\rho)] = \Tr[\Lambda_{\bar n, \eta}^*(\hat D(\br) \circ \Lambda_{\bar n, \eta} \circ \mathcal G(\rho)].
\end{equation}
\end{proof}
Therefore, the error introduced by approximating the cubic phase gate as identity when computing a characterstic function can be bounded in terms of the position coordinate of the first mode of the displacement operator and the quantity $M$, which, through Lemma \ref{lem:curvature_bound} can be related to the first two quadrature moments of the state whose characteristic function we want to estimate.

We combine this approximation scheme with the sampling oracles defined in the previous section in Sections \ref{appsec:adaptive_symplectic} and \ref{appsec:high_cubicity} to obtain two regimes of efficient classical simulation of noisy bosonic circuits: (i) High value of Symplectic coherence and cubicity combined allows for efficient classical simualtion (Section \ref{appsec:adaptive_symplectic}) (ii) Sufficiently high values of cubicity (depending of energy and depth of the computation) allows for efficient classical simulation (Section \ref{appsec:high_cubicity}).

Having detailed the necessary preliminaries, we are now ready to introduce the family of bosonic circuits considered in this work, consisting of Gaussian unitary gates and cubic phase gates, with environmental interaction modeled by a single-mode thermal loss channel preceding and following the cubic phase gates (i.e.\ universal bosonic circuits with noisy cubic phase gates).


\section{Classical simulation of noisy bosonic circuits}
\label{app:c-sim}
\noindent Having detailed the required technical results, we are now ready to provide classical simulation algorithms for estimating expectation values of quantum observables at the output of a noisy bosonic circuit with $L$ layers of the following form:
\begin{equation}\label{appeq:noisy_bosonic}
    \calU \coloneqq \calU_L \circ \calU_{L-1}\circ \dots \circ \calU_1
\end{equation}
with 
\begin{align}
   \calU_j \coloneqq \Lambda_{\bar n, \eta}\circ \calC_j\circ \Lambda_{\bar n, \eta}\circ \calG_j,
\end{align}
where $\mathcal C_j$ is a cubic phase gate with cubicity $\gamma_j$, $\calG_j$ is the Gaussian unitary gate with symplectic matrix 
$S_j$ and displacement vector $\bm{d}_j$, and $\Lambda_{\bar n, \eta}$ is the thermal loss channel (see Figure \ref{fig:noisy_bosonic_circit}). 

As a first step, we devise an efficient classical simulation algorithm when the cubicity of the cubic phase gates is sufficiently small throughout i.e. for "near-Gaussian" bosonic circuits.

\subsection{Near-Gaussian circuits} \label{appsec:near-Gaussian_approx}
In this section, we develop efficient classical simulation algorithms for probability estimation and quadrature expectation value estimation at the output of noisy bosonic circuits described by Eq.~\ref{appeq:noisy_bosonic}, when the cubicity parameter throughout the circuit is sufficiently low (hence we call these circuits ``near-Gaussian circuits''). Before detailing the estimation algorithm, we establish the following technical result:

\begin{lemma}[Derivative bound with respect to the cubicity parameter]\label{applem:curvature_cubic}
  Given a quantum state $\rho$, displacement operator $\hat D(\br)$ with $r_1 = q_1 + i p_1$, and the cubic phase gate $\mathcal C_\gamma$ parameterized by $\gamma$, we have that
  \begin{align}
    \max_{\gamma \in \mathbb{R}}\bigg|{\frac{\partial}{\partial \gamma} \Tr[\calC_\gamma^*(\Lambda_{\bar n, \eta}^*(D(\br)))\rho]} \bigg| \leq g_{\rho,\alpha,\gamma}(\hat q),
\end{align} 
where $g_{\rho,\alpha,\gamma}(\hat q)$ is a function depending on the expectation value of the first four moments of the position quadrature of the first mode $\hat{q}_1$.  
\end{lemma}
\begin{proof}
For simplicity, we consider a single-mode displacement operator and note that the analysis is analogous for multi-mode displacement operators since the thermal noise channel and the cubic phase gate only act on the first mode. We have
    \begin{eqnarray}
        \calC_\gamma^*(\Lambda_{\bar n, \eta}^*(\hat D(r_1))) &=& \exp\left(-(1/2 + \bar n)(1-\eta)(q_1^2 + p_1^2)\right)\calC_\gamma^*(D(\sqrt{\eta}r_1)) \nonumber \\
        &=& \exp\left(-(1/2 + \bar n)(1-\eta)(q_1^2 + p_1^2)\right)\exp(i\sqrt{\eta} p_1 \hat q_1 - i \sqrt{\eta} q_1\hat p_1 - i 6\gamma \sqrt{\eta} q_1 \hat q^2).
    \end{eqnarray}
Therefore,
\begin{equation}
    \frac{\partial}{\partial \gamma} \Tr[\calC_\gamma^*(\Lambda_{\bar n, \eta}^*(\hat D(\br)))\rho] = \exp\left(-(1/2 + \bar n)(1-\eta)(q_1^2 + p_1^2)\right) \Tr\left[\left(\frac{\partial}{\partial \gamma} \exp(i\sqrt{\eta} p_1 \hat q_1 - i \sqrt{\eta} q_1\hat p_1 - i 6\gamma \sqrt{\eta} q_1 \hat q^2)\right) \rho\right].
\end{equation}
To find the derivative of $\exp(i\sqrt{\eta} p_1 \hat q_1 - i \sqrt{\eta} q_1\hat p_1 - i 6\gamma \sqrt{\eta} q_1 \hat q^2)$ with respect to $\gamma$, we first note that the derivative of $i\sqrt{\eta} p_1 \hat q_1 - i \sqrt{\eta} q_1\hat p_1 - i 6\gamma \sqrt{\eta} q_1 \hat q^2$ and $\exp(i\sqrt{\eta} p_1 \hat q_1 - i \sqrt{\eta} q_1\hat p_1 - i 6\gamma \sqrt{\eta} q_1 \hat q^2)$ do not commute, therefore to find the derivative, we use Duhamel's formula \cite[Eq.\ 12]{Bauer2013TimeOrdering} and we obtain
\begin{eqnarray}
    &&\frac{\partial}{\partial \gamma} \exp(i\sqrt{\eta} p_1 \hat q_1 - i \sqrt{\eta} q_1\hat p_1 - i 6\gamma \sqrt{\eta} q_1 \hat q^2) \nonumber \\ \nonumber &&\quad= \int_{s=0}^1\mathrm ds \exp(is(\sqrt{\eta} p_1 \hat q_1 -  \sqrt{\eta} q_1\hat p_1 - 6\gamma \sqrt{\eta} q_1 \hat q^2)) \frac{\partial}{\partial \gamma}(i(\sqrt{\eta} p_1 \hat q_1 -  \sqrt{\eta} q_1\hat p_1 - 6\gamma \sqrt{\eta} q_1 \hat q^2)) \\
    &&\quad\quad\quad\quad\times\exp(-is(\sqrt{\eta} p_1 \hat q_1 -  \sqrt{\eta} q_1\hat p_1 - 6\gamma \sqrt{\eta} q_1 \hat q^2)) \exp(i\sqrt{\eta} p_1 \hat q_1 -  i\sqrt{\eta} q_1\hat p_1 - i 6\gamma \sqrt{\eta} q_1 \hat q^2) \nonumber \\
    &&\quad = -6i \sqrt{\eta} q_1\int_{s=0}^1\mathrm ds \left(\hat q^2 - is \sqrt{\eta} q_1 [\hat p, \hat q^2] + \frac{(-is\sqrt{\eta} q_1)^2}{2}[\hat p, -4i\hat q]\right) \exp(i\sqrt{\eta} p_1 \hat q_1 -  i\sqrt{\eta} q_1\hat p_1 - i 6\gamma \sqrt{\eta} q_1 \hat q^2) \nonumber \\
    &&\quad= -6i \sqrt{\eta} q_1 \int_{s=0}^1\mathrm ds (\hat q^2 - 4s\sqrt{\eta}q_1 \hat q + 4s^2 \eta  q_1^2) \times \calC_\gamma^*(\hat D(\sqrt{\eta}r_1)) \nonumber \\
    &&\quad = -6i \sqrt{\eta} q_1 \left(\hat q^2 - 2\sqrt{\eta} q_1\hat q + \frac43 \eta q_1^2\right) \calC_\gamma^*(\hat D(\sqrt{\eta}r_1)).
\end{eqnarray}
Therefore,
\begin{eqnarray}
    \frac{\partial}{\partial \gamma} \Tr[\calC_\gamma^*(\Lambda_{\bar n, \eta}^*(\hat D(\br)))\rho] &=& - 6i \sqrt{\eta} q_1 \exp\left(-(1/2 + \bar n)(1-\eta)(q_1^2 + p_1^2)\right) \Tr\left[\left(\hat q^2 - 2\sqrt{\eta} q_1\hat q + \frac43 \eta q_1^2\right) \calC_\gamma^*(\hat D(\sqrt{\eta}r_1)) \rho\right] \nonumber \\
    &=& - 6i \sqrt{\eta} q_1 \exp\left(-(1/2 + \bar n)(1-\eta)(q_1^2 + p_1^2)\right) \Tr\left[\left(\hat q^2 - 2\sqrt{\eta} q_1\hat q + \frac43 \eta q_1^2\right) \hat D(\sqrt{\eta}r_1)\calC_\gamma(\rho)\right]. \nonumber
\end{eqnarray}
Using the triangle inequality, we get
\begin{eqnarray}
    \left|\frac{\partial}{\partial \gamma} \Tr[\calC_\gamma^*(\Lambda_{\bar n, \eta}^*(\hat D(\br)))\rho]\right| &\leq& 6\sqrt{\eta}|q_1|\exp\left(-(1/2 + \bar n)(1-\eta)(q_1^2 + p_1^2)\right) \nonumber \\ && \times \left( \left|\Tr[\hat q^2\hat D(\sqrt{\eta}r_1)\calC_\gamma(\rho)]\right| + 2 \sqrt{\eta}|q_1| \left|\Tr[\hat q \hat D(\sqrt{\eta}r_1)\calC_\gamma(\rho)]\right| + \frac43 \eta q_1^2 \left|\Tr\left[\hat D(\sqrt{\eta}r_1)\calC_\gamma(\rho)\right]\right| \right). \nonumber\\
\end{eqnarray}
To bound the third term, we use the fact that the magnitude of the characteristic function of a state is less or equal to $1$. To bound the first and second terms, we use Lemma \ref{lem:trace_second_moment} shown explicitly for the first term hereafter:
\begin{eqnarray}
    \left|\Tr[\hat q^2\hat D(\sqrt{\eta}r_1)\calC_\gamma(\rho)]\right| &\leq& \frac{1}{\sqrt{2}} \sqrt{\Tr\left[\{\hat q^2\hat D(\sqrt{\eta}r_1), \hat D(-\sqrt{\eta}r_1) \hat q^2\} \calC_\gamma(\rho)\right]} \nonumber \\
    &=& \frac{1}{\sqrt{2}} \sqrt{\Tr[(\hat q^4 + (\hat q + 2 \sqrt{\eta}q_1)^4) \calC_\gamma(\rho)]} \nonumber \\
    &=& \frac{1}{\sqrt{2}} \sqrt{\Tr[(\hat q^4 + (\hat q + 2 \sqrt{\eta}q_1)^4) \rho]}.
\end{eqnarray}
Finally, we get
\begin{eqnarray}
    \left|\frac{\partial}{\partial \gamma} \Tr[\calC_\gamma^*(\Lambda_{\bar n, \eta}^*(\hat D(r_1)))\rho]\right| &\leq& 6\sqrt{\eta}|q_1|\exp\left(-(1/2 + \bar n)(1-\eta)(q_1^2 + p_1^2)\right) \nonumber \\ && \times \left( \frac{1}{\sqrt{2}} \sqrt{\Tr[(\hat q^4 + (\hat q + 2 \sqrt{\eta}q_1)^4) \rho]}  + \sqrt{2\eta}|q_1| \sqrt{\Tr[(\hat q^2 + (\hat q + 2 \sqrt{\eta}q_1)^2) \rho]} + \frac43 \eta q_1^2 \right). \nonumber
\end{eqnarray}
\end{proof}
\noindent Therefore, if the first four quadrature moments of $\rho$ (or the first two energy moments) are bounded, then $\left|\frac{\partial}{\partial \gamma} \Tr[\calC_\gamma^*(\Lambda_{\bar n, \eta}^*(D(\br)))\rho]\right|$ is also uniformly bounded, $\forall \gamma \in \mathbb{R}$ and $\forall \br \in \R^{2m}$. 

Lemma \ref{applem:curvature_cubic} allows us to estimate the characteristic function of the output state of a noisy bosonic circuit efficiently up to arbitrary precision, provided that we can compute the characteristic function of the input state efficiently. This is formalized in the following Lemma:

\begin{lemma}[Estimation of characteristic functions of near-Gaussian circuits]\label{applem:near_G_characteristic}
    Given an input quantum state $\rho_0$ and a noisy bosonic circuit described by Eq.~\ref{appeq:noisy_bosonic} such that $\forall j \in [1,\dots,L-1]$,
    \begin{equation}
        \max_{\gamma \in \R, \br \in \R^{2m}} \bigg|{\frac{\partial}{\partial \gamma} \Tr[\calC_{j+1}^*(\Lambda_{\bar n, \eta}^*(\hat D(\br)))\Lambda_{\bar n,\eta}\circ\mathcal G_{j+1}\circ\mathcal{U}_j \circ \dots\circ \mathcal{U}_1(\rho)]} \bigg| \leq R,
    \end{equation}
    then $\forall  \br \in \R^{2m}$ there we can compute efficiently $\tilde{\br} \in \R^{2m}$ and $A$, with $|A| \leq 1$, such that
    \begin{equation}
        \left|\Tr[\hat D(\br) \mathcal U(\rho)] - A\Tr[\hat D(\tilde{\br})\rho_0]\right| \leq  L \gamma_{\max} R,
    \end{equation}
    where
    \begin{equation}
        \gamma_{\max} = \max_{j \in (1,\dots,L)} |\gamma_j|.
    \end{equation}
\end{lemma}
\begin{proof}
    The Lemma follows by approximating all the cubic phase gates by the identity operator. We first write
    \begin{equation}
        \Tr[\hat D(\br) \mathcal U(\rho_0)] = \Tr[\mathcal C_{L}^{*}(\Lambda_{\bar n, \eta}^{*}(\hat D(\br)))\Lambda_{\bar n, \eta}\circ \mathcal G_{L} \circ\mathcal U_{L-1} \circ \dots \circ \mathcal U_1(\rho_0)].
    \end{equation}
Then, recall that for a differentiable function $f$ defined over $\R$, we have
\begin{equation}
    f(b) - f(a) = \int_{x= a}^b f'(x) dx,
\end{equation}
and therefore, by the triangle inequality,
\begin{equation}
    |f(b) - f(a)| \leq |b-a| \max_{x\in[a,b]} |f'(x)| \leq |b-a| \max_{x\in\R} |f'(x)|.
\end{equation}
Hence,
\begin{eqnarray}
    \left|\Tr[\mathcal C_{L}^{*}(\Lambda_{\bar n, \eta}^{*}(\hat D(\br)))\Lambda_{\bar n, \eta}\circ \mathcal G_{L} \circ\mathcal U_{L-1} \circ \dots \circ \mathcal U_1(\rho_0)] -  \Tr[\Lambda_{\bar n, \eta}^{*}(\hat D(\br))\Lambda_{\bar n, \eta}\circ \mathcal G_{L} \circ\mathcal U_{L-1} \circ \dots \circ \mathcal U_1(\rho_0)]\right|&& \nonumber \\   \leq |\gamma_L| \bigg|{\frac{\partial}{\partial \gamma} \Tr[\calC_{j+1}^*(\Lambda_{\bar n, \eta}^*(\hat D(\br)))\Lambda_{\bar n,\eta}\circ\mathcal G_{j+1}\circ\mathcal{U}_j \circ \dots\circ \mathcal{U}_1(\rho)]} \bigg|_{\gamma \in [0,\gamma_L]} \leq \gamma_{\max} R.
\end{eqnarray}
Now,
\begin{eqnarray}
    \nonumber&&\Tr[\Lambda_{\bar n, \eta}^{*}(\hat D(\br))\Lambda_{\bar n, \eta}\circ \mathcal G_{L} \circ\mathcal U_{L-1} \circ \dots \circ \mathcal U_1(\rho_0)]\\
    &&\quad\quad= \exp(-(\frac12 + \bar n)(1-\eta^2)(q_1^2 + p_1^2) + i \bm{d}_L^T \Omega \bar{\br}_L) \Tr[\hat D(\br_{L-1})\, \mathcal U_{L-1}\circ\dots\circ\mathcal U_1(\rho_0)],
\end{eqnarray}
where
\begin{eqnarray}
    \bar{\br} &=& (\eta q_1, \eta p_1, q_2, p_2,\dots,q_m,p_m), \nonumber \\
    \br_{L-1} &=& S_L^{-1} \bar{\br}.
\end{eqnarray}
For the next layer, we repeat the same process, and so on for the remaining layers. With the triangle inequality the total error is bounded by $L\gamma_{\max} R$, which concludes the proof.
\end{proof}
\noindent The algorithm for estimating $\Tr[\mathcal U(\rho_0)\hat D(\br)]$ for $\br \in \mathbb{R}^{2m}$ also allows us to estimate expectation values of local projectors. The intuition is as follows: given that
\begin{equation}
  \Tr[\mathcal U(\rho_0)\hat O] = \frac{1}{\pi^m} \int_{\mathbb{R}^{2m}}\mathrm d^{2m}\br \chi_{\hat O} (\br) \chi^{*}_{\mathcal U(\rho_0)}(\br),
\end{equation}
sampling a phase-space point using the (normalized) magnitude of characteristic function of the output operator ($|\chi_O(\br)|/(\pi^m \fnormone{\hat O})$) as the probability distribution and estimating the characteristic function at that point using Lemma \ref{applem:near_G_characteristic} allows us to calculate the expectation value of these output operators. This is formalized in the following Lemma:

\begin{lemma}[Estimation of expectation value of local projectors in near-Gaussian circuits]\label{applem:nearG_proj_estimation}
Let $\rho_0$ be an initial state and $\calU \coloneqq \calU_L \circ \, \calU_{L-1} \circ \, \dots \circ \, \calU_1$ be a noisy bosonic circuit (as in Eq.~\ref{appeq:noisy_bosonic}). For all $j\in [L-1]$, assume that the partially evolved state $\rho_j \coloneqq \calU_j \circ\, \calU_{j-1} \circ\, \dots \circ\, \calU_1(\rho_0)$ satisfies the following derivative bound:
\begin{align}
\max_{\substack{ \br' \in \mathbb{R}^{2m},\gamma \in \R }} \bigg|{\frac{\partial}{\partial \gamma} \Tr[\calC_{j+1}^*(\Lambda_{\bar n, \eta}^*(\hat D(\br)))\Lambda_{\bar n,\eta}\circ\mathcal G_{j+1}(\rho_j)]} \bigg| \leq R
\end{align}
and the cubic phase gates are such that 
\begin{equation}
    \gamma_{\max} = \max_{j \in [L]}|\gamma_j| \leq \frac{\epsilon}{2^k L R}.
\end{equation}
Then, for all $\br \in \mathbb{R}^{2m}$, there is a classical randomized algorithm running in time
\begin{align}
    4^k m\epsilon^{-2}\log(1/\delta)
\end{align} that approximates $\Tr[\hat O\calU(\rho_0)]$ for $\hat O = \left(\bigotimes_{i=1}^k \ket{\alpha_i}\bra{\alpha_i}\right)\otimes  \mathbb{I}^{\otimes(m-k)} $ with additive error $\epsilon$ and success probability $1-\delta$.
\end{lemma}
\begin{proof}
    As noted before (Eqs. \ref{eq:optical_equivalence_1} and \ref{eq:optical_equivalence_2}),
    \begin{equation}
        \Tr[\mathcal U(\rho_0)\hat O] = \frac{1}{\pi^m} \int_{\mathbb{R}^{2m}}\mathrm d^{2m}\br \chi_{\hat O} (\br) \Tr[\mathcal U(\rho_0)\hat D(-\br)] = \fnormone{\hat O}\mathbb{E}_{p_{\hat O}}\left[\mathrm{arg}(\chi_{\hat O}(\br))\Tr[\mathcal U(\rho_0)\hat D(-\br)]\right].
    \end{equation}
    From Lemma \ref{applem:near_G_characteristic}, $\forall \br \in \R^{2m}$ we can efficient compute $A(-\br) \leq 1$ and $\tilde{\br}$ such that
    \begin{eqnarray}
        &&\left|\fnormone{\hat O}\mathbb{E}_{p_{\hat O}}\left[\mathrm{arg}(\chi_{\hat O}(\br))\Tr[\mathcal U(\rho_0)\hat D(-\br)]\right] -  \fnormone{\hat O}\mathbb{E}_{p_{\hat O}}\left[\mathrm{arg}(\chi_{\hat O}(\br))A(-\br)\Tr[\mathcal\rho_0\hat D(\tilde{\br})]\right]\right| \nonumber \\ &&\leq \fnormone{\hat O}\mathbb{E}_{p_{\hat O}}\left[\left|\Tr[\mathcal U(\rho_0)\hat D(-\br)] - A(-\br)\Tr[\mathcal\rho_0\hat D(\tilde{\br})]\right|\right] \nonumber\\
        &&\leq \fnormone{O} \gamma_{\max}LR.
    \end{eqnarray}
    Therefore, if $\gamma_{\max} \leq \epsilon/(2\fnormone{\hat O} LR)$, the approximation error is at most $\epsilon/2$. 

    Sampling $\br$ from $p_{\hat O}$ $N$ times and taking the statistical average of the estimator $\fnormone{\hat O} \mathrm{arg}(\chi_{\hat O}(\br)) A(-\br)\Tr[\mathcal\rho_0\hat D(\tilde{\br})]$ at these points, from the Chernoff-Hoeffding bound (Lemma \ref{applem:CH}), we guarantee that this statistical average will be $\epsilon$-close to $\Tr[\mathcal U(\rho_0)\hat O]$ with failure probability $1-\delta$ as long as number of samples:
    \begin{eqnarray}
        \mathcal O\left(\fnormone{\hat O}^2 \epsilon^{-2} \log(1/\delta)\right) = \mathcal O\left(4^k \epsilon^{-2} \log(1/\delta)\right) , 
    \end{eqnarray}
    where we have used the fact that $\fnormone{\hat O} = 2^k$ for $\hat O  = \left(\bigotimes_{i=1}^k \ket{\alpha_i}\bra{\alpha_i}\right)\otimes  \mathbb{I}^{\otimes(m-k)} $ and that $|\mathrm{arg}(\chi_{\hat O}(\br)) A(-\br)\Tr[\mathcal\rho_0\hat D(\tilde{\br})]| \leq 1, \forall \br \in \R^{2m}$. Multiplying the time $\mathcal O(m)$ to get one sample gives the required time complexity.
\end{proof}
Also, combining Lemmas \ref{applem:quad_finite_diff} and \ref{applem:near_G_characteristic} allows estimating quadrature moments at the output of noisy bosonic circuits with cubic phase gates of low cubicity, provided there are some additional constraints on the output state of the circuit. This is formalized in the following Lemma:

\begin{lemma}[Quadrature moment estimation for near-Gaussian circuits]\label{applem:low_cubicity_quad_estimation}
Let $\rho_0$ be an initial state and $\calU \coloneqq \calU_L \circ \, \calU_{L-1} \circ \, \dots \circ \, \calU_1$ be a noisy bosonic circuit (Eq.~\ref{appeq:noisy_bosonic}). For all $j\in [L]$, assume that the partially evolved state $\rho_j \coloneqq \calU_j \circ\, \calU_{j-1} \circ\, \dots \circ\, \calU_1(\rho_0)$ satisfies the derivative bound:

\begin{align}
\max_{\substack{ \br' \in \mathbb{R}^{2m},\gamma \in \R }} \bigg|{\frac{\partial}{\partial \gamma} \Tr[\calC_{j+1}^*(\Lambda_{\bar n, \eta}^*(\hat D(\br)))\Lambda_{\bar n,\eta}\circ\mathcal G_{j+1}(\rho_j)]} \bigg| \leq R
\end{align}   
and up to sixth moment of position and momentum quadratures are upper bounded by $E$ on the output state $\mathcal U(\rho_0)$, then $\Tr[\hat q_j \mathcal U(\rho_0)], \Tr[\hat q_j^2 \mathcal U(\rho_0)], \Tr[\hat p_j \mathcal U(\rho_0)], \Tr[\hat p_j^2 \mathcal U(\rho_0)]$ can be computed efficiently with precision $\epsilon$ as long as the cubicity of cubic phase gates satisfies
\begin{equation}
\gamma_{\max} = \max_{j \in [L]}|\gamma_j| \leq \min\left(\frac{\epsilon^2}{LR(E+1)^2},\frac{81\epsilon^4}{LR(8E+12)^4}\right).
\end{equation}
\end{lemma}
\begin{proof}
    We give the proof for computation of $\Tr[\mathcal U(\rho_0) \hat q_j]$ and $\Tr[\mathcal U(\rho_0) \hat q_j^2]$. The proof for the computation of the momentum quadrature follows similarly.
    From Lemma \ref{applem:quad_finite_diff},
    \begin{equation}
        \left|\Tr[\mathcal U(\rho_0) \hat q_j] - \frac{\Tr[\mathcal U(\rho_0) \hat D_j(0,\delta)] - 1}{i\delta}\right| \leq \delta E.
    \end{equation}
    Further, from Lemma \ref{applem:near_G_characteristic}, $\forall \br \in \R^{2m}$, we can compute $A(\br)$ and $\tilde{\br}$ (here $\br$ is the zero vector except a non-zero $\delta$ in the momentum coordinate of the $j$'th mode) such that
    \begin{equation}
         \left|\Tr[\mathcal U(\rho_0) \hat q_j] - \frac{A(\br)\Tr[\rho_0 \hat D(\tilde{\br})] - 1}{i\delta}\right| \leq \delta E + \frac{L\gamma_{\max}R}{\delta}
    \end{equation}
    Therefore, for precision $\epsilon$, it is enough to impose $\gamma_{\max} = \delta^2/LR$ and $\delta \leq\epsilon/(E+1)$. This gives the bound on cubicity
    \begin{equation}
        \gamma_{\max} \leq \frac{\epsilon^2}{LR(E+1)^2}.
    \end{equation}
    Similarly, for $\Tr[\mathcal U(\rho_0) \hat q_j^2]$, from Lemma \ref{applem:quad_finite_diff},
    \begin{equation}
     \left| \Tr[\mathcal U(\rho_0) \hat q_j^2] - \frac{2}{\delta^2}\left((1 - \Tr[\mathcal{U}(\rho_0) \hat D_j(0,\delta)]) + \frac{1}{\delta}(\Tr[\mathcal{U}(\rho_0) \hat D_j(0,\delta^2)] - 1)\right)\right|   \leq \frac{8}{3}E \delta.
    \end{equation}
    From Lemma \ref{applem:near_G_characteristic}, for $\Tr[\mathcal{U}(\rho_0) \hat D_j(0,\delta)]$ we can compute $A(\delta)$ and $\tilde{\br}_\delta$, and similarly for $\Tr[\mathcal{U}(\rho_0) \hat D_j(0,\delta^2)]$, we can find $A(\delta^2)$ and $\tilde{\br}_{\delta^2}$ such that
    \begin{equation}
        \left| \Tr[\mathcal U(\rho_0) \hat q_j^2] - \frac{2}{\delta^2}\left((1 - A(\delta)\Tr[\rho_0\hat D(\tilde{\br}_\delta)]) + \frac{1}{\delta}(A(\delta^2)\Tr[\rho_0\hat D(\tilde{\br}_{\delta^2})] - 1)\right)\right|   \leq \frac83E \delta + \frac{2}{\delta^2}\gamma_{\max}LR + \frac{2}{\delta^3}\gamma_{\max}LR.
    \end{equation}
    Therefore, taking $\gamma_{\max} \leq \frac{\delta^4}{LR}$, then assuming $\delta < 1$:
    \begin{equation}
         \left| \Tr[\mathcal U(\rho_0) \hat q_j^2] - \frac{2}{\delta^2}\left((1 - A(\delta)\Tr[\rho_0\hat D(\tilde{\br}_\delta)]) + \frac{1}{\delta}(A(\delta^2)\Tr[\rho_0 \hat D(\tilde{\br}_{\delta^2})] - 1)\right)\right|   \leq \frac83 E \delta + 4\delta.
    \end{equation}
    Therefore, for precision $\epsilon$, we may pick $\delta \leq \frac{3\epsilon}{8E+12}$, which gives the bound on cubicity as
    \begin{equation}
        \gamma_{\max} \leq \frac{81 \epsilon^4}{(8E+12)^4 LR}.
    \end{equation}
\end{proof}
Therefore, as expected, when the non-Gaussian gates are approaching identity, the quadrature moments and overlaps with $k$ local coherent state projectors can be estimated efficiently. In what follows, we explore more subtle resource regimes where either the output of the noisy circuit is trivial in the sense that it concentrates to zero, or efficient classical simulation the circuit is possible, even though high amount of quantum computational resources are present.

\subsection{Noise-induced concentration in estimation of global observables}
\label{app:noise-induced-conc}
In this section, we prove concentration bounds for expectation values under noisy bosonic evolutions. Our strategy is to establish norm decay properties of linear maps acting on displacement operators and to propagate these through circuit layers. The first step is to quantify how the Fourier 1-norm of an operator evolves under noisy channels. This allows us to characterize the contractive behavior of noisy dynamics and is summarized in the following general Lemma:

\begin{lemma}[1-norm decay in the Fourier-Weyl basis]
\label{lem:fubini_applied}
  Let $\calL$ be a linear map such that $\fnormone{\calL(\hat D(\br))} \leq \chi$ for some constant $\chi$, for almost every $\br \in \mathbb{R}^{2m}$ (i.e., for all $\br \in \mathbb{R}^{2m}$ except on a subset of measure zero)
  Let $O$ be an operator whose characteristic function $\chi_O$ is an $L^1$-integrable function.
  Then the characteristic function of $\calL(O)$ is an $L^1$-integrable function, and moreover
  \begin{align}
      \fnormone{\calL(O)} \leq \chi \fnormone{O}.
  \end{align} 
\end{lemma}
\begin{proof}
We start by expanding $\calL(O)$ in the Fourier-Weyl basis
    \begin{align}
    \calL(O) = &\frac{1}{\pi^m}\int_{\mathbb{R}^{2m}} \mathrm d\br' \Tr[O\hat D(\br')]\calL(\hat D(-\br')).
\end{align}
And,
\begin{equation}
    \Tr[\calL(O) \hat D(\br)] = \frac{1}{\pi^m}\int_{\mathbb{R}^{2m}}\mathrm d\br' \Tr[O\hat D(\br')]\Tr[\calL(\hat D(-\br'))\hat D(\br)].
\end{equation}
Therefore we get,
\begin{eqnarray}
    \fnormone{\calL(O)} &=& \frac{1}{\pi^m} \int\mathrm d\br \left|\Tr[\calL(O) \hat D(\br)]\right| \nonumber \\
    &=& \frac{1}{\pi^m} \int\mathrm d\br \left|\frac{1}{\pi^m}\int_{\mathbb{R}^{2m}}\mathrm d\br' \Tr[O\hat D(\br')]\Tr[\calL(\hat D(-\br'))\hat D(\br)]\right| \nonumber\\
    &\leq& \frac{1}{\pi^m} \int\mathrm d\br' \left|\Tr[O\hat D(\br')]\right| \frac{1}{\pi^m}\int \mathrm d\br \left|\Tr[\calL(\hat D(-\br'))\hat D(\br)]\right| \nonumber \\
    &\leq& \chi \left(\frac{1}{\pi^m} \int\mathrm d\br' \left|\Tr[O\hat D(\br')]\right|\right) = \chi \fnormone{O}\end{eqnarray}
where the third line follows from the triangle inequality and Fubini--Tonelli's theorem (Lemma\ \ref{lem:fubini}) and the final inequality follows from the fact that $\fnormone{\calL(\hat D(\br))} \leq \chi$ for almost every $\br \in \mathbb{R}^{2m}$.
\end{proof}


Having established a generic 1-norm decay bound, we now analyze how individual noisy circuit layers act on displacement operators. Specifically, we introduce linear maps 
$\calB_j$ describing the combination of cubic, Gaussian, and noisy evolutions, upper bound their Fourier 1-norm and ultimately show that they admit a strong contraction property for sufficiently strong noise.

\begin{lemma}\label{lem:schrodinger_middle}
For $j = 2,\dots,t-1$, we consider the linear map $\calB_j$ defined as follows
    \begin{equation}\label{appeq:B_j}
        \mathcal{B}_j = \mathcal C_{j,\mathfrak D} \circ \Lambda_{\bar n,\eta} \circ \mathcal G_{j} \circ \Lambda_{\bar n,\eta} \circ \mathcal C_{j-1,\mathfrak C},
    \end{equation}
where the given bosonic quantum channels are introduced in Section \ref{appsec:decompositions}. For all $\br = (q_1,p_1,\dots, q_m,p_m) \in \mathbb{R}^{2m}$ such that $q_1\neq 0$, it holds that
    \begin{equation}
        \fnormone{\calB_j(\hat D(\br))}\leq \frac{\Gamma(1/4)}{\sqrt{24\pi}\sigma_j|\gamma_j|^{1/2}\eta}\times ((1/2 + \bar n)(1-\eta))^{-1/4},
    \end{equation}
    where
    \begin{equation}
        \sigma_j = |(S_{j})_{q_1,p_1}|.
    \end{equation}
\end{lemma}
\begin{proof}
The evolved operator $\calB_j(\hat D(\br))$ can be expressed as follows:
    \begin{eqnarray}
        \calB_j(\hat D(\br)) = \frac{1}{\sqrt{24 \pi \gamma_j} \eta^{7/4}}\int_{\mathbb{R}} \mathrm d\bar{p}_1 &&\frac{\exp\left(i\left(-\frac{(p_1 - \bar{p}_1)^2}{24\gamma_{j-1} q_1} + 4\gamma_{j-1} q_1^3 + \pi/4 - \bm{d}_j^T\Omega S_j\tilde{\br}\right)\right)\exp\left(-(1/2 + \bar n)(1/\eta - 1)(\tilde q_1^2 + \tilde p_1^2)\right)}{\sqrt{\tilde q_1}} \nonumber \\
        &&\times \exp(-(1/2 + \bar n)(1/\eta - 1)(q_1^2 + \bar p_1^2)) \hat D\left(\tilde{r}/\sqrt{\eta}\right),
    \end{eqnarray}
    where
    \begin{eqnarray}
        \tilde r &=& S_j(q_1/\sqrt{\eta},\bar p_1/\sqrt{\eta},q_2,p_2,\dots,q_m,p_m), \nonumber \\
        \bar r &=& (q_1/\sqrt{\eta},\bar p_1/\sqrt{\eta},q_2,p_2,\dots,q_m,p_m),
    \end{eqnarray}
   and $S_j$ is the symplectic matrix associated with $\mathcal G_j$. Therefore, defining $\sigma_j = |(S_{j})_{q_1,p_1}|$, we get
    \begin{eqnarray}
       \fnormone{\calB_j(\hat D(\br))} &\leq& \frac{1}{\sqrt{24\pi}|\gamma_j|^{1/2 }\eta^{7/4}} \int_{\mathbb{R}}\mathrm d\bar p_1 \frac{\exp(-(1/2 + \bar n)(1/\eta - 1)((\sigma_j/\sqrt{\eta}) \bar p_1 + (...))^2)}{\sqrt{(\sigma_j/\sqrt{\eta}) \bar p_1 + (...)|}} \nonumber \\&=&\frac{\Gamma(1/4)}{\sqrt{24\pi}\sigma_j|\gamma_j|^{1/2}\eta}\times ((1/2 + \bar n)(1-\eta))^{-1/4},
    \end{eqnarray}
    where $(...)$ in the second line consists of terms depending on $(q_1,p_1,\dots,q_m,p_m)$ and where we have used Lemma \ref{applem:gamma_integral}.
\end{proof}
\noindent We also require an analogous bound to control the evolution across the final noisy layer. The next lemma provides such a bound.
\begin{lemma}\label{lem:schrodinger_final}
    Given an operator $O$ and the linear map
    \begin{equation}\label{appeq:B_t}
        \calB_L = \Lambda_{\bar n,\eta}\circ\calG_L\circ\Lambda_{\bar n,\eta}\circ\calC_{L-1,\mathfrak C},
    \end{equation}
    we have that
    \begin{equation}
        \fnormone{ {\calB_L(\hat D(\br))}}  \leq \frac{\sqrt\pi}{\sigma_L\eta}((1/2+\bar n)(1 - \eta))^{-1/2},
    \end{equation}
   $\forall \br \in (q_1,p_1,\dots,q_m,p_m)$, where
    \begin{equation}
        \sigma_t = |(S_t)_{q_1,p_1}|.
    \end{equation}
\end{lemma}
\begin{proof}
The evolved operator $\calB_L(\hat D(\br))$ can be expressed as follows:
    \begin{eqnarray}
        \calB_L(\hat D(\br))
        =\frac{1}{\eta^2} \int_{\mathbb{R}}\mathrm d\bar{p}_1 &&\exp\left(i\left(-\frac{(p_1 -\bar{p}_1)^2}{24 \gamma_{L-1} q_1}+4\gamma_{L-1}q_1^3 +\pi/4 - \bm{d}_L^T\Omega S_L \br\right)\right)\exp(-(1/2+\bar n)(1/\eta-1)(\tilde q_1^2 + \tilde p_1^2))  \nonumber\\ 
        &&\times \exp(- i \bm{d}_L^T \Omega S_L \bar{\br}) \exp(-(1/2 + \bar n)(1/\eta - 1)(q_1^2 + \bar p_1^2)) \hat D(\tilde{\br}/\sqrt{\eta}) ,
    \end{eqnarray}
    where
    \begin{eqnarray}
        \tilde r &=& S_L (q_1/\sqrt
        \eta,\bar p_1/\sqrt{\eta},q_2,p_2,\dots,q_m,p_m), \nonumber \\
        \bar r &=& (q_1/\sqrt{\eta},\bar p_1/\sqrt{\eta},q_2,p_2,\dots,q_m,p_m),
    \end{eqnarray}
    and $S_L$ is the symplectic matrix associated with the Gaussian channel $\mathcal G_L$. Defining $\sigma_L = |(S_L)_{q_1,p_1}|$,
    \begin{eqnarray}
        \nonumber\fnormone{ {\calB_L(\hat D(\br))}}  &&\leq \frac{1}{\eta^2}\int_{\mathbb{R}} d\bar p_1 \exp(-(1/2 + \bar n)(1/\eta - 1)((\sigma_{L}/\sqrt{\eta}) \bar p_1 + (...))^2) \\&&=  \frac{\sqrt\pi}{\sigma_L \eta}((1/2+\bar n)(1 - \eta))^{-1/2},
    \end{eqnarray}
    where $(...)$ in the second line are terms depending on $(q_1,p_1,\dots,q_m,p_m)$.
\end{proof}
Equipped with the above lemmas, we can now combine these ingredients to obtain a global concentration result. In particular, we show that expectation values of global observables decay exponentially with the number of layers, thereby justifying the noise-induced concentration regime described in Theorem \ref{theo:noise_induced_conc} in the main text.
\begin{theorem}\label{apptheo:global_noise_induced conc}
    Given the input state $\rho = \ket 0\! \bra 0^{\otimes m}$, an output global operator $O$, we have that 
    \begin{equation}
        \left| \Tr[\mathcal U(\rho)  O]\right| \leq  2^m \times \frac{((1/2+\bar n)(1-\eta))^{-1/2}}{\sqrt{2}\sigma_t \sqrt{\eta}} \left(\frac{\Gamma(1/4)}{\sigma_{\min}\eta \sqrt{24\pi}|\gamma|_{\min}}((1/2+\bar n)(1-\eta))^{-1/4}\right)^{L-1} \times \|O\|_1,
    \end{equation}
    where 
    \begin{eqnarray}
        \sigma_{\min} &=& \min_{j} \sigma_k, \nonumber \\
        |\gamma|_{\min} &=& \min_{j} |\gamma_j|.
    \end{eqnarray}
    Therefore, assuming that the contraction coefficient
    \begin{equation}
         \mathfrak c := \frac{\Gamma(1/4)}{\sigma_{\min}\eta\sqrt{24\pi|\gamma|_{\min}}}(1/2+\bar n)(1-\eta)^{-1/4}< 1
    \end{equation}
    and for sufficiently large $t = \Omega(m)$, the output expectation value exponentially concentrates to zero.
\end{theorem}
\begin{figure}
    \centering
\includegraphics[width=0.7\linewidth]{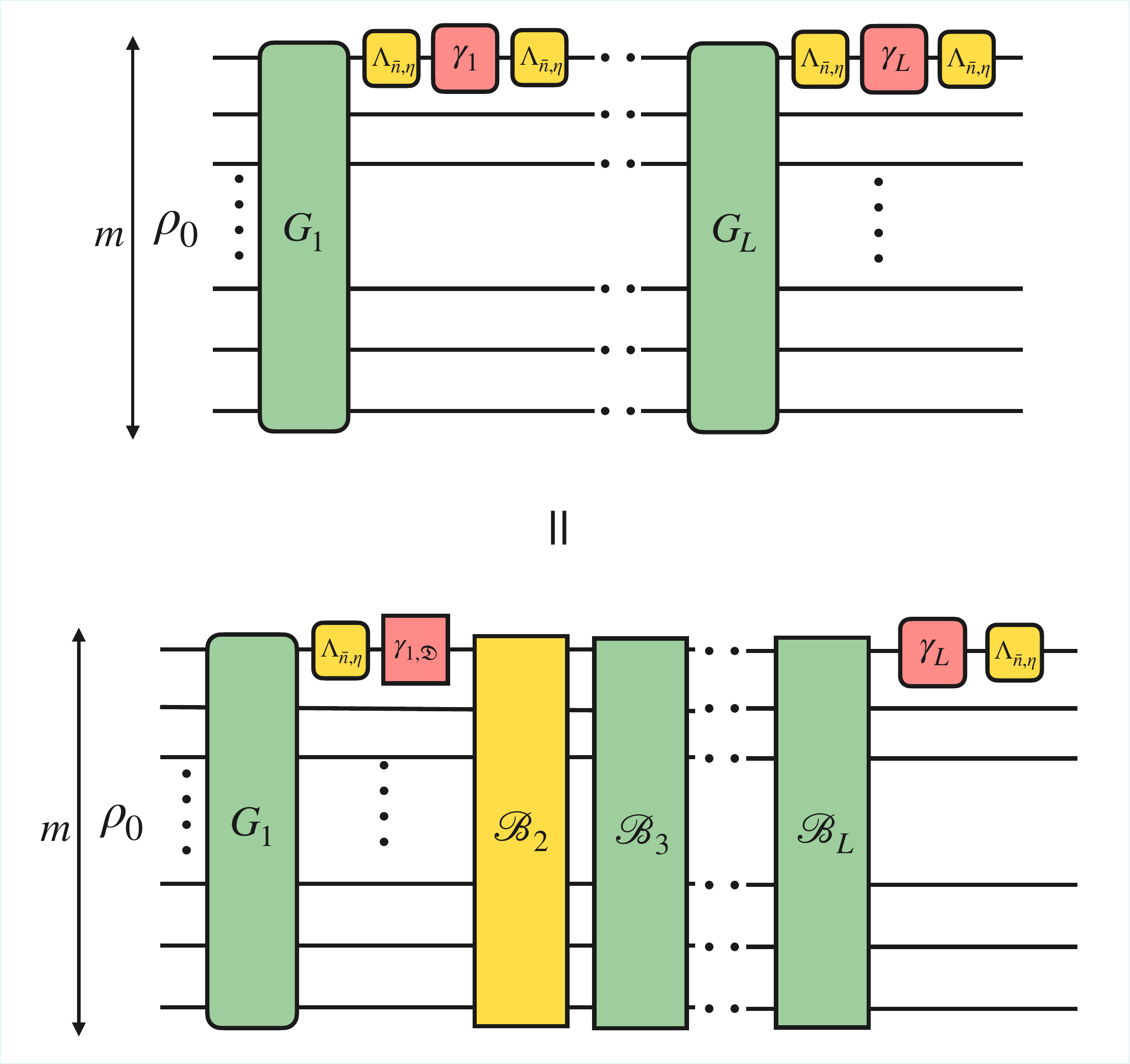}
    \caption{The circuit decomposition of the noisy bosonic circuit $\mathcal U$ (Eq.~\ref{appeq:noisy_bosonic}) that helps us in proving Theorem \ref{apptheo:global_noise_induced conc}. The expression for $\mathcal B_L$ and $\mathcal B_j, \forall j \in \{2,\dots,L-1\}$ are given by Equations \ref{appeq:B_t} and \ref{appeq:B_j} respectively. Conjugating the cubic phase gate with output operator $\hat O$ and upper bounding the Fourier one-norm of $\calB_{L}\circ \dots\circ\calB_2\circ \calC_{1,\mathfrak{D}}\circ \Lambda_{\bar n, \eta}\circ \calG_1 (\rho)$, we prove that the expectation value of output operators with at most exponential Schatten 1-norm concentrates exponentially to zero whenever $\mathfrak{c}_1 <1$ (Eq.~\ref{appeq:first_cc}) for sufficiently large circuit depth $L = \Omega(m)$.}
    \label{fig:global_conc_circuit}
\end{figure}
\begin{proof}
First we decompose the noisy bosonic circuit as
\begin{equation}
    \mathcal U = \Lambda_{\bar n,\eta} \circ\mathcal C_L   \circ \mathcal B_L \circ \mathcal B_{L-1}\circ \cdots \circ \mathcal B_2 \circ \mathcal C_{1,\mathfrak D}\circ \Lambda_{\bar n,\eta} \circ \mathcal G_1,
\end{equation}
where $\mathcal B_L$ and $\mathcal B_j, \forall j \in \{2,\dots,L-1\}$ are given by Equations \ref{appeq:B_t} and \ref{appeq:B_j} respectively (Figure \ref{fig:global_conc_circuit}). Then we have that
    \begin{align}
       \left| \Tr[O \calU(\rho)] \right| = &\left|\Tr[\calC_L^*(\Lambda_{\bar n,\eta}^*(O)) \, \calB_{L}\circ \dots\circ\calB_2\circ \calC_{1,\mathfrak{D}}\circ \Lambda_{\bar n, \eta}\circ \calG_1 (\rho)]\right|
        \\ \leq &\norm{\calC_L^*(\Lambda_{\bar n,\eta}^*(O)) }_1 \norm{\calB_{L}\circ \dots\circ\calB_2\circ \calC_{1,\mathfrak{D}}\circ \Lambda_{\bar n, \eta}\circ \calG_1 (\rho)}_\infty
        \\ \leq\frac{1}{\eta} &\norm{O}_1 \fnormone{\calB_{L}\circ \dots\circ\calB_2\circ \calC_{1,\mathfrak{D}}\circ \Lambda_{\bar n, \eta}\circ \calG_1 (\rho)},
    \end{align}
where, in the second line, we applied H\"older's inequality (Lemma~\ref{lem:holder}), and in the last line, we used Lemma \ref{applem:thermal_Schatten_one} and the unitary invariance of Schatten 1-norm of operators, and the fact that the operator norm is upper bounded by the Fourier 1-norm (Lemma~\ref{lem:haus-young}).
We now expand the operator $ \calC_{1,\mathfrak{D}}\circ \Lambda_{\bar n, \eta}\circ \calG_1 (\rho) $ int the displacement operator basis.
    \begin{align}
        \calC_{1,\mathfrak{D}}\circ \Lambda_{\bar n, \eta}\circ \calG_1 (\rho) = &\frac{1}{\pi^m}\int_{\mathbb{C}^m}\mathrm d\br \chi_{\rho}(\br) \frac{\exp(-(1/2+\bar n)(1/\eta - 1)(\tilde q_1^2 + \tilde p_1^2))}{\sqrt{24\pi\gamma_1 \tilde q_1} \eta^{3/4}} e^{-i\bm{d_1}^T\Omega\br} D\left(\frac{\tilde\br}{\sqrt{\eta}}\right),
    \end{align}
    where
    \begin{equation}
        \tilde r = S_1 (q_1,p_1,q_2,\dots,q_m,p_m).
    \end{equation}
The characteristic function of the operator $ \calC_{1,\mathfrak{D}}\circ \Lambda_{\bar n, \eta}\circ \calG_1 (\rho) $ is $L^1$, and moreover the associated Fourier 1-norm is upper bounded by
    \begin{align}
        \fnormone{\calC_{1,\mathfrak{D}}\circ \Lambda_{\bar n, \eta}\circ \calG_1 (\rho)}
        \leq &\frac{1}{\pi^m}\int_{\mathbb{C}^m}\mathrm d\br \chi_{\rho}(\br) \frac{\exp(-(1/2+\bar n)(1/\eta - 1)(\tilde q_1^2))}{\sqrt{24 \pi\abs{\gamma_1\tilde q_1} }\eta^{3/4}}\\  \leq
        & \frac{2^m\Gamma(1/4)}{\sqrt{48}\pi|\gamma_1|^{1/2}\eta^{1/2}\sigma_1} \times  \left((1/2 + \bar n)(1 - \eta)\right)^{-1/4},
    \end{align}
    where $\tilde \br = S_1 (q_1,p_1,\dots,q_m,p_m)$ in the first line, and the second line follows from the fact that for the input vacuum state
    \begin{equation}
        \chi_{\rho}(\br) = \Pi_{i=1}^m\exp\left(-\frac{q_i^2 + p_i^2}{2}\right).
    \end{equation}
    We use Lemma \ref{applem:gamma_integral} for integral over $p_1$ and standard Gaussian integral formula for integration over the rest of the coordinates. Using the fact that $\exp(-x^2) \leq 1 \;\forall x$ and carefully choosing which terms to discard, we obtain the given upper bound on $\fnormone{\calC_{1,\mathfrak{D}}\circ \Lambda_{\bar n, \eta}\circ \calG_1 (\rho)}$.
    
    We now bound the Fourier-one norm of $\calB_{L}\circ \dots\circ\calB_1\circ \calC_{1,\mathfrak{D}}\circ \Lambda_{\bar n, \eta}\circ \calG_1 (\rho)$ applying iteratively Lemma~\ref{lem:fubini}. In particular, applying $L-2$ times Lemma\ \ref{lem:schrodinger_middle} and Lemma\ \ref{lem:schrodinger_final}, we can upper bound the Fourier 1-norm of the evolved state as follows:
    \begin{align}
        \fnormone{\calB_{L}\circ \dots\circ\calB_1\circ \calC_{1,\mathfrak{D}}\circ \Lambda_{\bar n, \eta}\circ \calG_1 (\rho)} \leq & 2^m \times \frac{((1/2+\bar n)(1-\eta))^{-1/2}}{\sqrt{2}\sigma_t \sqrt{\eta}} \times\Pi_{i=1}^{L-1}\left(\frac{\Gamma(1/4)}{\sigma_i\eta\sqrt{24\pi|\gamma_i|}}((1/2+\bar n)(1-\eta))^{-1/4}\right) \\
        \leq & 2^m \times \frac{((1/2+\bar n)(1-\eta))^{-1/2}}{\sqrt{2}\sigma_t \sqrt{\eta}} \left(\frac{\Gamma(1/4)}{\sigma_{\min}\eta \sqrt{24\pi|\gamma|_{\min}}}((1/2+\bar n)(1-\eta)^{-1/4})\right)^{L-1},
    \end{align}
    and so 
    \begin{equation}
        \left| \Tr[O \calU(\rho)] \right| \leq 2^m \times \frac{((1/2+\bar n)(1-\eta))^{-1/2}}{\sqrt{2}\sigma_t \sqrt{\eta}} \left(\frac{\Gamma(1/4)}{\sigma_{\min}\eta \sqrt{24 \pi|\gamma|_{\min}}}((1/2+\bar n)(1-\eta)^{-1/4})\right)^{L-1} \times \|O\|_1.
    \end{equation}
\end{proof}
\noindent Therefore, for sufficiently lqrge depth $L = \Omega(m)$ of the noisy bosonic circuit with circuit parameters such that 
\begin{equation}\label{appeq:first_cc}
    \mathfrak c := \frac{\Gamma(1/4)}{\sigma_{\min}\eta\sqrt{24\pi|\gamma|_{\min}}}(1/2+\bar n)(1-\eta)^{-1/4}< 1,
\end{equation}
then the expectation value of the output state of the circuit for global operators with a finite Schatten 1-norm (a prominent example being global projectors) exponentially concentrates to zero, and hence the classical simulation of such overlaps becomes trivial.

In the next two sections, we give classical simulation algorithms identifying regimes where efficient classical simulation of overlap estimation with coherent state projectors and quadrature moment estimation is possible.
\subsection{Unbiased estimators}\label{appsec:unbiased_estimator}
In this section, we describe how to estimate the characteristic function at the output of the noisy bosonic circuit described by Eq.~\ref{appeq:noisy_bosonic} in two ways:
\begin{itemize}
    \item We assume that the inverse of the symplectic matrices associated with each of the Gaussian unitaries $G_j$, $S_j^{-1}$ all have non-zero symplectic coherence with respect to the first mode, i.e.\ they mix the position and momentum quadratures of the first mode \cite{upreti2025interplay}. From Eq.~\ref{appeq:sigma_sigma_inv}, $(S_j)_{q_1,p_1} = - (S_j^{-1})_{q_1,p_1}$ and hence we require that:
\begin{align}
    &\sigma_j \coloneqq \left|\left(S_j\right)_{q_1, p_1}\right| , 
    \\& \sigma_{\min} = \min_{j\in[L]} \sigma_j > 0.
\end{align}
    \item We assume nothing about the noisy bosonic circut.
\end{itemize}
At the end, we discuss the drawbacks of these estimators.
\subsection{Unbiased estimator for circuits with non-zero sympletic coherence}

\noindent Given the channel $\mathcal{U}_{\lfloor 1, t \rfloor}$ modeling the first $t$ circuit layers, i.e. 
\begin{align}
    \mathcal{U}_{\lfloor 1, t \rfloor} = \calU_t \circ \calU_{t-1}\circ \dots \circ \calU_1,
\end{align} 
its adjoint can be rewritten as
\begin{align} \label{appeq:first_t_layers}
    \mathcal{U}_{\lfloor 1, t \rfloor}^*&= \mathcal U_1^* \circ \mathcal U_2^* \circ \dots \circ \mathcal U_t^*\\  &= \bigcirc_{j=t-1}^{0}\; \calU_{t-j}^*
    \\ & = \bigcirc_{j=t-1}^{0}\;\left(\calG_{t-j}^*\circ \Lambda^{*}_{\bar n, \eta}\circ \calC_{t-j, \mathfrak{C}}^*\circ \calC_{t-j, \mathfrak{D}}^*
    \circ \Lambda^{*}_{\bar n, \eta}\right)
    \\ & =   \calG_1^* \circ \Lambda^{*}_{\bar n, \eta}\circ \calC_{1, \mathfrak{C}}^* \\ &\circ \left[ \bigcirc_{j=t-2}^{0}\;  \left( \calC_{{t-j-1}, \mathfrak{D}}^*
    \circ \Lambda^{*}_{\bar n, \eta} \circ \calG_{t-j}^*\circ \Lambda^{*}_{\bar n, \eta}\circ  \calC_{t-j, \mathfrak{C}}^*\right) \right]
    \\ &\circ \calC_{t, \mathfrak{D}}^*
    \circ \Lambda^{*}_{\bar n, \eta}
    \\ & \coloneqq \left[\bigcirc_{j=t-1}^{0}\;  \calA_{t-j} \right] \circ \calC_{t, \mathfrak{D}}^*
    \circ \Lambda^{*}_{\bar n, \eta}
\end{align}
Here, recall that the maps $\mathcal{A}_j$ are defined as
\begin{align}\label{appeq:A_maps}
    \calA_{j} \coloneqq \begin{cases}
    \calG_1^* \circ  \Lambda^{*}_{\bar n, \eta}\circ \calC_{1, \mathfrak{C}}^* & \text{if } j= 1,\\
        \calC_{{j-1}, \mathfrak{D}}^*
    \circ \Lambda^{*}_{\bar n, \eta} \circ \calG_{j}^*\circ \Lambda^{*}_{\bar n, \eta}\circ \calC_{j, \mathfrak{C}}^* & \text{if } 2 \leq j \leq t.
    \end{cases}
\end{align}
Expressing the noisy circuit layers in terms of the linear maps $\mathcal{A}_j$ is particularly convenient for classical simulation. As shown in Lemmas \ref{lem:oracleAj} and \ref{lem:oracleA1} of Section \ref{sec:unbiased-oracles}, these maps can be modeled by unbiased Monte Carlo estimators, which form the basis of the following technical result:

\begin{lemma}[Unbiased estimation of expectation values of displacement operators assuming non zero symplectic coherence of Gaussian gates]
\label{lem:unbiased-symplectic}
Assume that $\eta, \bar n, \abs{\gamma}_{\min} \in \Theta(1)$.
Let $\mathfrak{c}_2$ be defined as
\begin{align}
    \mathfrak{c}_2\coloneqq 
\frac{\Gamma(\tfrac14)}
     {\sigma_{\min}\eta^{3/4}\sqrt{24\pi|\gamma|_{\min}}\left\{\,(1/2+\bar n)(1-\eta)\right\}^{1/4}}.
\end{align}
Let $\br \in \mathbb{R}^{2m}$ such that $q_1 \neq 0$. Then, there is a randomized classical algorithm that runs in time \begin{align}
    \calO\left(mt\epsilon^{-2} \log(1/\delta)\abs{q_1}^{-1} \mathfrak{c}_2^{2(t-1)}\right)
\end{align} and approximates the expectation value $\Tr[D(\br) \calU_{\lfloor 1, t\rfloor}(\rho_0)]$ (Eq.~\ref{appeq:first_t_layers}) with additive error $\epsilon$ and success probability $1-\delta$.   
\end{lemma}
\begin{proof}
We will show that the procedure formalized in Algorithm\ \ref{alg:sim-unbiased} yields the desired approximation.
By Lemmas \ref{lem:oracleAj} and \ref{lem:oracleA1}, we can simulate the oracles associated to the linear maps $\calA_j$ with zero bias and magnitudes $A_j$ equal to

\begin{align}\label{appeq:A_j}
    A_{j} \coloneqq 
\begin{cases}
\sqrt{\pi}\,\bigl((\tfrac{1}{2}+\bar n)(1-\eta)\bigr)^{-1/2} & \text{if } j=1, \\
\dfrac{\Gamma(\tfrac14) \bigl((\tfrac{1}{2}+\bar n)(1-\eta)\bigr)^{-1/4}}{\sigma_j \,\eta^{3/4}\sqrt{\,24\pi\,|\gamma_{j-1}|\,}}
 & \text{if } 2 \leq j \leq t.
\end{cases}
\end{align}
By iteratively applying the ``Chained Sampling Lemma'' (Lemma\ \ref{applem:CHained-sampling}), we can simulate the oracle $\mathsf{Sample}(\bigcirc_{j=0}^{t-1}\;  \calA_{t-j} (D(\br)))$ with zero bias and magnitude $A$, where
\begin{align}
    A =\prod_{j=1}^t A_j = &
\frac{\sqrt{\pi}\;\Gamma(\tfrac14)^{\,t-1} \bigl((\tfrac12+\bar n)(1-\eta)\bigr)^{-\frac{t-1}{4}}}
     {(24\pi)^{\frac{t-1}{2}}\eta^{\frac34(t-1)}\;\prod_{j=2}^{t}(\sigma_j\,\sqrt{|\gamma_{j-1}|})}
     \\\leq  & \frac{\sqrt{\pi}\;\Gamma(\tfrac14)^{\,t-1} \bigl((\tfrac12+\bar n)(1-\eta)\bigr)^{-\frac{t-1}{4}}}
     {(24\pi)^{\frac{t-1}{2}} \eta^{\frac34(t-1)}\; \left(\sigma_{\min}\sqrt{|\gamma|_{\min}}\right)^{t-1} (\frac12+\bar n)^{1/2}(1-\eta)^{1/2}}.
\end{align}
Let $Y\in \mathbb{C}$, $\br'\in \mathbb{R}^{2m}$ such that
\begin{eqnarray}
    \calC_{t, \mathfrak{D}}^*
    \circ \Lambda^{*}_{\bar n, \eta}(\hat D(\br)) &&= \sqrt{\frac{1}{24\pi\gamma_t q_1}}
    \exp\Bigl(-(1+2\bar{n})(1-\eta)\tfrac{q_1^2 + p_1^2}{2}\Bigr)\hat D(\sqrt{\eta}q_1,\sqrt{\eta}p_1,q_2,p_2,\dots,q_m,p_m) \nonumber \\ &&=Y \times\hat D(\br'),
\end{eqnarray}
where we have set $Y \coloneqq \sqrt{\frac{1}{24\pi\gamma_t q_1 }}
    \exp\!\Bigl(-(1+2\bar{n})(1-\eta)\tfrac{q_1^2+p_1^2}{2}\Bigr)$ and $\br' = (\sqrt{\eta}q_1,\sqrt{\eta}p_1,q_2,p_2,\dots,q_m,p_m)$. Then, given $(Z,\boldsymbol{s})$ drawn from the oracle $\mathsf{Sample}(\bigcirc_{j=0}^{t-1}\;  \calA_{t-j} (\hat D(\br')))$, it holds that
\begin{align}
\mathbb{E}_{Z,\boldsymbol{s}}\left[YZ \Tr[\hat D(\boldsymbol{s})\rho_0]\right]=
\Tr[\left[\bigcirc_{j=0}^{t-1}\;  \calA_{t-j} \right] \circ \calC_{t, \mathfrak{D}}^*
    \circ \Lambda^{\mathrm{(Re)},*}_{\bar n, \eta}(\hat D(\br))\rho_0]
    = \Tr[\hat D(\br)\, \calU_{\lfloor 1, t\rfloor}(\rho_0)].
\end{align}
Defining
\begin{equation}\label{appeq:mathfrak_c}
    \mathfrak{c}_2\coloneqq 
\frac{\Gamma(\tfrac14)}
     {\sigma_{\min}\eta^{3/4}\sqrt{24\pi|\gamma|_{\min}}\left\{\,(1/2+\bar n)(1-\eta)\right\}^{1/4}},
\end{equation}
we observe that
\begin{equation}
    \abs{ YZ \Tr[\hat D(\boldsymbol{s})\rho_0]}\leq \abs{YZ} \leq A\sqrt{\frac{1}{24\pi\abs{\gamma}_{\min} \abs{q_1} }}
     \leq \frac{\sqrt{\pi}}{\sqrt{24\pi}\abs{\gamma}_{\min}^{1/2}|q_1|^{1/2}(1/2 + \bar n)^{1/2}(1-\eta)^{1/2}} \mathfrak c_2^{t-1}.
\end{equation}
By the Chernoff--Hoeffding bound (Lemma~\ref{applem:CH}), we can upper bound the number of draws $N$ from $\mathsf{Sample}(\bigcirc_{j=0}^{t-1}\;  \calA_{t-j} (\hat D(\br')))$
required to estimate the $\Tr[\hat D(\br) \calU_{\lfloor 1, t\rfloor}(\rho_0)]$ with additive error at most $\epsilon$ and success probability at least $1-\delta$, obtaining
\begin{align}
    N \in \calO\left(\epsilon^{-2}\log(1/\delta)\frac{\pi}{24\pi\abs{\gamma}_{\min}|q_1|(1/2 + \bar n)(1-\eta)} \mathfrak c_2^{2(t-1)}\right). 
\end{align}
In particular, this implies that
\begin{align}
    N \in \calO\left(\epsilon^{-2} \log(1/\delta)\abs{q_1}^{-1} \mathfrak{c}_2^{2(t-1)}\right).
\end{align}
As drawing a single sample from $\mathsf{Sample}(\bigcirc_{j=0}^{t-1}\;  \calA_{t-j} (\hat D(\br')))$ takes time $\calO(mt)$, the total runtime scales as 
\begin{align}
    \calO\left(mt\epsilon^{-2} \log(1/\delta)\abs{q_1}^{-1} \mathfrak{c}_2^{2t}\right).
\end{align}
\end{proof}

\begin{algorithm}[H]
\textbf{Parameters}: number of estimation rounds $N$ \\ 
\textbf{Input}: Classical descriptions of noisy circuit $\mathcal{U}_{\lfloor 1, t \rfloor} = \calU_t \circ \calU_{t-1}\circ \dots \circ \calU_1$ with non-zero symplectic coherence with respect to the first mode, classical description of an initial state $\rho_0$, and a vector $\br \in \mathbb{R}^{2m}$ such that $q_{1} \neq 0$\\
\textbf{Output}: An estimate of $\Tr[\hat D(\br)\,\mathcal{U}_{\lfloor 1, t \rfloor}(\rho_0)]$
\begin{algorithmic}[1]
    \State \textbf{Initialize}: 
        \State $\br_t \gets \br = (q_1,p_1,\dots,q_m,p_m) $
        \State $B_t \gets (24\pi \gamma_t q_1)^{-1/2}\exp\left(-(\tfrac12 + \bar n)(1-\eta)(q_1^2+p_1^2)\right)$
\For{ $j = 1 \text{ to } N $, }
    \For{ $k = t \text{ to } 1$, }
        \State $(\br_{k-1}, Z_{k-1}) \gets \mathsf{Sample}({\calA}_k(\br_{k}))$ \Comment{cf. Lemmas\ \ref{lem:oracleAj} and\ \ref{lem:oracleA1}}
        \State $B_{k-1} \gets  B_{k}\,{Z_{k-1}}$
    \EndFor

    \State  $X_j \gets  B_0\Tr[\hat D(\br_0) {\rho_0}]$
 \EndFor
 \State Return output $X \gets \frac{1}{N}\sum_{j=1}^N X_j$
 
\end{algorithmic}
\caption{Unbiased estimator for expectation values of displacement operators \label{alg:sim-unbiased}}
\end{algorithm}

\subsection{General unbiased estimator}
In this section, we devise an unbiased estimator for estimating characteristic function at the output of the noisy bosonic circuit, when we are not assuming property of the circuit. This is done by chaining the physical sampling oracle for $\mathcal U_j^*(\hat D(\br))$ given by Lemma \ref{applem:physical_oracle} and applying Lemma \ref{applem:CHained-sampling}. This is formalized in the following Lemma:

\begin{lemma}[General unbiased operator of characteristic function of displacement operators] \label{applem:unbiased_general}
    Given an input quantum state $\rho_0$ and a noisy bosonic circuit $\mathcal U$ given by Eq.~\ref{appeq:noisy_bosonic}, the characteristic function $\Tr[\hat D(\br)\mathcal U(\rho_0)]$ can be estimated with precision $\epsilon$ and success probability $1-\delta$ in time
    \begin{equation}
        O\left(\frac{mL\log(1/\delta)}{(24 \pi \sqrt{\eta} \gamma_{\min} |q|)^L} \epsilon^{-2} \right),
    \end{equation}
    where $q$ is a constant $\in \R$ and 
    \begin{equation}
        \gamma_{\min} = \min_j |\gamma_j|.
    \end{equation}
\end{lemma}
\begin{proof}
    We note that
    \begin{equation}
        \Tr[\hat D(\br) \mathcal U(\rho_0)] = \Tr[\bigcirc_{j=0}^{L} \mathcal U_j^*(\hat D(\br))\rho_0].
     \end{equation}
     Therefore, from Lemma \ref{applem:CHained-sampling}, by chaining the physical sampling oracles for each of the $\mathcal U_j$, given by Lemma \ref{applem:physical_oracle}, we obtain a sampling oracle for $\Tr[D(\br) \mathcal U(\rho_0)]$ with zero bias and magnitude upper bounded by $A^L$ with 
     \begin{equation}
         \frac{1}{\sqrt{24 \pi \sqrt{\eta} \gamma_{\min} |q|}}
     \end{equation}
     where $q$ is some constant $\in R$. $A$ is obtained by upper bounding the magnitude of the physical sampling oracle for each of the $\mathcal U_j^*$ given by Lemma \ref{applem:physical_oracle} and noting that $\exp(-x^2) \leq 1, \forall x \in \R$, and noting that we are considering all possible position coordinate of the first mode of the displacement operator whose characteristic function needs to be estimated.

     Therefore, by the Chernoff-Hoeffding bound (Lemma \ref{applem:CH}), by making $N$ queries to this sampling oracle, we can obtain an estimate of $\Tr[\hat D(\br) \mathcal U(\rho_0)]$ with precision $\epsilon$ and success probability $1-\delta$ as long as the number of samples
     \begin{equation}
         N = \mathcal O\left(A^{2L} \epsilon^{-2} \log(1/\delta)\right) = \mathcal O\left(\frac{\log(1/\delta)}{(24 \pi \sqrt{\eta} \gamma_{\min} |q|)^L} \epsilon^{-2} \right).
     \end{equation}
     Multiplying by the time $\mathcal O(mL)$ required to call the sampling oracle once gives us the required time complexity.
\end{proof}
\noindent At this point, we note that the runtime for the estimators given by both Lemma \ref{lem:unbiased-symplectic} and Lemma \ref{applem:unbiased_general} diverge as the position coordinate of the first mode of the displacement operator goes to zero. Therefore, we need better simulation algorithms for characteristic function that work efficiently $\forall q_1 \in \R$. This is done by combining the given estimators with the approximation scheme of cubic phase gates (Lemma \ref{applem:cubic_approx_small_q}) and is detailed in the next section.

\subsection{Adaptive algorithms}
We now turn to adaptive algorithms, which refine the unbiased estimation procedures introduced in the previous section. While the unbiased estimators already allows for the approximation of displacement operator expectation values, they suffers from the limitation that position coordinate of the first mode of the displacement operator needs to be sufficiently large throughout for them to be efficient. For the unbiased estimator with non-zero symplectic coherence (Lemma \ref{lem:unbiased-symplectic}), the condition that $q_1$ should be large enough is checked only once for initial displacement operator, for the general unbiased estimator, this condition on $q_1$ is checked on the evolved displacement operator after each layer of Gaussian gate plus noisy cubic phase gates. Further, the unbiased estimator given by Lemma \ref{lem:unbiased-symplectic} requires both the Gaussian unitaries to have non-vanishing symplectic coherence and the cubic gates to have non-vanishing cubicity. To overcome these drawbacks, we design two adaptive algorithms based on the approximation scheme introduced in Lemma \ref{applem:cubic_approx_small_q}. Both algorithms achieve a runtime independent of $\abs{q_1}$: the first still requires non-vanishing symplectic coherence and cubicity, whereas the second is tailored to regimes the cubicity is sufficiently high with respect to other parameters, such as the noise rate and the energy, and presents a runtime independent of the amount of symplectic coherence.

The first adaptive algorithm (given in Section \ref{appsec:adaptive_symplectic}) gives regimes of efficient classical simulation regimes for noisy bosonic circuits with high symplectic coherence combined with cubicity whereas the second adaptive algorithm (Section \ref{appsec:high_cubicity}) gives efficient classical simulation regimes for noisy bosonic circuit with high enough cubicity (and independent of symplectic coherence) with respect to energy and depth of the computation.

\subsubsection{Adaptive algorithm for circuits with non-zero symplectic coherence}\label{appsec:adaptive_symplectic}
To overcome the limitations of small initial value of $q_1$ when estimating characteristic function at the output of the noisy bosonic circuit whose Gaussian gates all have a non-zero symplectic coherence with respect to the first mode, instead of using the unbiased sampling oracle given by Lemma \ref{lem:unbiased-symplectic} in the beginning, we use the approximation scheme from Lemma \ref{applem:cubic_approx_small_q}, replacing the cubic phase gates as identity, and evolve the displacement operator deterministically for the first few layers to take care of the small value $q_1$. After we implement the approximation scheme for a layer, at the end of each such layer, we check the value of $\br$ of the evolved displacement operator and if its first coordinate $q_1$ is sufficiently large, we switch to the unbiased oracle given by Lemma \ref{lem:unbiased-symplectic} for the rest of the circuit, otherwise we use the approximation scheme for the next layer as well, and we repeat the process after each layer. This algorithm allows for simulation of the characteristic function with the following time complexity:
\begin{lemma}[Adaptive estimation of expectation values of displacement operators]\label{lem:adaptive_estimation_charateristic}
Let $\rho_0$ be an initial state and $\calU \coloneqq \calU_L \circ \, \calU_{L-1} \circ \, \dots \circ \, \calU_1$ be a noisy bosonic circuit. For all $j\in [L-1]$, assume that the partially evolved state $\rho_j \coloneqq \calU_j \circ\, \calU_{j-1} \circ\, \dots \circ\, \calU_1(\rho_0)$ satisfies the following derivative bound:
\begin{align}\label{derivative_bound}
\max_{\substack{ \br' \in \mathbb{R}^{2m} }} \bigg |\left\{\frac{\partial}{\partial q_1}\Tr[\Lambda_{\bar n,\eta^*}(\hat D(\br))\mathcal C_{j+1} \circ \Lambda_{\bar n , \eta} \circ \mathcal G_{j+1}(\rho_j)]\right\}_{\br =\br'} \bigg |\leq M.
\end{align}
For all $\br \in \mathbb{R}^{2m}$, there is a classical randomized algorithm running in time
\begin{align}
   \mathcal O(m\mathfrak c_2^{\ell}  ML^2\epsilon^{-3}\log(1/\delta)),
\end{align}
where
\begin{equation}
    \mathfrak c_2  :=
\frac{\Gamma(\tfrac14)}
     {\sigma_{\min}\eta^{3/4}\sqrt{24\pi|\gamma|_{\min}}\left\{\,(1/2+\bar n)(1-\eta)\right\}^{1/4}},
\end{equation}
and $0\leq \ell \leq 2L$, that approximates $\Tr[D(\br)\calU(\rho_0)]$ with additive error $\epsilon$ and success probability $1-\delta$.
\end{lemma}
\begin{proof}
We will show that the procedure formalized in Algorithm\ \ref{alg:sim-adaptive-general} yields the desired approximation.
Consider the expectation value $\Tr[\hat D(\br)\,\calU(\rho)]$. 
Let $\br^{(j)} = (q_1^{(j)},p_1^{(j)}, \dots, q_m^{(j)},p_m^{(j)})$ be defined iteratively as follows:
\begin{eqnarray}\label{appeq:adaptive_position_co}
    \br^{(L)} &=& \tilde{\br}^{(L)} =  \br,
    \\\bar{\br}^{(j)} &=& (0,\eta \tilde p_1^{(j)},\tilde q_2^{(j)},\tilde p_2^{(j)},\dots,q_m^{(j)},\tilde p_m^{(j)} ), 
    \\ \tilde{\br}^{(j)} &=& S_j^{-1} \bar{\br}^{j+1}
\end{eqnarray}
We also define the partially evolved state $\rho_j$ and the partially Heisenberg-evolved observable $O_j$
\begin{align}
    &\rho_j = \calU_j \circ \calU_{j-1} \circ \dots \circ \calU_1(\rho),
    \\&O_{j} = \calU^*_{j+1} \circ \dots \circ \calU_L^*(O),
\end{align}
where $ O_L = \hat D(\br)$. 
We denote by $\tilde O_j$ the observable obtained by approximating the cubic phase gates
\begin{align}\label{appeq:approximated_operator}
     \tilde O_j = \prod_{i={j}}^{L-1}{\exp\left(-(1/2 + \bar n)(1-\eta^2)\tilde p_{i+1}^2 + i \bm{d}_i^T \Omega \bar{\bm r}_{i+1}\right)} \hat D(\tilde{\br}^{(j)}),
\end{align}
and we set $\tilde O_L = O_L = \hat D(\br)$.
Given a threshold parameter $\tau \geq 0$, let $j^*$ be the largest $j$ such that
\begin{align}
    \left\lvert \tilde{q}_1^{(j)}\right \rvert\geq \tau. 
\end{align}
We assume that all the states $\rho_j$ with $j\geq j^*$ lies in the set $\calS_M$, i.e., they satisfy the derivative bound in Eq.\ \ref{derivative_bound}.
We have,
\begin{align}
    \left\lvert\Tr[\left(O_{j^*} - \tilde O_{j^*}\right)\rho_{j^*}] \right\rvert
    = &\left| \Tr[O_L\rho_L] - \Tr[\tilde{O}_{j^*}\rho_{j^*}] \right|
    = \left| \Tr[\tilde O_L\rho_L] - \Tr[\tilde{O}_{j^*}\rho_{j^*}] \right|
    \\ = & \left\lvert\sum_{j = j^{*}}^{L-1} \Tr[\tilde O_{j+1} \rho_{j+1}] - \Tr[\tilde O_{j}\rho_j]  \right\rvert
     \leq (L-j^*)  \tau M \leq \tau M L,
\end{align}
The first line follows from the identity $\Tr[O_{j^*}\rho_{j^*}] = \Tr[O_L\, \mathcal{U}_L \circ \mathcal{U}_{L-1} \circ \dots \circ \mathcal{U}_1(\rho)] = \Tr[O_{L}\rho_L]$. The second line follows by definition ($\tilde O_L = O_L = D(\br)$), the third line by telescoping, and the final step from the triangle inequality together with Lemma~\ref{applem:cubic_approx_small_q}.

\noindent For achieving additive error at most $\epsilon/2$, we can set $\tau = \epsilon/(2ML)$. We estimate 
\begin{equation}\label{appeq:modified_characteristic}
    \Tr[\widetilde{O}_{j^*}\rho_{j^*}] =\left( \prod_{i={j^*}}^{L-1}{\exp\left(-(1/2 + \bar n)(1-\eta^2)\tilde p_{i+1}^2 + i \bm{d}_i^T \Omega \bar{\bm r}_{i+1}\right)}\right) \Tr[\hat D(\tilde{\br}^{(j^*)}) \rho_{j^*}]
\end{equation}
with the unbiased estimator for estimating characteristic functions (Algorithm \ref{alg:sim-unbiased}). This can be done efficiently now since we guarantee that $|\tilde q_1^{j^*}|$ is above a certain threshold, since $|\tilde q_1^{j^*}| \geq \tau = \epsilon/(2ML)$ and therefore the unbiased estimator can be implemented efficiently (i.e. we do not encounter the problem of the divergence of the time complexity of unbiased estimator as the initial position coordinate of the first mode approaches zero). And by Lemma\ \ref{lem:unbiased-symplectic}, $\Tr[\widetilde{O}_{j^*}\rho_{j^*}]$ can be estimated with precision $\epsilon/2$ and success probability $1-\delta$ in time $\mathcal O( m(L - j^*)\epsilon^{-2}\log(1/\delta)|q_1|^{-1}\mathfrak c^{2(L-j^*-1)}) = \mathcal{O}(m(L-j^*)\epsilon^{-2}\log(1/\delta)/\tau \times \mathfrak c^{2(L- j^* - 1)})$. We note that since the prefactor of the characteristic function to be estimated in Eq.~\ref{appeq:modified_characteristic} is $\leq 1$, it does not factor into the time complexity of the unbiased estimator here. Inserting $\tau = \epsilon/2ML$ gives us the required expression.
\end{proof}
\noindent The circuit decomposition that helps us in devising the simulation algorithm used in Lemma \ref{lem:adaptive_estimation_charateristic} is given in Figure \ref{fig:adaptive_estimation_characteristic}.

\begin{algorithm}[H]
\textbf{Parameters}: number of estimation rounds $N$, threshold $\tau\geq 0$ \\ 
\textbf{Input}: Classical descriptions of noisy circuit $\mathcal{U} = \calU_L \circ \calU_{L-1}\circ \dots \circ \calU_1$ with non-zero symplectic coherence with respect to the first mode, classical description of an initial state $\rho_0$, and a vector $\br \in \mathbb{R}^{2m}$ (possibly with $q_{1} = 0$)\\
\textbf{Output}: An estimate of $\Tr[\hat D(\br)\,\mathcal{U}(\rho_0)]$
\begin{algorithmic}[1]
    \State \textbf{Initialize}: 
        \State $\br^{(L)} =\tilde{\br}^{(L)} = \left(q_1^{(L)},p_1^{(L)},\dots, q_m^{(L)},p_m^{(L)}\right) \gets \br $
        \State $B_L \gets 1$
        \State $v \gets \abs{q_1^{(L)}}$
        \State $j \gets L$
\While{$v \leq \tau$ and $j>0$}    \Comment{Check adaptive condition; if not satisfied, switch from approximation to unbiased estimator}
\State $j \gets j-1$
        \State $\bar\br^{(j+1)} \gets \left(0,\eta\tilde p_1^{(j+1)},\dots, \tilde q_m^{(j+1)},\tilde p_m^{(j+1)}\right)$
        \State $\tilde{\br}^{(j)} \gets S_j^{-1}\bar{\bm r}^{(j+1)}$
        \State $B_j \gets  B_{j+1}\times \exp\left(-(1/2 + \bar n)(1-\eta^2)\tilde p_{j+1}^2 + i \bm{d}_j^T \Omega \bar{\bm r}_{j+1}\right)$
        \State $v \gets \abs{q_1^{(j)}}$
\EndWhile
\If{j > 0}
\State Produce an estimate $X$ of $\Tr[\hat D(\br_j)\calU_{\lfloor 1, j \rfloor} (\rho_0)]$ using Algorithm\ \ref{alg:sim-unbiased} with $N$ estimation rounds.
\Else
\State Set X = 1.
\EndIf 
\State Return  $B_j\times X$
 
\end{algorithmic}
\caption{Adaptive algorithm for expectation values of displacement operators \label{alg:sim-adaptive-general}}
\end{algorithm}

\begin{figure}
    \centering
    \includegraphics[width=0.8\linewidth]{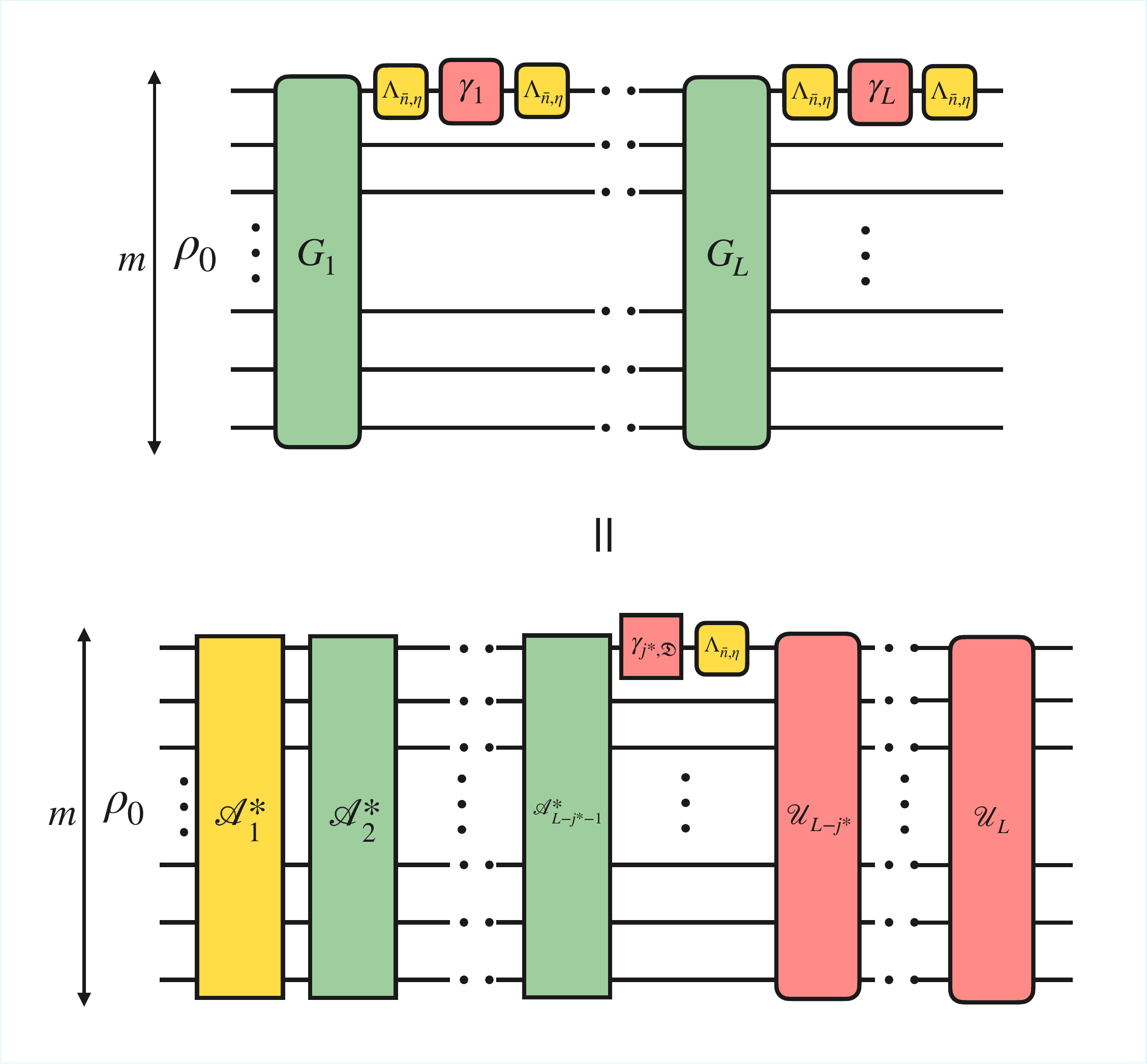}
    \caption{Circuit decomposition of the noisy bosonic circuit $\mathcal U$ (Eq.~\ref{appeq:noisy_bosonic}) behind the simulation algorithm that leads to the time complexity of estimation of characteristic function given by Lemma \ref{lem:adaptive_estimation_charateristic}. The expressions for $\mathcal A_1$ and $\mathcal A_j, \forall j \geq 2$ are given by Eq.~\ref{appeq:A_maps}. Given a phase-space point $\br \in \mathbb{R}^{2m}$ at which we want to estimate the characteristic function, we evolve the displacement operators for the first $j^*$ layers deterministically using the approximation scheme (Lemma \ref{applem:cubic_approx_small_q}) until the sampled phase-space point has large enough $q_1$, after which we use the unbiased oracle for nonzero symplectic coherence (Lemma \ref{lem:unbiased-symplectic}) to sample the phase-space points.}
    \label{fig:adaptive_estimation_characteristic}
\end{figure}

As outlined after Lemma \ref{applem:near_G_characteristic}, the algorithm for estimating $\Tr[\mathcal U(\rho_0)\hat D(\br)]$ for $\br \in \mathbb{R}^{2m}$ also allows us to estimate expectation values of local projectors. This is formalized in the following Lemma:

\begin{lemma}[Estimation of expectation value of local projectors]\label{applem:adaptive_proj_estimation}
Let $\rho_0$ be an initial state and $\calU \coloneqq \calU_L \circ \, \calU_{L-1} \circ \, \dots \circ \, \calU_1$ be a noisy bosonic circuit (Eq.~\ref{appeq:noisy_bosonic}). For all $j\in [L-1]$, assume that the partially evolved state $\rho_j \coloneqq \calU_j \circ\, \calU_{j-1} \circ\, \dots \circ\, \calU_1(\rho_0)$ satisfies the following derivative bound:
\begin{align}
\max_{\substack{ \br' \in \mathbb{R}^{2m} }} \bigg |\left\{\frac{\partial}{\partial q_1}\Tr[\Lambda_{\bar n,\eta^*}(\hat D(\br))\mathcal C_{j+1} \circ \Lambda_{\bar n , \eta} \circ \mathcal G_{j+1}(\rho_j)]\right\}_{\br =\br'} \bigg |\leq M.
\end{align}
For all $\br \in \mathbb{R}^{2m}$, there is a classical randomized algorithm running in time
\begin{align}
    8^k m \mathfrak c_2^{\ell} ML^2\epsilon^{-3}\log(1/\delta),
\end{align}
where
\begin{equation}
   \mathfrak c_2  = \frac{\Gamma(\tfrac14)}
     {\sigma_{\min}\eta^{3/4}\sqrt{24\pi|\gamma|_{\min}}\left\{\,(1/2+\bar n)(1-\eta)\right\}^{1/4}} ,
\end{equation}
and $0\leq \ell \leq 2L$, that approximates $\Tr[\hat O\,\calU(\rho_0)]$ for $\hat O = \left(\bigotimes_{i=1}^k \ket{\alpha_i}\bra{\alpha_i}\right)\otimes  \mathbb{I}^{\otimes(m-k)} $ with additive error $\epsilon$ and success probability $1-\delta$.
\end{lemma}
\begin{proof}
    From Lemma \ref{lem:adaptive_estimation_charateristic}, given $\br \in \R^{2m}$, we approximate the last $j^*$ cubic phase gates as identity, where $0 \leq j^* \leq L$ and the exact value of $j^*$ is determined by a threshold on the position coordinate of the first mode. Then, for the first $L-j^*$ layers, we can sample $\tilde{\br} \in \mathbb{R}^{2m}$ using the unbiased estimator (Algorithm \ref{alg:sim-unbiased}) with the appropriate probability distribution using Monte Carlo Markov chain sampling, such that
    \begin{equation}
        \left |\Tr[\mathcal U(\rho_0) \hat D(\br)] - Z\times A \times \mathbb{E}_{\tilde{\br}}\left[\Tr[\rho_0 \hat D(\tilde{\br})]\right] \right| \leq \tau M L,
    \end{equation}
    where 
    \begin{eqnarray} \label{appeq:exp_Z}
        Z &\leq&\mathfrak c_2^{\ell/2}/\sqrt{\tau} \nonumber \\ 
        A &=& \left(\prod_{i={j^*}}^{L-1}{\exp\left(-(1/2 + \bar n)(1-\eta^2)\tilde p_{i+1}^2 + i \bm{d}_i^T \Omega \bar{\bm r}_{i+1}\right)}\right) \leq 1,
    \end{eqnarray}
    $0\leq l \leq j^*$, $\bar{\br}$ is given by Eq.~\ref{appeq:adaptive_position_co}, and $\tau$ is a freely choosen parameter. We note that if the position coordinate of the first mode is small enough throughout, we simply approximate all the cubic phase gates as identity in which case $j^* = L, Z = 1 \leq 1/\sqrt{\tau}$ and the probability distribution over $\tilde r$ is simply a dirac delta function peaking at $\tilde{\br} = \tilde{\br}$, where $\tilde {\br}_1$ is given by Eq.~\ref{appeq:adaptive_position_co}.
    
    \noindent Therefore, given 
    \begin{equation}
        \Tr[\mathcal U(\rho_0)\hat O] = \fnormone{\hat O} \mathbb{E}_{p_{\hat O}(\br)}\left[\Tr[\mathcal U(\rho_0) \hat D(\br) \mathrm{arg}\left(\Tr[\hat O \hat D(-\br)]\right)] \right].
    \end{equation}
We have
\begin{eqnarray}
    &&\left|\fnormone{\hat O} \mathbb{E}_{p_{\hat O}(\br)}\left[\Tr[\mathcal U(\rho_0) \hat D(\br)] \mathrm{arg}(\Tr[\hat O \hat D(-\br)]) \right] - \fnormone{\hat O} \mathbb{E}_{p_{\hat O}(\br)}\left[\mathrm{arg}(\Tr[\hat O \hat D(-\br)]) \times Z A\mathbb{E}_{\tilde{\br}}\left[\Tr[\rho_0 \hat D(\tilde{\br})]\right]\right]\right| \nonumber \\
    &&\hspace{5mm} \leq \fnormone{\hat O} \mathbb{E}_{p_{\hat O}(\br)}\left[\left|\Tr[\mathcal U(\rho_0) \hat D(\br)] -  ZA\mathbb{E}_{\tilde{\br}}\left[\Tr[\rho_0 \hat D(\tilde{\br})]\right]\right|\right] = \tau ML \fnormone{\hat O},
\end{eqnarray}
where we have used the triangle inequality in the second line, Therefore, for the bias error to be $\epsilon/2$ we choose
\begin{equation}
    \tau = \frac{\epsilon}{2\fnormone{O}ML}.
\end{equation}
Finally, through the Chernoff--Hoeffding bound (Lemma \ref{applem:CH}), the statistical average of $N$ samples of $\fnormone{\hat O} ZA \times\mathrm{arg}\left(\Tr[\hat O \hat D(-\br)]\right)\Tr[\rho \hat D(\tilde {\br})]$ over appropriately chosen phase-space points $\br$ and $\tilde{\br}$ will be $\epsilon$-close to the true expectation value with probability $1-\delta$ when the number of samples satisfies
\begin{equation}
    N = \frac{\fnormone{\hat O}^2 |Z|^2 |A|^2}{\epsilon^2}\log(1/\delta) \leq 8^k\mathfrak c_2^{\ell} ML\epsilon^{-3}\log(1/\delta).
\end{equation}
The last inequality comes from the expression of $Z$ and A (Eq.~\ref{appeq:exp_Z}), and the expression for Fourier 1-norm of $\hat{O} = \left(\bigotimes_{i=1}^k \ket{\alpha_i}\bra{\alpha_i}\right)\otimes  \mathbb{I}^{\otimes (m-k)}$ (given in Eq.~(\ref{appeq:Fourier_1_norm_local_projectors})). Combining the sample complexity with the time required to sample once gives the required time complexity.
\end{proof}
\noindent Note that from Lemma \ref{lem:curvature_bound}, bounding the first four moments of the position quadrature throughout the circuit gives a bound on $M$ throughout the circuit. Further, reflecting on Lemma \ref{applem:adaptive_proj_estimation}, this gives us two regimes of classical simulation of expectation value of local projectors: (i) For arbitrary $\mathfrak c_2 \in \mathcal{O}(1)$, $k = \mathcal O(\log(m)
)$ and $L = \mathcal O(\log(m))$ allows for efficient classical simulation of output local projectors, whereas for (ii) $\mathfrak c_2 < 1$, $k = \mathcal O(\log(m)
)$ and $L = \mathcal O(\poly(m))$ allows for efficient classical simulation of output local projectors. This gives the first two regimes of efficient classical simulation detailed in Theorem \ref{theo:overlap_gate_noise} of the main text. Note that since $\mathfrak c_2 \leq \mathfrak{c}$, with $\mathfrak c$ given by Eq.~\ref{appeq:first_cc} (This is because $1/\eta^{3/4} \leq 1/\eta$ for $\eta \leq 1$), therefore $\mathfrak c < 1$ implies $\mathfrak c_2 < 1$ and in point (i) and (ii) of Theorem \ref{theo:overlap_gate_noise} in the main text, we write the conditions in terms of $\mathfrak c$ to enhance readability of the main text.

Further, combining Lemmas \ref{applem:quad_finite_diff} and \ref{lem:adaptive_estimation_charateristic} also allows for efficient estimation of expectation value of first two quadrature moments, under additional guarantees on the magnitude of the quadrature moments of the output state. This is summarized in the following Lemma:

\begin{lemma}[Quadrature moment estimation]\label{applem:adaptive_quad_estimation}
Let $\rho_0$ be an initial state and $\calU \coloneqq \calU_L \circ \, \calU_{L-1} \circ \, \dots \circ \, \calU_1$ be a noisy bosonic circuit. For all $j\in [L-1]$, assume that the partially evolved state $\rho_j \coloneqq \calU_j \circ\, \calU_{j-1} \circ\, \dots \circ\, \calU_1(\rho_0)$ satisfies the derivative bound:
\begin{align}\label{eq:curvature_bound}
\max_{\substack{ \br' \in \mathbb{R}^{2m} }} \bigg |\left\{\frac{\partial}{\partial q_1}\Tr[\Lambda_{\bar n,\eta^*}(\hat D(\br))\mathcal C_{j+1} \circ \Lambda_{\bar n , \eta} \circ \mathcal G_{j+1}(\rho_j)]\right\}_{\br =\br'} \bigg |\leq M.
\end{align}
and up to sixth moments of position and momentum quadratures are upper bounded by $E$ on the output state $\mathcal U(\rho_0)$, then $\Tr[\hat q_j \mathcal U(\rho_0)], \Tr[\hat q_j^2 \mathcal U(\rho_0)], \Tr[\hat p_j \mathcal U(\rho_0)], \Tr[\hat p_j^2 \mathcal U(\rho_0)]$ can be estimated with precision $\epsilon$ and failure probability $1- \Delta$ by a classical algorithm running in time
\begin{equation}
   \mathcal{O}(mL\mathfrak c_2^\ell  ML(4E+3)^{12}\epsilon^{-12}\log(1/\Delta))
\end{equation}
where
\begin{equation}
    \mathfrak c_2 = \frac{\Gamma(\tfrac14)}
     {\sigma_{\min}\eta^{3/4}\sqrt{24\pi|\gamma|_{\min}}\left\{\,(1/2+\bar n)(1-\eta)\right\}^{1/4}} ,
\end{equation}
and $0 \leq l \leq 2L$.
\end{lemma}
\begin{proof}
    We explain the estimation of $\Tr[\hat q_1 \mathcal U(\rho)]$ and $\Tr[\hat q_1^2 \mathcal U(\rho)]$, the estimation of first mode momentum quadrature as well as quadrature moments of other modes follows similar time complexity. From Lemma \ref{applem:quad_finite_diff}, 
    \begin{equation}
        \frac{\Tr[\mathcal U(\rho_0) \hat D_1(0,\delta)]-1}{i\delta}
    \end{equation}
    approximates $\Tr[\mathcal{U}(\rho_0) \hat q_1]$ up to precision $\delta\Tr[\mathcal U (\rho_0) \hat q_1^4] \leq \delta  \Tr[\mathcal U (\rho_0) \hat q_1^4] \leq \delta E$. Further, from Lemma \ref{lem:adaptive_estimation_charateristic}, we can approximate $\Tr[\mathcal U(\rho_0) \hat q_1]$ up to precision $\varepsilon$ with probability $1-\Delta$ in time
    \begin{equation}
       mL \mathfrak c_2^\ell ML\varepsilon^{-3}\log(1/\Delta).
    \end{equation}
    Here we pick $\varepsilon = \delta^2$ for the precision on the estimate of $\Tr[\rho \hat q_1]$ to be of the order $\delta$. Further we require $\delta = \epsilon/(E+1)$ for the precision in estimation of the quadrature moment to be $\epsilon$. This gives the required time complexity for estimation of $\Tr[\mathcal U(\rho_0) \hat q_1]$. Similarly, from Lemma \ref{applem:quad_finite_diff}, $\Tr[\mathcal U(\rho_0) \hat q_1^2]$ can be estimated by
    \begin{equation}
        \frac{2}{\delta^2}\left((1-\Tr[\mathcal U(\rho_0) \hat D_1(0,\delta)]) + \frac{1}{\delta}(\Tr[\mathcal U(\rho_0) \hat D_1(0,\delta^2)] - 1)\right).
    \end{equation}
    up to precision $2/3 \times \delta \Tr[\rho\hat q_j^6] + 2\delta \Tr[\rho \hat q_j^4] \leq 8/3\times E$. Similar to the previous discussion, we want to estimate $\Tr[\mathcal{U}(\rho_0)\hat D_1(0,\delta)]$ up to precision $\delta^3$ and $\Tr[\mathcal{U}(\rho_0)\hat D_1(0,\delta^2)]$ up to precision $\delta^4$ and we choose $\delta = \frac32\times\frac{\epsilon}{4E+3}$ for the precision of estimation of quadrature moment to be $\epsilon$. This gives the required time complexity and we keep the most significant time complexity in the statement of the Lemma.
\end{proof}
Therefore, Lemma \ref{applem:adaptive_quad_estimation} yields two regimes of efficient classical simulation of the first two quadrature moments, provided that first four quadrature moments are bounded throughout the circuit, and up to sixth moment of quadrature are bounded in the output state: (i) For arbitrary, $\mathfrak c_2 \in \mathcal O(1)$ and $L = \mathcal O(\log(m))$ allows for efficient classical simulation of first two quadrature moments, whereas for (ii) $\mathfrak c_2 < 1$ and $L = \mathcal O(\poly(m))$ allows for efficient classical simulation of first two quadrature moments. This gives the first two regimes of efficient classical simulation detailed in Theorem \ref{theo:quad_gate_noise} of the main text. Note that since $\mathfrak c_2 \leq \mathfrak{c}$, with $\mathfrak c$ given by Eq.~\ref{appeq:first_cc} (This is because $1/\eta^{3/4} \leq 1/\eta$ for $\eta \leq 1$), therefore $\mathfrak c < 1$ implies $\mathfrak c_2 < 1$ and in point (i) and (ii) of Theorem \ref{theo:quad_gate_noise} in the main text, we write the conditions in terms of $\mathfrak c$ to enhance readability of the main text.

In the final section, we detail classical simulation algorithms for estimating expectation values of local coherent state projectors and quadrature moments for noisy bosonic circuits for which the cubicity of all the cubic phase gates is high as a function of other circuit parameters, i.e.\ ``high-cubicity'' circuits.

\subsubsection{Adaptive algorithms for high-cubicity circuits}\label{appsec:high_cubicity}

In this section, we assume nothing about the circuit  whose cubicity is above a certain threshold, depending on $M$ (Eq.~(\ref{eq:curvature_bound})) and $L$. Informally, this algorithm works as follows: Given that we want to estimate $\Tr[\hat D(\br)\mathcal U(\rho_0)]$ for some $\br \in \R^{2m}$, input state $\rho_0$ and the noisy bosonic circuit $\mathcal U$ given by Eq.~\ref{appeq:noisy_bosonic}, we first write
\begin{equation}
    \Tr[\hat D(\br)\mathcal U(\rho_0)] = \Tr[\mathcal U_L^{*}(\hat D(\br)) \bigcirc_{j=1}^{L-1} \mathcal U_{L - j}^* (\rho_0)]
\end{equation}
Now, depending on the position coordinate of the first mode, $q_1$ for $\br$, we do one of the following:
\begin{itemize}
    \item For $|q_1| \leq \tau$ for some threshold $\tau$, we approximate the cubic phase gate as identity, deterministically evolving the displacement operator under the noisy channels and the Gaussian gate to get the new displacement operator $\hat D(\br_{L-1})$ (we also get a prefactor $\leq$ 1, for the ease of understanding we ignore it here and take care of it in the formal Lemma) and introducing a bias error, or
    \item For $|q_1| > \tau$, we call the physical sampling oracle $\mathcal U_L^*(\hat D(\br))$ given by Lemma \ref{applem:physical_oracle} with magnitude given by Eq.~\ref{appeq:mag_physical_oracle} and zero bias, to sample some $\br_{L-1} \in \R^{2m}$.
\end{itemize}
After obtaining $\br_L$, to approximate
\begin{equation}
    \Tr[\hat D(\br_{L_1}) \bigcirc_{j=1}^{L-1} \mathcal U_{L - j}^* (\rho_0)] = \Tr[\mathcal U_{L-1}^{*}(\hat D(\br_{L-1})) \bigcirc_{j=2}^{L-1} \mathcal U_{L - j}^* (\rho_0)]
\end{equation}
we repeat the same process and so on. Note that to bound the error when chaining this algorithm for $L$ layers, we require that
\begin{equation} \label{appeq:mag_cond}
\frac{1}{\sqrt{24 |\gamma|\sqrt{\eta}\tau}} \exp\left(-(1/2 + \bar n)(1-\eta)\tau^2 \right) (1/2 + \bar n)^{-1/2}(1-\eta)^{-1/2} < 1,
\end{equation}
as the bias error introduced in the $L - j$'th layer of the evolution of the displacement operator is multiplied by the magnitude of the sampling oracle in the $L - j +1$'th layer by Lemma \ref{applem:CHained-sampling}, and therefore we need to keep the magnitude of the sampling oracle less than equal to one throughout. Note that Eq.`\ref{appeq:mag_cond} gives the following the equivalent to the following condition on $|\gamma|$:
\begin{equation}
    |\gamma| > \frac{1}{24 \pi \sqrt{\eta}(1/2 + \bar n)(1-\eta)} \exp(-(1/2 + \bar n)(1-\eta)2\tau^2).
\end{equation}
The algorithm is formalized in the following Lemma:

\begin{lemma}[Expectation values of displacement operators in high-cubicity circuits]\label{applem:high_gaussian_characterisitic_est}
 Let $\rho_0$ be an initial state and $\calU \coloneqq \calU_L \circ \, \calU_{L-1} \circ \, \dots \circ \, \calU_1$ be a noisy bosonic circuit. For all $j\in [L-1]$, assume that the partially evolved state $\rho_j \coloneqq \calU_j \circ\, \calU_{j-1} \circ\, \dots \circ\, \calU_1(\rho_0)$ satisfies the following derivative bound:
\begin{align}
\max_{\substack{ \br' \in \mathbb{R}^{2m} }} \bigg |\left\{\frac{\partial}{\partial q_1}\Tr[\Lambda_{\bar n,\eta^*}(\hat D(\br))\mathcal C_{j+1} \circ \Lambda_{\bar n , \eta} \circ \mathcal G_{j+1}(\rho_j)]\right\}_{\br =\br'} \bigg |\leq M.
\end{align}
then for all $\br \in \mathbb{R}^{2m}$, and the cubicity of the cubic phase gates is such that $\mathfrak{d}_\epsilon < 1$, where
\begin{equation}\label{appeq:cond_d_eps}
    \mathfrak d_{\epsilon}  = \frac{\sqrt{ML}}{\sqrt{12 |\gamma|\sqrt{\eta}\epsilon}} \exp\left(-(1/2 + \bar n)(1-\eta)\frac{\epsilon^2}{4 M^2 L^2} \right) (1/2 + \bar n)^{-1/2}(1-\eta)^{-1/2},
\end{equation}
or equivalenty,
\begin{equation}
   \gamma_{\min} =   \min_{j\in[L]}\abs{\gamma_j} > \frac{ML}{12\epsilon\eta^{1/2}(1/2+\bar n)(1-\eta)} \exp\left(-(1/2 + \bar n)(1-\eta^2)\frac{\epsilon^2}{2M^2L^2}\right),
\end{equation}
there is a classical randomized algorithm running in time
\begin{align}
    \mathcal O(mL\epsilon^{-2}\log(1/\delta)),
\end{align}
that approximates $\Tr[\hat D(\br)\calU(\rho_0)]$ with additive error $\epsilon$
and success probability $1-\delta$.   
\end{lemma}
\begin{proof}
We will show that the procedure formalized in Algorithm\ \ref{alg:sim-adaptive-high-gaussian} yields the desired approximation. We write 
\begin{equation}
     \Tr[\hat D(\br)\mathcal U(\rho_0)] = \Tr[\mathcal U_L^{*}(\hat D(\br)) \bigcirc_{j=1}^{L-1} \mathcal U_{L - j}^* (\rho_0)]
\end{equation}
with $\br = (q_1,p_1,q_2,p_2,\dots,q_m,p_m)$. Then we do one of the following:
\begin{itemize}
    \item For $|q_1| \leq \tau$, we put $q_1 = 0$ and consequently, write the cubic phase gate as identity, such that, from Lemma \ref{applem:cubic_approx_small_q},
    \begin{equation}
        \left|\Tr[\mathcal U_L^{*}(\hat D(\br)) \bigcirc_{j=1}^{L-1} \mathcal U_{L - j}^* (\rho_0)] - \Tr[\mathcal G_L^*\circ\Lambda_{\bar n,\eta}^*\circ\Lambda_{\bar n,\eta}^*\hat D(\br_0) \bigcirc_{j=1}^{L-1} \mathcal U_{L - j}^* (\rho_0) ] \right| \leq |q_1| M \leq \tau M,
    \end{equation}
    where $\br_0 = (0,p_1,q_2,p_2,\dots,q_m,p_m)$. Now,
    \begin{equation}
        \Tr[\mathcal G_L^*\circ\Lambda_{\bar n,\eta}^*\circ\Lambda_{\bar n,\eta}^*\hat D(\br_0) \bigcirc_{j=1}^{L-1} \mathcal U_{L - j}^* (\rho_0) ] = A_1 \Tr[\hat D(S_L^{-1}\tilde{\br_0}) \bigcirc_{j=1}^{L-1} \mathcal U_{L - j}^* (\rho_0)],
    \end{equation}
    where
    \begin{eqnarray}
        A_1 &=& \exp(-(1/2 + \bar n)(1-\eta^2)p_1^2 + i \bm{d_L}^T \Omega \tilde{\br}_0) \text{ with } |A_1| \leq 1, \nonumber \\
        \tilde{\br}_0 &=& (0,\eta p_1,q_2,p_2,\dots, q_m,p_m).
    \end{eqnarray}
    This gives the description of the evolved displacement operator.
    \item If $|q_1| > \tau$, we call the physical sampling oracle for $\mathcal U_L^*$ given by Lemma \ref{applem:physical_oracle} such that
    \begin{equation}
        \Tr[\mathcal U_L^{*}(\hat D(\br)) \bigcirc_{j=1}^{L-1} \mathcal U_{L - j}^* (\rho_0)] = \mathbb{E}_{\tilde \br_{L-1}}\left[A_2 \Tr[\hat D(\tilde \br_{L-1})\bigcirc_{j=1}^{L-1} \mathcal U_{L - j}^* (\rho_0)]\right].
    \end{equation}
    At this point, we assume that the magnitude of the sampling oracle
    \begin{equation}
    |A_2| \leq \frac{\exp(-(1/2 + \bar n)(1-\eta)q_1^2)}{\sqrt{24 \pi\sqrt{\eta}|\gamma| |q_1|(1/2 +\bar n)(1-\eta)}} \leq \frac{\exp(-(1/2 + \bar n)(1-\eta)\tau^2)}{\sqrt{24 \pi\sqrt{\eta}|\gamma| \tau(1/2 +\bar n)(1-\eta)}} < 1.
    \end{equation}
    We note that this will be assumed for the rest of the $L-1$ layers. We explain later how this connects to the condition on cubicity we have specified.
\end{itemize}
Therefore, combining these two points, we can write an adaptive sampling oracle using either the approximation scheme (Lemma \ref{applem:cubic_approx_small_q}) or the sampling oracle, such that
\begin{equation}
    \left|\Tr[\mathcal U_L^{*}(\hat D(\br)) \bigcirc_{j=1}^{L-1} \mathcal U_{L - j}^* (\rho_0)] - \mathbb{E}_{\br_{L-1}}\left[A \Tr[\hat D(\br_{L-1}) \bigcirc_{j=1}^{L-1} \mathcal U_{L - j}^* (\rho_0)] \right]\right| \leq \tau M,
\end{equation}
such that the magnitude of the sampling oracle $|A| \leq \max(|A_1|,|A_2|) \leq 1 $. We note that if $q_1$ for $\br$, we are deterministically evolving the displacement operator and the probability distribution for the sampling oracle is simply a dirac delta function centered on the evolved displacement operator.

Combing $L$ of such adaptive sampling oracles, from Lemma \ref{applem:CHained-sampling}, we obtain a sampling oracle $\mathbb{E}_{\tilde {\br}}\left[Z \Tr[\hat D(\tilde{\br}) \rho_0]\right]$ with $|Z| \leq 1$ such that
\begin{equation}
   \left| \Tr[\hat D(\br) \mathcal U(\rho_0)] - \mathbb{E}_{\tilde {\br}}\left[Z \Tr[\hat D(\tilde{\br}) \rho_0]\right] \right| \leq \tau ML.
\end{equation}
We choose $\tau = \epsilon/(2ML)$ which gives the condition on the magnitude
\begin{equation}
    \mathfrak d_{\epsilon}  = \frac{\sqrt{ML}}{\sqrt{12 |\gamma|\sqrt{\eta}\epsilon}} \exp\left(-(1/2 + \bar n)(1-\eta)\frac{\epsilon^2}{4 M^2 L^2} \right) (1/2 + \bar n)^{-1/2}(1-\eta)^{-1/2} < 1
\end{equation}
and the equivalent condition on $\gamma$ given in the Lemma.
Therefore, from the Chernoff-Hoeffding bound (Lemma \ref{applem:CH}), through $N$ queries of this adaptive sampling oracle, we can estimate $\Tr[\hat D(\br) \mathcal U(\rho_0)]$ with precision $\epsilon$ and success probability $1-\delta$, as long as the number of samples
\begin{equation}
    N \in \mathcal O(\epsilon^{-2}\log(1/\delta)).
\end{equation}
Multiplied by the time complexity $\mathcal{O}(mL)$ to call the sampling oracle once gives us the required time complexity.
\end{proof}

\begin{algorithm}[H]
\textbf{Parameters}: number of estimation rounds $N$, bias $\epsilon'$ \\ 
\textbf{Input}: Classical descriptions of noisy circuit $\mathcal{U} = \calU_L \circ \calU_{L-1}\circ \dots \circ \calU_1$, classical description of an initial state $\rho_0$, and a vector $\br \in \mathbb{R}^{2m}$ (possibly with $q_{1} = 0$)\\
\textbf{Output}: An estimate of $\Tr[\hat D(\br)\mathcal{U}(\rho_0)]$
\begin{algorithmic}[1]
    \State \textbf{Initialize}: 
        \State $\br^{(L)} = \left(q_1^{(L)},p_1^{(L)},\dots, q_m^{(L)},p_m^{(L)} \right) \gets \br $
        \State $B_L \gets 1$
\For{ $j = 1 \text{ to } N $, }
    \For{ $k = L \text{ to } 1$, }
    \If{$|q_1|^{(k)} \leq \tau$}
    \State $(\br^{(k-1)},Z_{k-1}) \gets (\exp(-(1/2+\bar n)(1-\eta^2)p_1^2 + i \bm{d}_k^T \Omega \br^{(k-1)}_{0}))$ with bias $|\tau| M$ \Comment{ cf. Lemma \ref{applem:cubic_approx_small_q}}
    \State \State $B_{k-1} \gets  B_{k}\,{Z_{k-1}}$
    \Else
    \State $(\br^{({k-1})}, Z_{k-1}) \gets \mathsf{Sample}^{(\calS_M)}(\mathcal G_k^*\circ\Lambda_{\bar n, \eta}^* \circ\calC_k^*\circ \Lambda_{\bar n, \eta}^*(\hat D(\br^{(k)})))$ with bias at most $\epsilon'$ \Comment{ cf. Lemma\ \ref{applem:physical_oracle}}
        \State $B_{k-1} \gets  B_{k}\,{Z_{k-1}}$
    \EndIf
        
    \EndFor

    \State  $X_j \gets  B_0\Tr[\hat D(\br_0) {\rho_0}]$
 \EndFor
 \State Return output $X \gets \frac{1}{N}\sum_{j=1}^N X_j$
 
\end{algorithmic}
\caption{Adaptive algorithm for expectation values of displacement operators in the high-cubicity regime \label{alg:sim-adaptive-high-gaussian}}
\end{algorithm}

\noindent Similar to the previous section, the estimation of characteristic function in this regime also allows for expectation value estimation of local projective measurements and quadrature moment estimation. This is formalized in the following two Lemmas:

\begin{lemma}[Expectation value of overlap of local projectors in high-cubicity circuits.]\label{applem:high-Gaussian_proj_estimation}
     Let $\rho_0$ be an initial state and $\calU \coloneqq \calU_L \circ \, \calU_{L-1} \circ \, \dots \circ \, \calU_1$ be a noisy bosonic circuit. For all $j\in [L-1]$, we assume that the partially evolved state $\rho_j \coloneqq \calU_j \circ\, \calU_{j-1} \circ\, \dots \circ\, \calU_1(\rho_0)$ satisfies the following derivative bound:
\begin{align}
\max_{\substack{ \br' \in \mathbb{R}^{2m} }} \bigg |\left\{\frac{\partial}{\partial q_1}\Tr[\Lambda_{\bar n,\eta^*}(\hat D(\br))\mathcal C_{j+1} \circ \Lambda_{\bar n , \eta} \circ \mathcal G_{j+1}(\rho_j)]\right\}_{\br =\br'} \bigg |\leq M.
\end{align}
and that the cubicity of the cubic phase gates is such that $\mathfrak{d}_{\epsilon/{2^k}} < 1$ (with $\mathfrak{d_\epsilon}$ given in Eq.~\ref{appeq:cond_d_eps}), or equivalently,
\begin{eqnarray}\label{appeq:high_gamma_condition_projectors}
    \gamma_{\min} = \min_{j \in {L}} |\gamma_j| > \frac{ML}{12\eta^{1/2}(1/2 + \bar n)(1-\eta)} \times \frac{2^k}{\epsilon}\exp\left(-(1/2 + \bar n)(1-\eta^2)\frac{\epsilon^2}{2M^2L^2 4^k}\right) .
\end{eqnarray}
Then, for all $\br \in \mathbb{R}^{2m}$, there is a classical randomized algorithm running in time
\begin{align}
   \mathcal O(4^k mL\epsilon^{-2}\log(1/\delta)),
\end{align}
that approximates $\Tr[\hat O\,\calU(\rho_0)]$ for $\hat O = \left(\bigotimes_{i=1}^k \ket{\alpha_i}\bra{\alpha_i}\right)\otimes  \mathbb{I}^{\otimes (m-k)} $ with additive error $\epsilon$ and success probability $1-\delta$.
\end{lemma}
\begin{proof} As noted before (Eqs. \ref{eq:optical_equivalence_1} and \ref{eq:optical_equivalence_2}),
    \begin{equation}
        \Tr[\mathcal U(\rho_0)\hat O] = \fnormone{\hat O} \mathbb{E}_{p_{\hat O}}[\mathrm{arg}(\chi_{\hat{O}}(\br))\Tr[\mathcal{U}(\rho_0)\hat D(-\br)]].
    \end{equation}
    From Algorithm \ref{alg:sim-adaptive-high-gaussian}, we can approximate $\Tr[\mathcal{U}(\rho_0)\hat D(-\br)]$ by the sampling oracle $Z \mathbb{E}_{p_{\br}}[\Tr[\rho_0\hat D(\tilde{\br})]]$ such that
    \begin{eqnarray}
        \left|\fnormone{\hat O} \mathbb{E}_{p_{\hat O}}[\mathrm{arg}(\chi_{\hat{O}}(\br))\Tr[\mathcal{U}(\rho_0)\hat D(-\br)]] - \fnormone{\hat O}\mathbb{E}_{p_{\hat O}}[\mathrm{arg}(\chi_{\hat{O}}(\br))A \mathbb{E}_{p_{\br}}[\Tr[\rho_0 \hat D(\tilde{\br})]]]\right| \leq \fnormone{\hat O} L \tau M,
    \end{eqnarray}
    where $\tau$ is a tunable parameter, provided that the magnitude of sampling oracle at each layer is smaller than $1$. We set $\tau = \epsilon/(2\fnormone{\hat O}L)$. This gives the condition for the cubicity---or equivalenty for the magnitude of the physical sampling oracle from Lemma \ref{applem:physical_oracle} to be less than 1---to be
    \begin{equation}
        \gamma_{\min} =  \min_{j\in[L]}\abs{\gamma_j} > \frac{ML\fnormone{\hat O}}{12\epsilon \eta^{1/2}(1/2+\bar n)(1-\eta)} \exp\left(-(1/2 + \bar n)(1-\eta^2)\frac{\varepsilon^2}{2M^2L^2\fnormone{O}^2}\right).
    \end{equation}
    Then, by sampling from $\fnormone{\hat O}\mathbb{E}_{p_{\hat O}}[\mathrm{arg}(\chi_{\hat{O}}(\br))A \mathbb{E}_{p_{\br}}[\Tr[\rho_0\hat D(\tilde{\br})]]]$ and using the Chernoff-Hoeffding bound (Lemma \ref{applem:CH}) with precision $\epsilon/2$ and combining with the time complexity $\mathcal O(mL)$ of obtaining one sample, we get the required time complexity of estimating $\Tr[\mathcal U(\rho_0)\hat O]$ with precision $\epsilon$ and success probability $1-\delta$. Using the fact that $\fnormone{\hat O} = 2^k (Eq.~\ref{appeq:Fourier_1_norm_local_projectors})$ when the output operators are $k$-local coherent state projectors, gives the statement of the Lemma.
\end{proof}

\noindent Therefore, efficient estimation of overlap of $k$-local coherent state projectors with the output state upto an arbitrary precision is possible, as long as $k = \mathcal O(\log(m))$, $L = \mathcal O(\poly(m))$ and $\mathfrak d_{\epsilon/2^k} < 1$, where $\mathfrak{d}_\epsilon$ is given by Eq.~\ref{appeq:cond_d_eps}. This gives point (iii) of Theorem \ref{theo:overlap_gate_noise} in the main text. For quadrature moment estimation, we have the following Lemma:

\begin{lemma}[Expectation value of quadrature moments in high-cubicity circuits.]\label{applem:high-cubicity_quad_estimation}
     Let $\rho_0$ be an initial state and $\calU \coloneqq \calU_L \circ \, \calU_{L-1} \circ \, \dots \circ \, \calU_1$ be a noisy bosonic circuit. For all $j\in [L]$, we assume that the partially evolved state $\rho_j \coloneqq \calU_j \circ\, \calU_{j-1} \circ\, \dots \circ\, \calU_1(\rho_0)$ satisfies the following derivative bound:
\begin{align}
\max_{\substack{ \br' \in \mathbb{R}^{2m} }} \bigg |\left\{\frac{\partial}{\partial q_1}\Tr[\Lambda_{\bar n,\eta^*}(\hat D(\br))\mathcal C_{j+1} \circ \Lambda_{\bar n , \eta} \circ \mathcal G_{j+1}(\rho_j)]\right\}_{\br =\br'} \bigg |\leq M.
\end{align}
and that the cubicity of the cubic phase gates is such that $\mathfrak d_{\epsilon^4/(8E+12)^4} < 1$ (with $\mathfrak d_\epsilon$ given by Eq.~\ref{appeq:cond_d_eps}), or equivalently
\begin{eqnarray}\label{appeq:high_gamma_condition_quadrature}
    \gamma_{\min} = \min_{j \in {L}} |\gamma_j| > \frac{ML}{12\eta^{1/2}(1/2 + \bar n)(1-\eta)}\frac{(8E+12)^4}{\epsilon^4}\exp\left(-(1/2+ \bar n)(1-\eta^2)\frac{\epsilon^8}{2M^2L^2 (8E+12)^8}\right).
\end{eqnarray}
Then, if up to the sixth quadrature moments of the output state $\calU(\rho_0)$ are upper bounded by $E$, then there is a classical algorithm running in time
\begin{equation}\mathcal O(mL(2E+3)^8\epsilon^{-8}\log(1/\delta)),
\end{equation}
that estimates $\Tr[\hat q_j \mathcal U(\rho_0)], \Tr[\hat q_j^2 \mathcal U(\rho_0)], \Tr[\hat p_j \mathcal U(\rho_0)], \Tr[\hat p_j^2 \mathcal U(\rho_0)]$, $\forall j \in \{1,\dots,m\}$ with additive error $\epsilon$ with success probability $1-\delta$.
\end{lemma}
\begin{proof}
    We explain the estimation procedure for $\Tr[\mathcal U(\rho_0)\hat q_j]$ and $\Tr[\mathcal U(\rho_0)\hat q_j^2]$, and the estimation of expectation value of momentum quadratures follow a similar pattern. From Lemma \ref{applem:quad_finite_diff},
   \begin{equation}
        \left|\Tr[\mathcal U(\rho_0) \hat q_j] - \frac{\Tr[\mathcal U(\rho_0)\hat D_j(0,\delta)] - 1}{i\delta}\right| \leq \delta E.
    \end{equation}
    From Lemma \ref{applem:high_gaussian_characterisitic_est}, we compute an estimate $A_{\delta}$ that approximates $\Tr[\mathcal U(\rho_0)\hat D_j(0,\delta)]$ up to error $\varepsilon$ and success probability $1-\Delta$, in time
    \begin{equation}
        \mathcal{O}(mL \varepsilon^{-2} \log(1/\Delta)),
    \end{equation}
    provided that 
\begin{equation}
   \gamma_{\min} =  \min_{j\in[L]}\abs{\gamma_j} > \frac{ML}{12\varepsilon\eta^{1/2}(1/2+\bar n)(1-\eta)} \exp\left(-(1/2 + \bar n)(1-\eta^2)\frac{\varepsilon^2}{2M^2L^2}\right).
\end{equation}
Further, $A_\delta$ satisfies
\begin{equation}
     \left|\Tr[\mathcal U(\rho_0) \hat q_j] - \frac{A_\delta - 1}{i\delta}\right| \leq \delta E + \frac{\varepsilon}{\delta}.
\end{equation} 
Therefore, choosing $\varepsilon = \delta^2$ with $\delta = \epsilon/(E+1)$ allows us to find the $\epsilon$-approximate value of $\Tr[\mathcal U(\rho_0)\hat q_j]$ with success probability $1-\Delta$ in time 
\begin{equation}
    \mathcal O(mL(E+1)^4 \epsilon^{-4} \log(1/\Delta)),
\end{equation}
and the condition on cubicity to be
\begin{equation}
    \gamma_{\min} =  \min_{j\in[L]}\abs{\gamma_j} > \frac{ML (E+1)^2}{12\epsilon^2\eta^{1/2}(1/2+\bar n)(1-\eta)} \exp\left(-(1/2 + \bar n)(1-\eta^2)\frac{\epsilon^4}{2M^2L^2(E+1)^4}\right).
\end{equation}
For estimation of $\Tr[\mathcal U(\rho_0) \hat q_j^2]$, we again note from Lemma \ref{applem:quad_finite_diff} that
\begin{equation}
\left| \Tr[\mathcal U(\rho_0) \hat q_j^2] - \frac{2}{\delta^2}\left((1 - \Tr[\mathcal{U}(\rho_0) \hat D_j(0,\delta)]) + \frac{1}{\delta}(\Tr[\mathcal{U}(\rho_0) \hat D_j(0,\delta^2)] - 1)\right)\right|   \leq \frac{8}{3}E \delta.    
\end{equation}
As before, we compute an $\varepsilon_1$ approximate value of $\Tr[\mathcal{U}(\rho_0) \hat D_j(0,\delta)]$, $A_\delta$ with success probability $\Delta$ in time 
\begin{equation}
    \mathcal O(mL\varepsilon_1^{-2} \log(1/\Delta))
\end{equation}
as long as
\begin{equation}
    \gamma_{\min} =  \min_{j\in[L]}\abs{\gamma_j} > \frac{ML}{12\varepsilon_1\eta^{1/2}(1/2+\bar n)(1-\eta)} \exp\left(-(1/2 + \bar n)(1-\eta^2)\frac{\varepsilon_1^2}{2M^2L^2}\right).
\end{equation}
Similarly, we compute we compute an $\varepsilon_2$ approximate value of $\Tr[\mathcal{U}(\rho_0) \hat D_j(0,\delta^2)]$, $A_{\delta^2}$ with success probability $\Delta$ in time 
\begin{equation}
    \mathcal O(mL\varepsilon_2^{-2} \log(1/\delta)),
\end{equation}
as long as
\begin{equation}
    \gamma_{\min} =  \min_{j\in[L]}\abs{\gamma_j} > \frac{ML}{12\varepsilon_2\eta^{1/2}(1/2+\bar n)(1-\eta)} \exp\left(-(1/2 + \bar n)(1-\eta^2)\frac{\varepsilon_2^2}{2M^2L^2}\right).
\end{equation}
$A_\delta$ and $A_{\delta^2}$ further satisfy
\begin{equation}
    \left| \Tr[\mathcal U(\rho_0) \hat q_j^2] - \frac{2}{\delta^2}\left((1 - A_{\delta}) + \frac{1}{\delta}(A_{\delta^2}- 1)\right)\right|   \leq \frac{8}{3}E \delta + \frac{2\varepsilon_1}{\delta^2} + \frac{2\varepsilon_2}{\delta^3}.    
\end{equation}
Choosing $\varepsilon_1 = \delta^3$ and $\varepsilon_2 = \delta^4$, with $\delta = 3\epsilon/(8E+12)$, we can compute an $\epsilon$-approximate value of $\Tr[\mathcal U(\rho_0)\hat q_j^2]$ with success probability $(1-\Delta)^2$ in time (keeping the dominant time complexity)
\begin{equation}
    \mathcal{O}(mL(2E+3)^8\epsilon^{-8}\log(1/\delta)),
\end{equation}
provided that
\begin{eqnarray}
    \gamma_{\min} > \frac{ML}{12\eta^{1/2}(1/2 + \bar n)(1-\eta)}&& \nonumber \\
    && \hspace{-50mm}\times\max\left(\frac{(8E+12)^3}{\epsilon^3}\exp\left(-(1/2+ \bar n)(1-\eta^2)\frac{\epsilon^6}{2M^2L^2 (8E+12)^6}\right),\frac{(8E+12)^4}{\epsilon^4}\exp\left(-(1/2+ \bar n)(1-\eta^2)\frac{\epsilon^8}{2M^2L^2 (8E+12)^8}\right)\right). \nonumber \\
    &&\hspace{-50mm}=\frac{ML}{12\eta^{1/2}(1/2 + \bar n)(1-\eta)}\frac{(8E+12)^4}{\epsilon^4}\exp\left(-(1/2+ \bar n)(1-\eta^2)\frac{\epsilon^8}{2M^2L^2 (8E+12)^8}\right). 
\end{eqnarray}
This also gives the equivalent condition on $\mathfrak d_\epsilon$ provided in the Lemma.
\end{proof}
\noindent Therefore, efficient estimation of first two quadrature moments of the output state upto arbitrary precision $\epsilon$ is possible, given $L = \mathcal O(\log m)$, $\mathfrak d_{\epsilon^4/(8E+12)^4} < 1$. This gives point (iii) of Theorem \ref{theo:quad_gate_noise} in the main text.



\end{document}